\documentclass[lettersize,journal]{IEEEtran}
\usepackage{amsmath,amsfonts}
\usepackage{algorithmic}
\usepackage{algorithm}
\usepackage{array}
\usepackage{textcomp}
\usepackage{stfloats}
\usepackage{url}
\usepackage{verbatim}
\usepackage{graphicx}
\usepackage{cite}
\usepackage{subcaption}
\usepackage{hyperref}
\newtheorem{definition}{Definition}
\newtheorem{lemma}{Lemma}
\newtheorem{proof}{Proof}
\newtheorem{observation}{Observation}
\newtheorem{corollary}{Corollary}
\DeclareMathOperator*{\argmax}{arg\,max}

\begin{document}

\title{A Participatory Budgeting based Truthful Budget-Limited Incentive Mechanism for Time-Constrained Tasks in Crowdsensing Systems}

\author{Chattu Bhargavi, Vikash Kumar Singh       
\thanks{ School of Computer Science and Engineering, VIT-AP University, Inavolu, Beside AP Secretariat, Amavarati, AP,  India}
}
\markboth{ }%
{Bhargavi \MakeLowercase{\textit{et al.}}}

\IEEEpubid{ }

\maketitle

\begin{abstract}
Crowdsensing, that is also known as crowd sensing or participatory sensing, is a method of data collection that involves gathering information from a large number of common people (or individuals), often using mobile devices, sensors, or other personal technologies. In this paper, the set-up with multiple task requesters and several task executors is considered in \emph{strategic} setting. Each task requester has multiple heterogeneous tasks and the estimated budget for the tasks. In our proposed model, the Government has the publicly known fund (or budget) and is limited. Due to limited Government fund it may not be possible for the platform to offer the funds to all the available task requesters. For that purpose, in the first tier, the voting by the city dwellers over the task requesters is carried out to decide on the subset of task requesters receiving the Government fund. Based on the city dwellers votes, the subset of task requesters is selected such that the total estimated budget of the selected task requesters is at most the Government fund. In the second tier, each task of the task requesters has start and finish times. Based on that, firstly, the tasks are distributed to distinct slots. In each slot, we have multiple task executors for executing the floated tasks. Each of the task executors report a cost (\emph{private}) for completing the floated task(s). Given the above discussed set-up, the objectives of the second tier are: (1) to schedule each task of the task requesters in the available slots, in a non-conflicting manner, and (2) to select a set of executors for the available tasks in such a way that the total incentive given to the task executors should be at most the budget for the tasks. For the discussed scenario, a truthful incentive based mechanism is designed that also takes care of budget criteria. In order to have an estimate about the number of task requesters receiving the Government's fund, the probabilistic analysis is done. Further, the theoretical analysis is done and it shows that the proposed mechanism is \emph{computationally efficient}, \emph{truthful}, \emph{budget feasible}, and \emph{individually rational}. The simulation is carried out and efficacy of the designed mechanism is compared with the state-of-the-art mechanisms. 
\end{abstract}


\section{Introduction}
\label{s:intro}
\IEEEPARstart{T}{he} term ``\emph{crowdsourcing}" was first coined in the year 2006 by Jeff Howe and Mark Robinson \cite{Howe2006}, but the roots of this concept can be found in the past centuries, as the case where the British government wants to have a measurement of the \emph{ship's longitudinal position}, or developing the well known \emph{Oxford  English dictionary} from the very initial stage \cite{Wiki_crowdsourcing}. It is the process of gathering the required data through an open call by utilizing the power of group of common people (or ``\emph{crowd workers}") \cite{NEVO2020101593, Singh_2020, JIN2020103351,fi14020049}. However, if the crowd workers are often using some sensing devices and provide the data by using those devices, then this will add a new dimension to the crowdsourcing and is termed as \emph{mobile crowdsourcing} (or \emph{mobile crowdsensing} or \emph{crowdsensing}) \cite{10443612, 8733838, Phuttharak2019ARO, s20072055,https://doi.org/10.48550/arxiv.2203.06647}.\\
\indent Over the past years, with the increased growth of the communication media and the advent of Web 2.0, it has progressed multi-fold, as the participation of crowd workers and the transfer of information in the crowdsourcing environment has become easy and economical. It is identified as having a wider range of applications such as \emph{environmental and road condition monitoring} \cite{doi:10.1177/2399808320987567, STANIEK2021554, s20195564}, \emph{healthcare} \cite{10356055, Covid-191, 8432319, Giannetsos20111295, DBLP:journals/corr/SinghM16}, \emph{natural or man-made disaster} \cite{Nagatani:2013:ERN:2421033.2421037,Poblet2014, Covid-19} and so on.\\
\indent One of the challenging aspect in mobile crowdsensing (a.k.a crowdsensing) is to have large number of individuals, using mobile devices, sensors, or other personal technologies as the task executors (a.k.a \emph{executors}) into the system in strategic setting\footnote{By strategic it is meant that the agents can gain by misreporting their private information. By private it is meant that it will be only known to him and not known to others.}. It gives rise to the question that: \emph{how to have large group of individuals using mobile devices, sensors, or other personal technologies in such systems}? One plausible solution could be to offer them some incentives (may be some monetary benefits or some social benefits). Over the past, several research in crowdsensing and crowdsourcing provides the incentives to the common individuals in strategic setting \cite{Mukhopadhyay2021, Singh_2020, Singh2019,DBLP:conf/hcomp/GoelNS14}. In \cite{Singh_2020} the proposed mechanism aims to optimize resource allocation and quality assurance in crowdsourcing scenarios, where multiple tasks need to be assigned to participating IoT devices. By utilizing combinatorial auctions, the mechanism efficiently assign tasks while considering the quality of data collected from IoT devices, ensuring reliable and high-quality crowdsourced outcomes. The paper presents experimental results demonstrating the effectiveness of the proposed mechanism in achieving improved resource allocation and quality assurance in IoT-based crowdsourcing.\\
\indent In \cite{Singh2019}, the mechanism utilizes peer evaluation to ensure high-quality contributions while adhering to budget constraints. The results obtained from simulation demonstrate the effectiveness of the proposed mechanism in achieving efficient resource allocation and incentivizing participants in IoT-based crowdsourcing scenarios. \cite{Mukhopadhyay2021} addresses the challenge of task allocation in dynamic environments where budget arrives gradually. The proposed mechanism incentivizes participants to provide truthful reports and optimizes task allocation based on the available budget. Experimental evaluations demonstrate the effectiveness of the mechanism in achieving accurate sensing outcomes while managing budget constraints and adapting to changing budget availability over time. In \cite{DBLP:conf/hcomp/GoelNS14}, an incentive compatible mechanism is proposed for the scenario where, there are a single task requester with multiple tasks and multiple task executors in strategic setting. The task requester hold some fixed budget - an amount that will be utilize as the payment (a.k.a \emph{incentive}) of the task executors that executes the tasks held by the task requester. \\
\indent From the above discussed papers and the works discussed in Section \ref{subsec:mob}, it can be seen that the set of tasks that will be floated for execution purpose is decided apriori. More formally, there exists no works in literature of crowdsourcing (more specifically mobile crowdsensing) system where the common people (or city dwellers) have flexibility to decide on the type of tasks to be floated for execution purpose. Due to this reason the crowdsensing system faces transparency issues and also a group of dissatisfied common people (or city dwellers). Now, the question is \emph{how to design a crowdsensing system with more transparency and more satisfied crowd?} The solution to the above raised query is to allow the common people to decide on the types of tasks to be floated. The city dwellers will vote on the available tasks. Further the votes will be processed and the set of tasks to be floated will be decided. Once the tasks that are to be floated are decided, after that the allocation of tasks executors to the tasks will be carried out and their payments will be decided. The involvement of city dwellers in decision making, makes the proposed crowdsensing framework more transparent and more satisfied crowd.
\begin{figure*}
        \centering
                \includegraphics[scale=0.8]{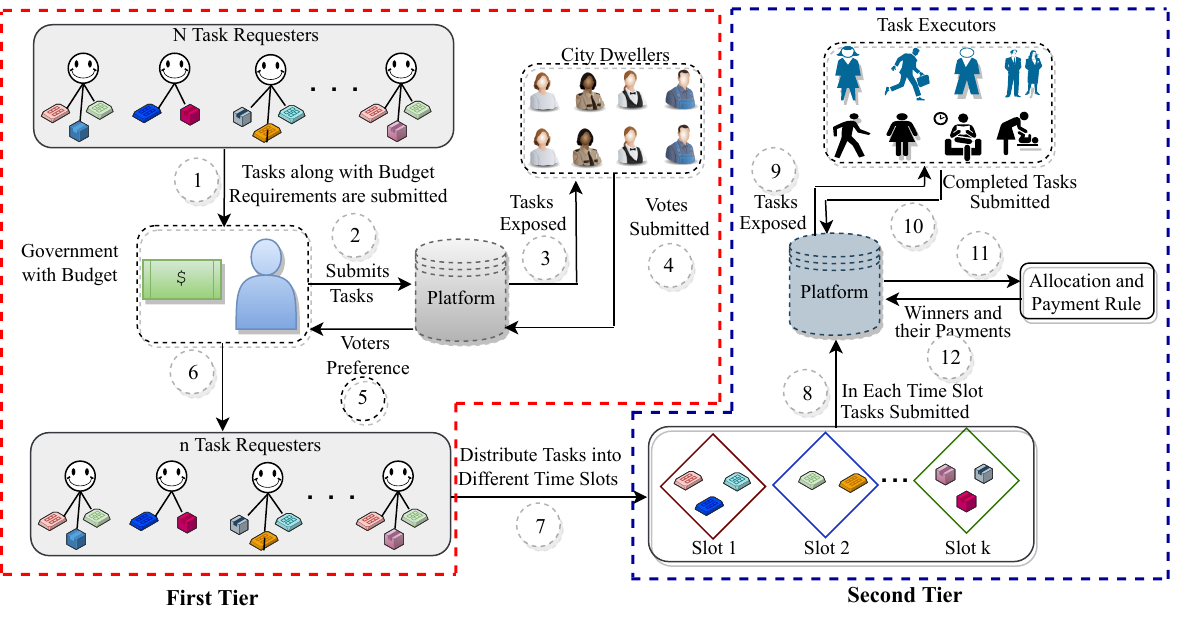}
        \caption{A Two-tiered crowdsensing framework in strategic setting}\label{fig:1aaa}
\end{figure*} 
\indent  In order to implement the above mentioned solution, in this paper, the mobile crowdsensing framework is studied as a two-tiered crowdsensing scenario as shown in FIGURE \ref{fig:1aaa}. In the first tier, there are $\mathcal{N}$ task requesters and each of them is holding multiple heterogeneous tasks that is to be completed by executors. Each task of the task requesters has a start and finish times. Also, each of the task requesters has an estimate on the budget required to get executed the endowed sensing tasks by the task executors. Now, this gives rise to the question that: \emph{who will supply the required budget (or fund) to the task requesters}? One of the plausible solutions could be the fund may be supplied by the \emph{Government} or by \emph{some big organizations} (such as Amazon, Toyota Motor, Reliance, TATA, etc.). However, it could be the case that the \emph{Government} may not have the sufficient funds to serve all the tasks revealed by the task requesters in the market. If that is the case, then the question is: \emph{which task requesters should be getting the funds and why}? In order to tackle this challenge, the \emph{Government} thought that why not to involve the city-dwellers of the city in order to decide that the available fund of the \emph{Government} should be invested for which set of tasks and based on that the subset of tasks will be getting the budget. The city-dwellers give votes on the tasks which they think that are worthwhile in their eyes. The votes of the city-dwellers are collected and processed. Based on the votes of the city dwellers and the available budget of the \emph{Government}, the subset of tasks get qualified for the funds.\\
\indent Once it is finalized that which set of task requesters will be receiving the funds from the Government, our next prime objective is to schedule the tasks into distinct slots in such a way that, all the tasks in any $i^{th}$ slot should be compatible among themselves. The reason behind distributing the tasks to different time slot is to give chance to the task executors to participate into multiple tasks execution process. Now, in each slot we have a budget associated with the available set of tasks. On the opposite side of the crowdsensing market there exists multiple task executors that reports a bid value for executing a particular task. Now in each slot the task executors execute the tasks and give back the executed tasks along with the cost to the platform. For each task, the platform selects the subset of executors from among the available ones in such  a manner that the incentive received by the executors is at most the budget allocated to the respective tasks. Similar procedure is applied in each of the available slots. Designing an incentive mechanism that takes into account the interest of common people on deciding the type of tasks to be considered given limited budget is \emph{non-trivial} and has not been done before. It is \emph{non-trivial} in the sense that the price (or payment) that is offered to the task executors depends on two components: (1) the budget allocated to the tasks of the task requesters, and (2) the bid value (that is a private information) of the task executors. In this article, an incentive  mechanism is designed that ensures the following: (1) distribution of the Government's fund to the most demanding tasks of the task requesters, (2) scheduling the tasks to the distinct slots such that in each slot the tasks are compatible, and (3) allocation of tasks to the task executors and deciding their payments. An incentive based truthful mechanism in this paper is called BUdget-Limited INcentive based mechanism for Crowdsensing system (BULINC). 

\subsection{Our contribution}
The contributions of this paper are:
\begin{enumerate}
\item Firstly, it is decided that among the available tasks which of the tasks will be getting their estimated budget from the fund raisers? For this purpose, the help from city dwellers are taken and they provide their votes over the tasks that they think are worthwhile. Based on the preferences received from the city dwellers, the fund raiser spend their available fund on the subset of tasks from among the available tasks.

\item Secondly, the set of selected tasks are segregated to $|d|$ different slots based on their starting time and finishing time. It is done to ensure that the overlapping tasks are placed in different slots and it will be helpful to the task executors (TEs) in executing multiple tasks during a course of time.

\item Thirdly, from each of the available slots, for each task, the set of executors are considered as the winners among the available executors such that the budget invested in the payment of the executors is within the budget assigned to the tasks. 

\item An incentive compatible mechanism called BULINC is proposed that also take care of the budget constraint criteria. The theoretical analysis shows that BULINC is \emph{truthful} (see Lemma \ref{lem:2}), \emph{computationally efficient} (see Lemma \ref{lem:1}), \emph{individually rational} (see Lemma \ref{lem:3}), and \emph{budget feasible} (see Lemma \ref{lem:3bf}). Also, in order to have an estimate about the number of task requesters receiving the Government's fund, the probabilistic analysis is done (see Lemmas \ref{the1}, \ref{lemma:21}, and \ref{lemma:ATO}). 

\item The simulations are performed for the two tiers independently to measure the performance of BULINC with the state-of-the-art mechanisms.  
\end{enumerate}

\subsection{Paper Structure}
 The rest of the article is organized as follows. In section \ref{sec:rw} the works related to \emph{participatory budgeting}, and \emph{crowdsensing} are presented and elaborated. Section \ref{s:prelim} illustrates the proposed system model and discuss about game-theoretic properties. BULINC is introduced and discussed in section \ref{ref:ehpppm}. In section \ref{sec:atppm} the probabilistic analysis and the theoretical analysis of BULINC is carried out in the order mentioned. The simulations are done to compare BULINC with the state-of-the-art mechanisms in section \ref{sec:ef}. The conclusion with the future directions of the work is presented in section \ref{sec13}.

\section{Related Works}\label{sec:rw}
In this section, the works related to the areas of participatory budgeting, and mobile crowdsensing are discussed one-by-one in the given order.   

\subsection{Participatory Budgeting}
As the first tier of the proposed framework is utilizing the idea of participatory budgeting; therefore the works related to participatory budgeting is discussed in this subsection. In recent years, the practice of participatory budgeting has drag the attention of policy-makers, researchers, and the scholars from around the globe \cite{10.1177/0956247817746279, 10.1177/0020852320943668, 10.1177/0020852321991208}. The participatory democracy a.k.a direct democracy has imbibed lot of attention of late as more and more big cities involving the city-dwellers directly in policy-making \cite{pateman_2012, T.roughgarden_20164, A.Goel}. In fact, there is a huge interest in participatory democracy in general (interested reader can see \cite{pateman_2012,Smith:2009:DID:1795662, 10.1177/00208523221078938} for useful review of the relevant literature). Now a days, all the studies and development in participatory democracy is done under the newly coined terminology called participatory budgeting (PB) \cite{pbp}. In this, the government have some budget with them and they asks the local residents (or city-dwellers) to think of projects/ideas to decide how effectively and efficiently the endowed budget could be utilized in the development of the city. The supplied proposal from the residents side could be \emph{placing lamp post on street}, \emph{renovation of parks}, \emph{building of schools}, or \emph{building a hospital}. \cite{pbp} describes the utilization of PB in \emph{Brazil} and also draws on the experience of 25 municipalities in \emph{Latin America} and \emph{Europe}.\\
\indent In \cite{Aaron}, the author provide the evidence of how the increased participation in the decision making has changed investment patterns in favor of sectors such as \emph{housing},  \emph{sanitation and health}, \emph{education}. In \cite{A.Goel}, they have coined the question: ``\emph{when there are projects with different costs, and fixed budget, how can the varied preferences of voters be best aggregated?''} They have investigated the current voting method in practice $i.e.$ ``$k$-approval voting", underlining its drawback and proposed a novel scheme for this setting called ``\emph{knapsack voting}".

\begin{table*} 
\caption{Summary of selected papers discussed in Subsection \ref{subsec:mob}}
\label{tab: first table}
\centering
\scalebox{0.99}{
\begin{tabular}{c c c c c c c c} 
\hline
\hline
\textbf{Papers} & \textbf{Task requesters} & \textbf{Task executors} & \textbf{Tasks} & \textbf{Budget feasible} & \textbf{Truthful} & \textbf{City dwellers involved} & \textbf{Year}\\
\hline
  \cite{Singh_2020} & Multiple & Multiple & Homogeneous  &\checkmark & \checkmark & $\times$ & 2020\\ 
  \hline
\cite{https://doi.org/10.48550/arxiv.2203.06647} & Multiple & Multiple &  Multiple homogeneous & $\times$ & \checkmark & $\times$ & 2022  \\
\hline
 \cite{Mukhopadhyay2021} & Single & Multiple & Heterogeneous &\checkmark & \checkmark & $\times$ & 2022 \\ 
  \hline
 \cite{Singh2019} & Multiple & Multiple & Heterogeneous &\checkmark & \checkmark & $\times$ & 2020 \\
  \hline
 \cite{8031314} & None & Multiple mobile \\ & & vehicles with devices  & Multiple heterogeneous  &\checkmark & \checkmark & $\times$ & 2017 \\ 
  \hline
   \cite{6848055} & None & Multiple  & Multiple heterogeneous  & $\times$ & \checkmark  & $\times$ & 2014 \\ 
  \hline
  \cite{DBLP:journals/tpds/WangGCG18} & Single & Multiple & Homogeneous &\checkmark & \checkmark & $\times$ & 2018 \\ 
  \hline
  \cite{8667429} & Multiple & Multiple  & Multiple task modules  & $\times$ & $\times$ & $\times$ & 2019\\ 
  \hline
  \cite{9877885} & None & Multiple & Multiple heterogeneous & $\times$ & $\times$ & $\times$ & 2023\\ 
  \hline
  \cite{9877939} & Multiple & Multiple & Multiple homogeneous \\
                        &            &             & and heterogeneous & $\checkmark$ & $\checkmark$ & $\times$ & 2023\\
  \hline
  \cite{9992184} & None & Multiple & Multiple heterogeneous & $\checkmark$ & $\checkmark$ & $\times$ & 2024\\ 
    \hline
    \cite{10.1504/ijwgs.2021.116536} & Multiple & Multiple & Multiple heterogeneous/ \\
                                                    &             &             & homogeneous & $\times$ & $\checkmark$ & $\times$ & 2021\\ 
    \hline
  \textbf{Proposed} & \textbf{Multiple} & \textbf{Multiple} & \textbf{Multiple heterogeneous} & $\boldsymbol{\checkmark}$ & $\boldsymbol{\checkmark}$ & $\boldsymbol{\checkmark}$ & \textbf{-}\\ 
  \hline
  \hline
\end{tabular}
}
\end{table*} 
\subsection{Mobile Crowdsensing}
\label{subsec:mob}
The researchers in the field of crowdsensing (or mobile crowdsensing) can go through the following research articles \cite{SUHAG2023102952, e1, 9701340, 10.1145/3185504, Daniel:2018:QCC:3177787.3148148, 10256160, Phuttharak2019ARO,  8570744, BHATTI2020110611, s20072055, JIN2020103351} to have an idea about the field. In MCS, the major challenge is to drag the group of common people (called \emph{task executors}) for performing the tasks. In order to meet the major challenge of MCS mentioned above, several mechanisms (or algorithms) are designed that offer incentives to the crowd workers in return of executing or completing the assigned tasks \cite{Singh_2020, https://doi.org/10.48550/arxiv.2203.06647, Singh2019, Mukhopadhyay2021, 6848055,  10.1145/3371425.3371459, DBLP:journals/tpds/WangGCG18, 8667429}.\\
\indent In \cite{Singh_2020}, a quality based truthful mechanism is designed for allocating quality task executors to each of the tasks of the task requesters. In this, the task executors asks for bundle of tasks and report cost against the reported bundle of tasks. Here, both the bundle of tasks and the costs are private information. In \cite{https://doi.org/10.48550/arxiv.2203.06647}, an ascending auction based mechanism is designed for allocating the homogeneous tasks held by several task requesters to the quality tasks executors (in this case IoT devices). Once allocated, the task executors complete the allocated tasks and receive payment in return for their completed tasks. One of the constraints that is placed on the valuation of task executors and task requesters is that the valuations follow \emph{decreasing marginal return} policy. In \cite{Mukhopadhyay2021} the proposed framework consists of a single task requester having single task and a budget, and several task executors. The full budget that is required by the task requester is not available but some part of the budget is available in the crowdsensing market. The remaining budget of task requester will be brought into the system in multiple cycles. For this scenario, a truthful mechanism is designed that also take care that the total incentives given to the task executors is within the full budget of the task requester for the tasks. In \cite{Singh2019}, an incentive compatible mechanism is proposed for hiring a subset of quality tasks executors (in this case the IoT devices) for each task. The task executors are selected in such a manner that the total incentive provided to the task executors is within or equal to the budget allotted to each task.  In \cite{8031314}, two different kinds of incentive mechanisms are proposed in IoT-based crowdsensing for traffic surveillance or environmental pollution monitoring systems. In \cite{6848055} the sensory tasks along with the location of the tasks are reported by the platform. In this set-up, the smartphone users are the task executors and execute the tasks that lie in the coverage area of respective smartphone users. For executing the tasks in their coverage the task executors charges some cost and is private entity. For the discussed crowdsensing scenario, an allocation and payment rules are determined for allocating the sensory tasks to the task executors and for deciding the payment of the task executors.\\
\indent In \cite{DBLP:journals/tpds/WangGCG18} a truthful mechanism is designed that selects the crowd workers that maintain some threshold quality while doing the assigned tasks. In this the quality of task executors are variable $i.e.$ more the task executors are executing the tasks, the more they are becoming experienced and hence the quality is increasing.  In \cite{10.1145/3371425.3371459} an experience based truthful mechanism is proposed that utilizes the idea reverse auction for selecting the crowd worker and to decide their payment. It is also proved that the proposed mechanism have four properties: \emph{truthfulness}, \emph{individual rationality}, \emph{computational efficiency}, and \emph{profitability}. In \cite{9877885} a heuristic mechanism is designed for incentivizing the task executors in the social MCS system. The problem in MCS is solved through the two stage mechanism that considers (1) influencing other task executors through social connection, and (2) the quality of the work delivered. In \cite{9877939} the set-up consists of multiple service requesters having private budgets and multiple users carrying mobile devices in strategic setting. The efficient, fair, and budget feasible mechanism is designed for both homogeneous tasks and heterogeneous tasks setting. The proposed mechanism satisfies several economic properties such as \emph{truthfulness}, \emph{individual rationality}, \emph{budget feasibility}, and \emph{fairness}. \cite{9992184} studies bi-objective optimization case of mobile crowdsensing that optimizes (1) total value function, and (2) coverage function that have budget/cost constraint. 
In \cite{8667429} the focus is on performing a specific software development task by utilizing the concept of crowdsourcing. It utilizes the active time of crowd workers to group them together and then allocate the respective group of crowd workers to the multiple tasks. In \cite{10.1504/ijwgs.2021.116536} a double auction based truthful incentive mechanism is designed for the time constrained tasks in mobile crowdsensing. As the tasks are having starting and finishing times and may overlap in time. So, first of all the tasks are distributed to different slots in a non-overlapping manner. Once distributed, an egalitarian approach is utilized to allocate the tasks in a balanced manner to the task executors and the payment made to the task executors is decided. \\
\indent From the past works in crowd sensing discussed above it can be inferred that there is no involvement of city dwellers (or common people) to decide on the subset of tasks to be floated in the market for execution purpose and is carried out in this paper.  
\section{Notation and Preliminaries}\label{s:prelim}
In this section, the crowdsensing scenario discussed in Section \ref{s:intro} is modeled using mechanism design. There are $\mathcal{N}$ task requesters and $\mathcal{M}$ task executors. The set $\boldsymbol{r} = \{\boldsymbol{r}_1, \boldsymbol{r}_2, \ldots, \boldsymbol{r}_\mathcal{N}\}$ represents the set of task requesters and the set $\boldsymbol{e} = \{\boldsymbol{e}_1, \boldsymbol{e}_2, \ldots, \boldsymbol{e}_\mathcal{M}\}$ represents the set of task executors. Each task requester $\boldsymbol{r}_i$ is endowed with a set of heterogeneous tasks and is given as $\boldsymbol{t}^i = \{\boldsymbol{t}_1^i, \boldsymbol{t}_2^i, \ldots, \boldsymbol{t}_{n_i}^i\}$. Here, $\boldsymbol{t}_{k}^i$ represents $k^{th}$ task of $i^{th}$ task requester. Each of the task requesters $\boldsymbol{r}_i$ has an estimate on the cost (or budget), that will be required for the successful completion of the tasks held by him/her (henceforth him) and is represented by $\mathcal{B}_i$. Each task $\boldsymbol{t}_j^i$ have a starting time and a finishing time. The starting and finishing times of a task $\boldsymbol{t}_j^i$ is given as $s_j^i$ and $f_j^i$ respectively. Here, $f_j^i \geq s_j^i$. The set of starting time and the set of finishing time of the tasks of the task requester $\boldsymbol{r}_i$ is given as $s^i = \{s_1^i, s_2^i, \ldots, s_{n_i}^i\}$ and $f^i = \{f_1^i, f_2^i, \ldots, f_{n_i}^i\}$ respectively. The starting time vector of the tasks of all the task requesters is given as $s = \{s^1, s^2, \ldots, s^{\mathcal{N}}\}$ and the finishing time vector of the tasks of all the task requesters is given as $f = \{f^1, f^2, \ldots, f^{\mathcal{N}}\}$. In this paper, it is considered that, initially the estimated budget $\mathcal{B}_i$ is not available to any of the task requesters $\boldsymbol{r}_i$ and is made available by the fund raisers (may be \emph{Government} or \emph{some big organization}). As already discussed the set of city-dwellers is given as $\mathbb{C} = \{\mathbb{C}_1, \mathbb{C}_2, \ldots, \mathbb{C}_p\}$ will provide their votes (or preferences) on the set of task requesters whose tasks they have identified to be worthwhile. If say the tasks of task requester $\boldsymbol{r}_j$ is preferred over the tasks of task requester $\boldsymbol{r}_k$ by the city dweller $\mathbb{C}_i$ then it is represented as $\boldsymbol{r}_j \succ_i \boldsymbol{r}_k$. The set of preference ordering of all the city dwellers is given as $\succ = \{\succ_1, \succ_2, \ldots, \succ_p\}$. The gain of city-dweller $i \in \mathbb{C}$ from the tasks of $\boldsymbol{r}_j$ is denoted by $\boldsymbol{\mu}_{i,j}$. Ideally, the goal is to select the subset of task requesters among the available task requesters whose tasks will yield the maximum gain for the city dwellers along with the constraint that the sum of the budgets associated by the selected task requesters should be at most  the Government's fund $i.e.$ $\mathcal{B}$. Mathematically,
\begin{equation}
\argmax\limits_{\boldsymbol{\mathbb{F}} \in \boldsymbol{r}} \sum\limits_{j \in \boldsymbol{\mathbb{F}}} \bigg( \frac{1}{|\mathbb{C}|} \sum\limits_{i \in \mathbb{C}} \boldsymbol{\mu}_{i,j}\bigg)~~~subject~to~~
\sum\limits_{j\in \boldsymbol{\mathbb{F}}} \mathcal{B}_j \leq \mathcal{B}
\end{equation} 
\begin{table*}[t]
\caption{Notations used}
\label{tab: first table}
\centering
\begin{tabular}{c|c}
\hline
\textbf{Symbols} & \textbf{Descriptions}\\
\hline
$\boldsymbol{r}$ &  $\boldsymbol{r} = \{\boldsymbol{r}_1, \boldsymbol{r}_2, \ldots, \boldsymbol{r}_\mathcal{N}\}$: Set of task requesters.\\
$\boldsymbol{e}$ & $\boldsymbol{e} = \{\boldsymbol{e}_1, \boldsymbol{e}_2, \ldots, \boldsymbol{e}_\mathcal{M}\}$: Set of task executors.\\
 $\boldsymbol{r}_i$ & $i^{th}$ task requester.\\
$\boldsymbol{e}^\ell$ & $\boldsymbol{e}^\ell = \{\boldsymbol{e}_1, \boldsymbol{e}_2, \ldots, \boldsymbol{e}_{m_\ell}\}$: Set of task executors in any $\ell^{th}$ slot.\\ 
$\boldsymbol{e}_i$ & $i^{th}$ task executor.\\
$\boldsymbol{t}^i$ & $\boldsymbol{t}^i$ = $\{\boldsymbol{t}_1^i, \boldsymbol{t}_2^i, \ldots, \boldsymbol{t}_{n_i}^i\}$: Set of heterogeneous tasks.\\
$\boldsymbol{t}_{k}^i$ & $k^{th}$ task of $i^{th}$ task requester.\\
$\mathcal{B}_i$ & Estimated budget of task requester $\boldsymbol{r}_i$.\\
$\mathcal{B}$ & Available budget.\\
$s^i$ & $s^i = \{s_1^i, s_2^i, \ldots, s_{n_i}^i\}$: Set of starting times for the tasks of $\boldsymbol{r}_i$ task requester.\\
 $f^i$ & $f^i = \{f_1^i, f_2^i, \ldots, f_{n_i}^i\}$: Set of finishing times for the tasks of $\boldsymbol{r}_i$ task requester.\\
$s$ & $s = \{s^1, s^2, \ldots, s^{\mathcal{N}}\}$: Starting time vectors for the tasks of all the task requesters.\\
$f$ & $f = \{f^1, f^2, \ldots, f^{\mathcal{N}}\}$: Finishing time vectors for the tasks of all the task requesters.\\
$\mathbb{C}$ & $\mathbb{C} = \{\mathbb{C}_1, \mathbb{C}_2, \ldots, \mathbb{C}_p\}$: Set of city-dwellers.\\
$\succ$ & $\succ = \{\succ_1, \succ_2, \ldots, \succ_p\}$: Set of preference ordering of all the city dwellers.\\
$\boldsymbol{\mu}_{i,j}$ & The gain of city-dweller $i \in \mathbb{C}$ from the tasks of $\boldsymbol{r}_j$.\\
 $\mathbb{F}$ & $\mathbb{F} = \{\boldsymbol{r}_1, \boldsymbol{r}_2, \ldots, \boldsymbol{r}_n\}$: Set of task requesters belongs to winning set.\\ 
$v(\mathcal{A}_{ij})$ &  Publicly known valuation for subset of task executors $\mathcal{A}_{ij} \subseteq \boldsymbol{e}^\ell$.\\
 $c_i$ &  Cost of $i^{th}$ task executor.\\
  $c$ & $c = \{c_1, c_2, \ldots, c_\mathcal{M}\}$: Cost vector of the task executors.\\
 $u_{i}$ & Utility of $i^{th}$ task executor.\\ 
 $p_i$ & Payment of $i^{th}$  task executor.\\ 
\hline
\end{tabular}
\end{table*}                   

Here,  $\mathbb{F}$ is the set of task requesters belongs to winning set and is given as $\mathbb{F} = \{\boldsymbol{r}_1, \boldsymbol{r}_2, \ldots, \boldsymbol{r}_n\}$, $\boldsymbol{\mu}_{i,j}$ = $[Government~ budget~ spent~ on~ tasks~ of~ \boldsymbol{r}_j]$. The discussed optimization problem reduces to the \emph{Knapsack} problem \cite{A.Goel} with $\boldsymbol{r}$ is the set of task requesters to be fitted into a knapsack of capacity $\mathcal{B}$.  For any task requester $\boldsymbol{r}_j$, the average utility of the city dwellers is given as $\bigg( \frac{1}{|\mathbb{C}|} \sum\limits_{i \in \mathbb{C}} \boldsymbol{\mu}_{i,j}\bigg)$. In this it is assumed that for any task requester $\boldsymbol{r}_j$  some subset of tasks may get the fund and not to all. It means fractional funding is allowed in our proposed framework.  So, the output of the first tier is the set of task requesters whose tasks will be funded by the fund raisers (may be Government).\\
\indent Once the task requesters receive the funds, in the second tier following the crowdsourcing framework, each task requester $\boldsymbol{r}_i$ submits a tuple $\big<\mathcal{B}_i,s^i,f^i,\boldsymbol{t}^i\big>$ to the platform. On receiving the tuples from the participating task requesters, firstly, for each task requester, the set of compatible and incompatible tasks are determined. The two tasks $\boldsymbol{t}_j^i$ and $\boldsymbol{t}_k^i$ are said to be incompatible if one of the following occurs: (1) $s_j^i \leq s_k^i \leq f_k^i \leq f_j^i$ or (2) $s_j^i \leq s_k^i \leq f_j^i \leq f_k^i$ or (3) $s_k^i \leq s_j^i \leq f_k^i \leq f_j^i$ or (4) $s_k^i \leq s_j^i \leq f_j^i \leq f_k^i$. In order to make the tasks compatible with each other they are distributed into $|d|$ time slots, such that in any particular time slot the available tasks are compatible. It will allow the task executors to execute multiple tasks present in different time slots. Once the tasks of task requester $\boldsymbol{r}_i$ are populated to $d_i$ different time slots, such that $d_i \ll d$, the overall available budget $\mathcal{B}_i$ of the task requester $\boldsymbol{r}_i$ is distributed equally to the tasks of $\boldsymbol{r}_i$ task requester in $d_i$ slots and is given as $\big\lfloor \frac{\mathcal{B}_i}{n_i} \big\rfloor$. In any $\ell^{th}$ slot, say, there are $m_\ell$ number of task executors and is given as $\boldsymbol{e}^\ell = \{\boldsymbol{e}_1, \boldsymbol{e}_2, \ldots, \boldsymbol{e}_{m_\ell}\}$. Each task executor $\boldsymbol{e}_i$ has a privately held cost $c_i$ (cost he will charge for executing the subset of tasks). In any $\ell^{th}$ slot, for the task $\boldsymbol{t}_j^i$ of $\boldsymbol{r}_i$ task requester, there exist a publicly known valuation $v(\mathcal{A}_{ij})$  for subset of task executors $\mathcal{A}_{ij} \subseteq \boldsymbol{e}^\ell$. It represents the social welfare acquired from $\mathcal{A}_{ij}$. It is assumed that $v(\phi) = 0$ and $v(\mathcal{A}_{ij}) \leq v(\mathcal{F}_{ij})$ for any $\mathcal{A}_{ij} \subseteq \mathcal{F}_{ij} \subseteq \boldsymbol{e}^\ell$. In this set-up it is assumed that the valuation function is submodular. By submodular function it is meant that:
\begin{equation*}
v(\mathcal{A}_{ij}) + v(\mathcal{F}_{ij}) \geq v(\mathcal{A}_{ij} \cap \mathcal{F}_{ij}) + v(\mathcal{A}_{ij} \cup \mathcal{F}_{ij})
\end{equation*}
for any $\mathcal{A}_{ij}, \mathcal{F}_{ij} \subseteq \boldsymbol{e}^\ell$. As the participating task executors are \emph{strategic}, so any $i^{th}$ task executor may report his private cost as $\hat{c}_i$ instead of $c_i$ such that $\hat{c}_i \neq c_i$ so as to maximize his utility. The cost vector of the task executors is given as $c = \{c_1, c_2, \ldots, c_\mathcal{M}\}$. The utility of any task executor $\boldsymbol{e}_i$ is given as
\begin{equation}
\label{equ:1}
u_{i} =
\begin{cases}
p_i - c_i, & \text{if $\boldsymbol{e}_i$ wins} \\
0, & \text{otherwise}
\end{cases}
\end{equation}
where, $p_i$ is the payment made to $\boldsymbol{e}_i$ task executor. For the scenario of second tier, in each slot $j$, for each task $\boldsymbol{t}_j^i$ of any task requester $\boldsymbol{r}_i$, the prime goal is to have a bunch of task executors that will help in maximizing the social welfare $i.e$ $v(\mathcal{A}_{ij})$ with the constraint that $\sum\limits_{h \in \mathcal{A}_{ij}} c_h \leq \big\lfloor \frac{\mathcal{B}_i}{n_i} \big\rfloor$. For all the tasks of the $i^{th}$ task requester we have,  
\begin{equation}
\max\limits_{\mathcal{A}_{ij} \subseteq \boldsymbol{e}^\ell}\sum_{j = 1}^{n_i} v(\mathcal{A}_{ij})~~~subject~to~~
\sum_{j = 1}^{n_i}\sum\limits_{h \in \mathcal{A}_{ij}} c_h \leq \bigg\lfloor \frac{\mathcal{B}_i}{n_i} \bigg\rfloor
\end{equation}  
\subsection{Desirable Properties}
In this paper, BULINC satisfies the following properties: 
\begin{definition}[\textbf{Truthfulness (or incentive compatibility) \cite{NNis_Pre_2007}}]
\label{def:1}
BULINC is \emph{truthful} (or incentive compatible), if the participating task executors get maximum utility by reporting their private information in truthful manner. More formally, for any task executor $\boldsymbol{e}_i$ we have $u_{i} = p_i - c_i \geq \hat{u}_i = p_i - \hat{c}_i$, where $\hat{c}_i \neq c_i$.  
\end{definition}

\begin{definition}[\textbf{Individual rationality \cite{NNis_Pre_2007}}]
\label{def:3}
BULINC is individual rational, if the participating task executors are having zero or some positive utility. In other words, $u_i = p_i - c_i \geq 0$ for any $i^{th}$ task executor.   
\end{definition}

\begin{definition}[\textbf{Budget feasibility \cite{NNis_Pre_2007}}]
\label{def:2}
BULINC is budget feasible, if the payment received by the winning task executors is at most the available budget. More formally, $\sum\limits_{h \in \mathcal{A}_{ij}} c_h \leq \big\lfloor \frac{\mathcal{B}_i}{n_i} \big\rfloor$.
\end{definition}

\begin{definition}[\textbf{Computational efficiency}]
\label{def:4}
BULINC is computationally efficient, if the budget distribution by the fund raiser, allocation of jobs in a non-overlapping manner, and the allocation and payment determination of the task executors can be done in polynomial time.
\end{definition} 

\section{BUdget-Limited INcentive based mechanism for Crowdsensing system (BULINC)}
\label{ref:ehpppm}
In this section, BULINC is discussed and presented. BULINC consists of: (1) Budget distribution mechanism, (2) Compatible tasks distribution mechanism, and (3) Allocation and pricing rule.
\subsection{\textsc{Budget Distribution Mechanism}}
\label{sub:bdm}
It helps to decide on \emph{which of the task requesters should be allocated the Government fund}? The idea of \emph{Budget Distribution Mechanism} is in parallel to knapsack voting \cite{T.roughgarden_20164, A.Goel} and is presented in Algorithm \ref{algo:10}. The inputs to Algorithm \ref{algo:10} are the task requesters, the city dwellers, and budget. In line 1 of Algorithm \ref{algo:10} the set that will contain the winning task requesters $\mathbb{F}$ is set to $\phi$. Lines 2-5 generate the preference ordering (or votes) of the city dwellers over the task requesters. The preference ordering of each of the city dwellers $\mathbb{C}_i$ is held in $\succ_i$ in line 3. In line 4 the set $\succ$ holds the vote of all the city dwellers over the task requesters. Line 6 sorts the task requesters in ascending order based on the votes received from the city dwellers. From the ordering obtained after sorting, every time $\boldsymbol{r}_i$  task requester is taken and checked whether the budget required by any task requester $\boldsymbol{r}_i$ is less than the available Government's budget or not.
\begin {algorithm}[H]
\caption{\textsc{Budget Distribution Mechanism} ($\boldsymbol{r}$, $\mathbb{C}$, $\mathcal{B}$)}
\label{algo:10}
\noindent
\begin{algorithmic}[1]
\STATE $\mathbb{F} \gets \phi$, $\succ \gets \phi$
\FOR{each $\mathbb{C}_i \in \mathbb{C}$}
\STATE $\succ_i$ $\leftarrow$ GeneratePref ($\boldsymbol{r}$) \COMMENT{GeneratePref ($\boldsymbol{r}$) generates the preference of each city dweller $\mathbb{C}_i$ over subset of task requesters.}
\STATE $\succ$ $\leftarrow$ $\succ$ $\cup$ $\succ_i$
\ENDFOR
\STATE $\hat{\boldsymbol{r}}$ $\gets$ Sort ($\boldsymbol{r}$, $\succ$) \COMMENT{Sort task requesters in ascending order of votes received from city dwellers.}
\FOR{each $\boldsymbol{r}_i \in \hat{\boldsymbol{r}}$}
\IF{$\mathcal{B}_i$ $\leq$ $\mathcal{B}$}
\STATE $\mathbb{F} \leftarrow \mathbb{F} \cup \{\boldsymbol{r}_i\}$ \COMMENT{$\mathbb{F}$ holds the set of task requesters that receives Government fund.}
\STATE $\mathcal{B}\leftarrow \mathcal{B} - \mathcal{B}_i$
\ENDIF
\ENDFOR
\STATE return $\mathbb{F}$
\end{algorithmic}
\end{algorithm}
If the budget criteria mentioned in line 8 of Algorithm \ref{algo:10} is satisfied then in that case the task requester $\boldsymbol{r}_i$ is included in the winning set in line 9. In line 10 the budget $\mathcal{B}_i$ allocated to task requester $\boldsymbol{r}_i$ is deducted from the overall available budget. The process in lines 7-12 iterates till all the task requesters present in the ordering are accessed. Finally in line 13 the winning set $\mathbb{F}$ is returned as the selected set of task requesters.

 \subsubsection{Illustration of Budget Distribution Mechanism with an Example} 
In this subsection, let us consider an example to understand Algorithm \ref{algo:10}. Suppose, we have a set of 10 city dwellers    $\mathbb{C} = \{\mathbb{C}_1, \mathbb{C}_2, \ldots, \mathbb{C}_{10}\}$ and set of 5 task requesters  $\boldsymbol{r} = \{\boldsymbol{r}_1, \boldsymbol{r}_2, \ldots, \boldsymbol{r}_5\}$ along with set of heterogeneous tasks  $\boldsymbol{t}$ = $\{ \boldsymbol{t}^1,  \boldsymbol{t}^2, \ldots,\boldsymbol{t}^5\}$. The Government's budget is taken as $\mathcal{B}$ = 100\$. Let us say the budget requirement for $\boldsymbol{r}_1$ is $\mathcal{B}_1$ =  10\$, for $\boldsymbol{r}_2$ is $\mathcal{B}_2$ = 20\$, for $\boldsymbol{r}_3$ is $\mathcal{B}_3$ = 30\$, for $\boldsymbol{r}_4$ is $\mathcal{B}_4$ = 40\$, for $\boldsymbol{r}_5$ is $\mathcal{B}_5$ = 50\$. Now, the set of city dwellers that are present in  $\mathbb{C}$ are giving the votes over the subset of task requesters present in set $\boldsymbol{r}$ such that the sum of the budget requirements of the tasks requesters in the preference ordering of city dwellers is within the Government budget $\mathcal{B} = 100\$$.  

Say, a city dweller $\mathbb{C}_1$ gives preference as $\boldsymbol{r}_4\succ_1 \boldsymbol{r}_5 \succ_1\boldsymbol{r}_1$, $\mathbb{C}_2$ gives preference as $\boldsymbol{r}_5\succ_2 \boldsymbol{r}_2\succ_2 \boldsymbol{r}_3$, $\mathbb{C}_3$ gives preference as $\boldsymbol{r}_2\succ_3\boldsymbol{r}_5\succ_3\boldsymbol{r}_1$, $\mathbb{C}_4$ gives preference as $\boldsymbol{r}_3\succ_4\boldsymbol{r}_2\succ_4\boldsymbol{r}_1\succ_4\boldsymbol{r}_4$, $\mathbb{C}_5$ gives preference as $\boldsymbol{r}_1\succ_5\boldsymbol{r}_3\succ_5\boldsymbol{r}_5$, $\mathbb{C}_6$ gives preference as $\boldsymbol{r}_5\succ_6\boldsymbol{r}_2\succ_6\boldsymbol{r}_3$, $\mathbb{C}_7$ gives preference as $\boldsymbol{r}_5\succ_7\boldsymbol{r}_4\succ_7\boldsymbol{r}_1$ and similarly $\mathbb{C}_8$, $\mathbb{C}_9$, $\mathbb{C}_{10}$  give preferences as $\boldsymbol{r}_5\succ_8\boldsymbol{r}_3\succ_8\boldsymbol{r}_2$, $\boldsymbol{r}_3\succ_9\boldsymbol{r}_5\succ_9\boldsymbol{r}_2$, and $\boldsymbol{r}_2\succ_{10}\boldsymbol{r}_5\succ_{10} \boldsymbol{r}_3$ respectively. Following line 6 of Algorithm \ref{algo:10} the task requesters are sorted in descending order of the number of votes received (the task requesters appearing at the first position in the preference ordering of the city dwellers.) i.e. $\boldsymbol{r}_5$= 4 votes, $\boldsymbol{r}_2$= 2 votes, $\boldsymbol{r}_3$= 2 votes, $\boldsymbol{r}_4$= 1 vote, $\boldsymbol{r}_1$= 1 vote. From the sorted ordering, firstly the task requester $\boldsymbol{r}_5$ is picked up and a check is made that $50\$ \leq 100\$$, as the condition is true so $\boldsymbol{r}_5$ is placed in $\mathbb{F}$ $i.e.$ $\mathbb{F}$ = $\{\boldsymbol{r}_5\}$. Now the remaining Government's budget is $50\$$. In the next iteration, $\boldsymbol{r}_2$ is picked up from the sorted ordering and a check is made that $20\$ \leq 50\$$, as the condition is true so $\boldsymbol{r}_2$ is placed in $\mathbb{F}$ $i.e.$ $\mathbb{F}$ = $\{\boldsymbol{r}_5, \boldsymbol{r}_2\}$. Now, the remaining Government budget $\mathcal{B} = 30\$$. In similar manner in the next iteration $\boldsymbol{r}_3$ is picked up and placed in $\mathbb{F}$ $i.e.$ $\mathbb{F} = \{\boldsymbol{r}_5, \boldsymbol{r}_2, \boldsymbol{r}_3\}$. The remaining Government budget $\mathcal{B} = 0\$$. Hence, the for loop in lines 7-12 of Algorithm \ref{algo:10} terminates. Using line 13 of Algorithm \ref{algo:10}, $\mathbb{F} = \{\boldsymbol{r}_5, \boldsymbol{r}_2, \boldsymbol{r}_3\}$ is returned. 

\begin{figure*}%
    \centering
    \subfloat[\centering Task requesters along with tasks and their start and finish times]{{\includegraphics[width=8.0cm]{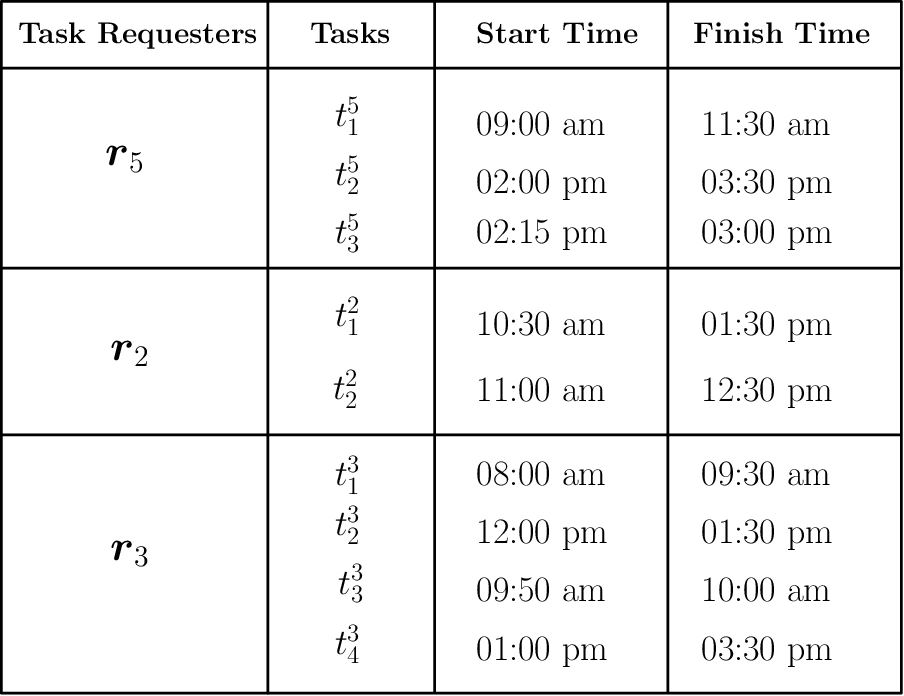} }}%
    \qquad
    \subfloat[\centering Tasks sorted in increasing order of start time and distribution of tasks in compatible slots]{{\includegraphics[width=8.0cm]{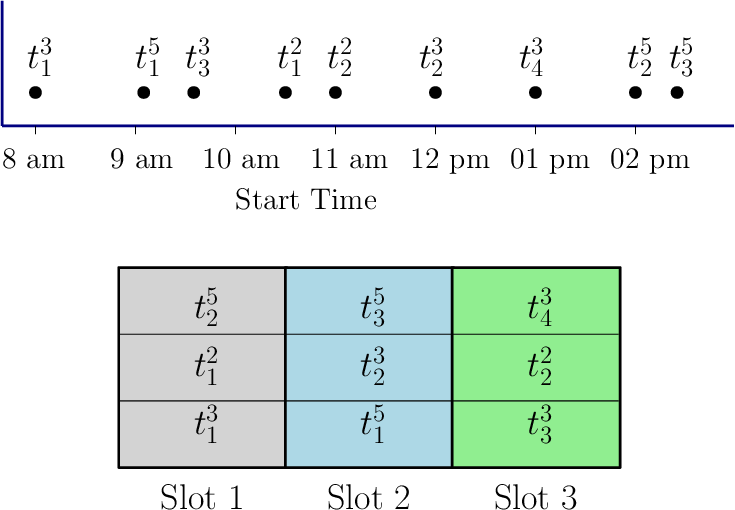} }}%
    \caption{Illustration of compatible tasks distribution mechanism}
    \label{fig:2a}
\end{figure*}

\subsection{\textsc{Compatible Tasks Distribution Mechanism}}
\label{sub:ctdm}
In this section, motivated by \emph{interval partitioning} algorithm \cite{Kleinberg:2005:AD:1051910} a mechanism is proposed. 
\begin {algorithm}[H]
\caption{\textsc{Compatible Tasks Distribution Mechanism} ($\mathbb{F}$, $s$)}
\label{algo:11}
\noindent
\begin{algorithmic}[1]
\FOR {each $\boldsymbol{r}_i \in \mathbb{F}$}
\STATE $\boldsymbol{t'}$ $\leftarrow$ $\boldsymbol{t'}$ $\cup$ $\boldsymbol{t}^i$ \COMMENT{Tasks of each task requester are held in $\boldsymbol{t'}$.}
\ENDFOR
\STATE $\hat{\boldsymbol{t}}$ $\gets$ Sort ($\boldsymbol{t'}$, $s$) \COMMENT{Sort tasks in ascending order of start time.}
\STATE $\ell \gets 1$ \COMMENT{Initialized to single slot.}
\FOR{each $\boldsymbol{t}_j^i \in \hat{\boldsymbol{t}}$}
\IF{$\boldsymbol{t}_j^i$ is compatible with all the tasks in the available slots $\ell$}
\STATE Place $\boldsymbol{t}_j^i$ in $\ell$ 
\ELSE
\STATE $\ell \gets\ell+1$ \COMMENT{Number of slots incremented by 1.}
\STATE Place task $\boldsymbol{t}_j^i$ to $\ell^{th}$ slot  
\ENDIF
\ENDFOR
\STATE return $\ell$
\end{algorithmic}
\end{algorithm} 
It will take care of distribution of tasks to some fixed number of time slots such that no two incompatible tasks will be holding the same slot. The detailing of \emph{compatible tasks distribution mechanism} is presented in Algorithm \ref{algo:11}. The inputs to Algorithm \ref{algo:11} are the task requesters captured in the winning set by Algorithm \ref{algo:10}, and the start time vectors of the tasks of all the task requesters. In lines 1-3 for loop iterates over the set $\mathbb{F}$ and the set of tasks of each task requester is held in $\boldsymbol{t}'$. In line 4 of Algorithm \ref{algo:11}, the tasks are ordered in  increasing order in terms of their start time. The variable $\ell$ keeping track of slot number and is initialized to 1. The for loop in lines 6-13 of Algorithm \ref{algo:11} the tasks in set $\hat{\boldsymbol{t}}$ is scheduled in the most suitable slot among the available slots. Finally, in line 14 a slot $\ell$ is returned.
 \subsubsection{Illustration of Compatible Tasks Distribution Mechanism with an Example}
The presented example illustrates the working of Algorithm \ref{algo:11}. The input to Algorithm \ref{algo:11} is the set of task requesters $\mathbb{F}$ = \{$\boldsymbol{r}_5$, $\boldsymbol{r}_2$, $\boldsymbol{r}_3$\} and the start time vectors of the tasks of the task requesters $\boldsymbol{r}_5$, $\boldsymbol{r}_2$, $\boldsymbol{r}_3$ depicted in FIGURE \ref{fig:2a}a. Here, task requester $\boldsymbol{r}_5$ has $\boldsymbol{t}_1^5, \boldsymbol{t}_2^5, \boldsymbol{t}_3^5$ tasks. Similarly the other two task requesters i.e. $\boldsymbol{r}_2$ and $\boldsymbol{r}_3$ are having the set of tasks  $\boldsymbol{t}_1^2, \boldsymbol{t}_2^2$, and $\boldsymbol{t}_1^3, \boldsymbol{t}_2^3, \boldsymbol{t}_3^3, \boldsymbol{t}_4^3$ respectively. Following lines 1-3 of Algorithm \ref{algo:11} in our running example we get  $\boldsymbol{t'}$ = \{$\boldsymbol{t}_1^5, \boldsymbol{t}_2^5, \boldsymbol{t}_3^5$, $\boldsymbol{t}_1^2, \boldsymbol{t}_2^2$, $\boldsymbol{t}_1^3, \boldsymbol{t}_2^3, \boldsymbol{t}_3^3, \boldsymbol{t}_4^3$\}. Following line 4 of Algorithm \ref{algo:11} the tasks in $\boldsymbol{t}'$ are sorted in increasing ordering of starting time, so we get $\hat{\boldsymbol{t}}$ = \{$\boldsymbol{t}_1^3, \boldsymbol{t}_1^5,  \boldsymbol{t}_3^3,  \boldsymbol{t}_1^2, \boldsymbol{t}_2^2, \boldsymbol{t}_2^3,  \boldsymbol{t}_4^3, \boldsymbol{t}_2^5,  \boldsymbol{t}_3^5$\}. Following lines 6-13 of Algorithm \ref{algo:11} the task $\boldsymbol{t}_1^3$ is selected from the ordering in set $\hat{\boldsymbol{t}}$ and is placed in first slot $i.e.$ $\ell = 1$. After that from the ordering task $\boldsymbol{t}_1^5$ is picked up from the sorted ordering and as task $\boldsymbol{t}_1^5$ is not compatible with task $\boldsymbol{t}_1^3$ present in slot 1, so task $\boldsymbol{t}_1^5$ is placed in slot 2 $i.e.$ $\ell = 2$. Next, task $\boldsymbol{t}_3^3$ is picked up and as task $\boldsymbol{t}_3^3$ is not compatible with $\boldsymbol{t}_1^3$ and $\boldsymbol{t}_1^5$, so it is placed in slot 3 $i.e.$ $\ell = 3$. After that task $\boldsymbol{t}_1^2$ is  selected from the ordering in set $\hat{\boldsymbol{t}}$ and is placed in first slot. Next, tasks $\boldsymbol{t}_2^2$ and $\boldsymbol{t}_2^3$ are compatible with the tasks in slot 3 and slot 2 respectively. So, tasks $\boldsymbol{t}_2^2$ and $\boldsymbol{t}_2^3$ are placed in slot 3 and slot 2 respectively. The task $\boldsymbol{t}_4^3$ is compatible with the tasks in slot 3, so it is placed in slot 3. Finally, tasks $\boldsymbol{t}_2^5$ and $\boldsymbol{t}_3^5$ are placed in slots 2 and 1 respectively. The tasks distribution to various slots is represented in FIGURE \ref{fig:2a}b.

\subsection{\textsc{Allocation and Pricing Rule}}\label{sub:apr}
The subroutine of BULINC for allocating tasks to the task executors and their payment is motivated by \cite{Singer:2014:BFM:2692359.2692366}. The inputs to Algorithm \ref{algo:1} are the task requesters in winning set returned by Algorithm \ref{algo:10}, the task executors, the number of slots returned by Algorithm \ref{algo:11}, costs the task executors, and budget. In line 1, the variable $k$ is initialized to 1, and set $\mathcal{W}$ and $p$ are initialized to $\phi$. Using lines 2-28, in each of the $\ell^{th}$ slot, for each task of the task requesters the task executors are allocated and their payment is decided. For each task $\boldsymbol{t}_j^i \in \boldsymbol{t}^i$, in line 5 the available task executors are held in $\hat{\boldsymbol{e}}$ once the task executors are ordered in ascending order of bid value.\\
\indent In lines 6-13 of Algorithm \ref{algo:11}, each time the task executor $\boldsymbol{e}_f$ is considered as a winner only when the selection criteria mentioned in line 7 is satisfied. After that the number of selected task executors $i.e.$ $k$ is increased. In the set $\mathcal{W}_j^i$ the winning task executors are pushed in for each task $\boldsymbol{t}_j^i$. In line 15 $\mathcal{W}_j^{i,i}$ is set to $\phi$. For the winning task executors determined above, the payment calculation is carried out in lines 16-19 of Algorithm \ref{algo:11}. In line 20, the winners (task executors) for all the tasks in $\boldsymbol{t}^i$ is held in $\mathcal{W}^i$. The payment of all the task executors for all the tasks in $\boldsymbol{t}^i$ is stored in line 21 in $p^i$ data structure. The winners for all the tasks of $\ell^{th}$ slot is held in $\mathcal{W}^\ell$ using line 24. In line 25, $\mathcal{W}$ holds the winners for the available tasks in the distinct slots. $p^{\ell}$ captures the task executors (as winners) in line 26. In line 27, the payment vector of the task executors (as winners) for the tasks in all the available slots is held in $p$. The sets $\mathcal{W}^\ell$ and $p^\ell$ are set to $\phi$ in line 28. Finally in line 30,  $\mathcal{W}^\ell$ and $p^\ell$ are returned.
\begin{algorithm}[H]
\caption{\textsc{Allocation and pricing rule} ($\mathbb{F}$, $\boldsymbol{e}$, $\ell$, $c$, $\mathcal{B}$)}\label{algo:1}
\begin{algorithmic}[1]
\STATE $k \gets 1$, $\mathcal{W} \gets \phi$, $p \gets \phi$.
\FOR{each slot $\ell$}
\FOR{each task requester $\boldsymbol{r}_i \in \mathbb{F}$}
\FOR{each task $\boldsymbol{t}_j^i \in \boldsymbol{t}^i$}
\STATE $\hat{\boldsymbol{e}}$ $\gets$ Sort~($\boldsymbol{e}^\ell$, $c$) \COMMENT{Sort $\boldsymbol{e}^\ell$ in increasing order of their cost and store it in $\hat{\boldsymbol{e}}$.}
\FOR{each $\boldsymbol{e}_f$ $\in$ $\hat{\boldsymbol{e}}$}
\IF{$c_f$ $\leq$ $\frac{\lfloor \mathcal{B}_i/n_i \rfloor}{k}$}
\STATE $\mathcal{W}_j^{i,f}$ $\leftarrow$ $\mathcal{W}_j^{i,f}$ $\cup$ $\{\boldsymbol{e}_f\}$ \COMMENT{Holds the winning task executors for each task $\boldsymbol{t}_j^i$.}
\STATE $k \gets k+1$
\ELSE
\STATE break;
\ENDIF
\ENDFOR
\STATE $\mathcal{W}_j^i$ $\leftarrow$ $\mathcal{W}_j^i$ $\cup$ $\mathcal{W}_j^{i,f}$ \STATE $\mathcal{W}_j^{i,f} \leftarrow \phi$
\FOR{each $\boldsymbol{e}_i$ $\in$ $\mathcal{W}_j^i$ }
\STATE $p_i$ $\gets$ $min\bigg\{\frac{\lfloor \mathcal{B}_i/n_i \rfloor}{k},~c_{k+1}\bigg\}$ 
\STATE $p_j^i \gets p_j^i \cup p_i$
\ENDFOR
\STATE $\mathcal{W}^i \gets \mathcal{W}^i \cup \mathcal{W}_j^i$  
\STATE $p^i \gets p^i \cup p_j^i$ \COMMENT{Holds the payment of all winning task executors for all the tasks in $\boldsymbol{t}^i$.}
\ENDFOR
\ENDFOR 
\STATE $\mathcal{W}^\ell$ $\gets$ $\mathcal{W}^\ell$ $\cup$ $\mathcal{W}^i$
\STATE $\mathcal{W}$ $\gets$ $\mathcal{W}$ $\cup$ $\mathcal{W}^\ell$ \COMMENT{Holds the winning task executors for all tasks in all the available slots.}
\STATE $p^\ell \gets p^\ell \cup p^i$
\STATE $p \gets p \cup p^\ell$
\STATE $\mathcal{W}^\ell \gets \phi$, $p^\ell \gets \phi$
\ENDFOR
\STATE return $\mathcal{W}$,  $p$
\end{algorithmic}
\end{algorithm}

\subsection{Illustration of Allocation and Pricing Rule with an Example}
The presented example illustrates the working of Algorithm \ref{algo:1}.  For understanding purpose let us fix a slot, say slot 1, and task $\boldsymbol{t}_1^3$ of task requester  $\boldsymbol{r}_3 \in \mathbb{F}$. Following line 5 of Algorithm \ref{algo:1} the task executors present in slot 1 is given as $\boldsymbol{e}^1$ = \{$\boldsymbol{e}_1, \boldsymbol{e}_2, \boldsymbol{e}_3, \boldsymbol{e}_4, \boldsymbol{e}_5, \boldsymbol{e}_6, \boldsymbol{e}_7, \boldsymbol{e}_8, \boldsymbol{e}_9, \boldsymbol{e}_{10}$\} is sorted in increasing order of their cost. The cost vector of the participating task executors of slot 1 is $c$ = \{3, 2, 9, 4, 3, 5, 3, 9, 10, 10\}. The sorted ordering of the task executors is given as $\hat{\boldsymbol{e}}$ = \{$\boldsymbol{e}_2, \boldsymbol{e}_5, \boldsymbol{e}_7, \boldsymbol{e}_1, \boldsymbol{e}_4, \boldsymbol{e}_6, \boldsymbol{e}_3, \boldsymbol{e}_8, \boldsymbol{e}_9, \boldsymbol{e}_{10}$\}. Following lines 6-13 of Algorithm \ref{algo:1}, the task executor $\boldsymbol{e}_2$ is picked-up and a check is made that $2$ $\leq$ $\frac{\lfloor 30/4 \rfloor}{1}$ = $2$ $\leq$ $7.5$. So, the stopping condition is true and  $\mathcal{W}_1^{3,2}$ = \{$\boldsymbol{e}_2$\}. Now, the $k$ value for next iteration will be incremented by $1$ and it will be $k=2$. In the next iteration the check for $\boldsymbol{e}_5$ is done that whether $3$ $\leq$ $\frac{\lfloor 30/4 \rfloor}{2}$ = $3$ $\leq$ $3.75$. So, the stopping condition is true and $\mathcal{W}_1^{3,5}$ = \{$\boldsymbol{e}_5$\}. Now, the $k$ value for next iteration will be incremented by $1$ and it will be $k=3$. In the next iteration $\boldsymbol{e}_7$ is picked up from the sorted ordering and a check is made $3$ $\leq$ $\frac{\lfloor 30/4 \rfloor}{3}$ = $3$ $\leq$ $2.5$. The stopping condition is not satisfied and no further task executors will be hired for task $\boldsymbol{t}_1^3$. So, $\mathcal{W}_1^3$ = \{$\boldsymbol{e}_2$, $\boldsymbol{e}_5$\}. $\mathcal{W}_1^{3,2} = \mathcal{W}_1^{3,5}$ =  $\phi$. Using lines 16-19 the payment of $\boldsymbol{e}_2$ and $\boldsymbol{e}_5$ is calculated. The payment  $p_2$ = $min\bigg\{\frac{\lfloor 30/4 \rfloor}{2}, 3\bigg\}$ = $min\bigg\{\frac{\lfloor {7.5} \rfloor}{2}, {3}\bigg\}$ = $min\bigg\{ {3.75}, {3}\bigg\} =3$. The payment  $p_5$ = $min\bigg\{\frac{\lfloor 30/4 \rfloor}{2}, 3\bigg\}$ = $min\bigg\{\frac{\lfloor {7.5} \rfloor}{2}, {3}\bigg\}$ = $min\bigg\{ {3.75}, {3}\bigg\} =3$. The total payment made is 3+3=6 $\leq$ 7.5.  In a similar manner, the task executors for other tasks of $\boldsymbol{r}_2$, $\boldsymbol{r}_3$, and $\boldsymbol{r}_5$ task requesters will be decided and are given as $\mathcal{W}_1^5$ = \{$\boldsymbol{e}_1$, $\boldsymbol{e}_2$, $\boldsymbol{e}_6$, $\boldsymbol{e}_8$\}, $\mathcal{W}_3^3$ = \{$\boldsymbol{e}_4$, $\boldsymbol{e}_1$, $\boldsymbol{e}_7$\}, $\mathcal{W}_1^2$ = \{$\boldsymbol{e}_9$, $\boldsymbol{e}_3$\},  $\mathcal{W}_2^2$ = \{$\boldsymbol{e}_{10}$, $\boldsymbol{e}_8$\},  $\mathcal{W}_2^3$ = \{$\boldsymbol{e}_1$, $\boldsymbol{e}_5$, $\boldsymbol{e}_4$\}, $\mathcal{W}_4^3$ = \{$\boldsymbol{e}_7$, $\boldsymbol{e}_3$\}, $\mathcal{W}_2^5$ = \{$\boldsymbol{e}_8$, $\boldsymbol{e}_6$, $\boldsymbol{e}_3$, $\boldsymbol{e}_1$,$\boldsymbol{e}_4$\}, $\mathcal{W}_3^5$ = \{$\boldsymbol{e}_{10}$, $\boldsymbol{e}_8$, $\boldsymbol{e}_6$\}. The payment of the winners are $p_1^5$ =\{4.15, 4.15, 4.15,4.15\}, $p_3^3$=\{2, 2, 2\}, $p_1^2$=\{5, 5\}, $p_2^2$=\{5, 5\},  $p_2^3$=\{2.5, 2.5, 2.5\},  $p_4^3$=\{3, 3\},  $p_2^5$=\{3, 3, 3, 3, 3\},  $p_3^5$=\{5, 5, 5\}.

\section{Analysis of BULINC}
\label{sec:atppm}
The theoretical analysis followed by the probabilistic analysis of BULINC is done in this section. Lemma \ref{lem:1} shows that BULINC is computationally efficient. It means that BULINC runs in poly time. In Lemma \ref{the1} it is estimated that how many task requesters among the available ones are receiving the Government fund and is given as $\frac{\mathcal{N}}{\lambda}$. Here, $\mathcal{N}$ as mentioned in Section \ref{s:prelim} is the number of task requesters, and $\frac{1}{\lambda}$ is the probability that any $i^{th}$ task requester receives the Government fund. The probability that at least $\frac{3\mathcal{N}}{\lambda}$ task requesters are getting fund for their tasks from the Government is at most $\frac{1}{3}$ and is shown in Lemma \ref{lemma:21}. In Lemma \ref{lemma:ATO} it is shown that the probability that the Government's fund is received by at least one of the task requesters is given as $1 - \frac{1}{e^{\mathcal{N}\lceil \ln \mathcal{N} \rceil}}$. On an average, the number of tasks placed in any $j^{th}$ slot by BULINC is given as $\frac{\sum_{i=1}^{\mathcal{N} n_i}}{|d|}$ and is proved in Lemma \ref{lemma:8}. Lemma \ref{lem:2} shows that BULINC is truthful. From the definition of \emph{truthfulness} given in Section \ref{s:prelim} it means that the participating task executors can maximize their utility by reporting their private information in \emph{truthful} manner and not by misreporting it. In Lemma \ref{lem:3} it is shown that BULINC is individually rational. It means that the participating task executors will have zero or positive utility in case of BULINC. Lemma \ref{lem:3bf} shows that BULINC is budget feasible. It means that the payment received by the task executors in case of BULINC is at most the available budget.  

\begin{lemma}
\label{lem:1}
BULINC is computationally efficient. 
\end{lemma}

\begin{proof}
In order to show that BULINC is computationally efficient, it is sufficient to show that Algorithm \ref{algo:10}, Algorithm \ref{algo:11}, and Algorithm \ref{algo:1} are computationally efficient. The time taken by Algorithm \ref{algo:10} is $O(\mathcal{N}\lg\mathcal{N})$. $O((\sum_{i=1}^{\mathcal{N}} n_i)^2)$ is the computation time of Algorithm \ref{algo:11}. Algorithm \ref{algo:1} is bounded above by $O(\ell(n(n_im_\ell + n_i|\mathcal{W}_i^j|)))$.  If $n_i$, $m_\ell$, and $|\mathcal{W}_i^j|$ is a function of $n$, then the computation time of Algorithm \ref{algo:1} will be $O(n^3)$. \\
\indent As the Algorithm \ref{algo:10}, Algorithm \ref{algo:11}, and Algorithm \ref{algo:1} are computationally efficient, so as the BULINC.
    
\end{proof}

\begin{lemma}
\label{the1}
The expected number of task requesters receiving funds for their tasks from the Government is given as $\frac{{\mathcal{N}}}{\lambda}$. More formally, 
E[$\mathbb{Z}^\ast$] = $\frac{{\mathcal{N}}}{\lambda}$.
\end{lemma}
\begin{proof}
The proof determines that in \emph{expectation how many task requesters out of $\mathcal{N}$ task requesters are getting funds for their tasks from the Government}? The total number of task requesters receiving funds from the Government is determined using the random variable $\mathbb{Z}^*$. So, the number of task requesters receiving funds from the Government is given as E[$\mathbb{Z}^\ast$]. In our setting, the sample space for any task requester $\boldsymbol{r}_i$ is given as $\mathcal{S}$ = \{$\boldsymbol{r}_i$ receives fund from the Government, $\boldsymbol{r}_i$ does not receive fund from the Government\}, with $Pr\{ \boldsymbol{r}_i$ receives fund from the Government $\} = \frac{1 }{\lambda}$, and $Pr\{ \boldsymbol{r}_i$ does not receives fund from the Government$\}$ = $1-\frac{1 }{\lambda}$. The random variable $\mathbb{Z}_i^\ast$ is given as: $\mathbb{Z}_i^* = I\{ \boldsymbol{r}_i$ receives fund from the Government\}. So, $\mathbb{Z}_i^*$ will be 1 if $\boldsymbol{r}_i$ gets fund from the Government, otherwise it will be 0. The expected number of times a task requester $\boldsymbol{r}_i$ receiving the funds from the Government for their tasks is just equal to  $E[\mathbb{Z}_i^*]$.
\begin{equation*}
 E[\mathbb{Z}_i^*] =  E[I\{\boldsymbol{r}_i~receives~fund~from~the~Government\}] 
  \end{equation*}
From the definition of indicator random variable \cite{Coreman_2009}, we have \\
  \begin{equation*}
 E[\mathbb{Z}_i^*] = 1 \cdot Pr\{\mathbb{Z}_i^* = 1\} + 0 \cdot Pr\{\mathbb{Z}_i^* = 0\}
 \end{equation*}
 \begin{equation*}
\hspace*{-14mm} = 1 \cdot \frac{1 }{\lambda} + 0 \cdot \bigg(1 - \frac{1}{\lambda}\bigg)
\end{equation*}
\begin{equation*}
\hspace*{-36mm} = 1 \cdot \frac{1}{\lambda}
\end{equation*}
\begin{equation}
\label{equ:t1}
\hspace*{-40mm} = \frac{1}{\lambda}
\end{equation}
The number of task requesters receiving funds for their tasks from the Government is given as: 
\begin{equation}
\label{equ:t2}
\hspace*{1.5mm} \mathbb{Z}^* = \sum_{i=1}^{\mathcal{N}}{\mathbb{Z}_i^*}
\end{equation}
After taking expectation both side of equation \ref{equ:t2} we get 
\begin{equation*}
\hspace*{1.5mm} E[\mathbb{Z}^*] = E\bigg[\sum_{i=1}^{\mathcal{N}}{\mathbb{Z}_i^*\bigg]}
\end{equation*}
From linearity of expectation, we get 
\begin{equation}
\label{equ:t3}
\hspace*{1.5mm} E[\mathbb{Z}^*] = \sum_{i=1}^{\mathcal{N}}{E[\mathbb{Z}_i^*]}
\end{equation}
Put the calculated value of $E[\mathbb{Z}_i^*]$ from equation \ref{equ:t1} to equation \ref{equ:t3}, we get
\begin{equation*}
E[\mathbb{Z}^*]=  \sum_{i=1}^{\mathcal{N}}\frac{1}{\lambda}
\end{equation*}
\begin{equation*}
\hspace*{12mm}=  \frac{1 }{\lambda}\sum_{i=1}^{\mathcal{N}}{1}
\end{equation*}
\begin{equation*}
\hspace*{10mm} = \frac{1 }{\lambda} \cdot {\mathcal{N}}
\end{equation*}
\begin{equation*}
\hspace*{6mm} =  \frac{{\mathcal{N}}}{\lambda}
\end{equation*}
Hence proved.
\end{proof}

\begin{observation}
\label{th:obs}
If the probability of receiving the funds from the Government for $\boldsymbol{r}_i$ task requester is $\frac{1}{5}$ i.e. $\lambda = 5$ then one-fifth of the total available task requesters will be getting the Government fund. On the other hand, if the probability of receiving the fund from the Government for $\boldsymbol{r}_i$ task requester is $\frac{1}{2}$  i.e. $\lambda = 2$ then half of the total available task requesters will be receiving the funds from the Government. It means that higher the probability that a task requester $\boldsymbol{r}_i$ receives a fund from the Government, higher will be the number of task requesters receiving the funds from the Government.
\end{observation}

\hspace*{5mm}
\begin{lemma}
\label{lemma:21}
In BULINC, we have
\begin{equation*}
Pr\bigg \{\mathbb{Z}^\ast \geq \frac{3\mathcal{N}}{\lambda}\bigg\} \leq \frac{1}{3}
\end{equation*}
\end{lemma}

\begin{proof}
From lemma \ref{the1}, we have already defined $\mathbb{Z}^\ast$ as the total number of task requesters receiving funds from the Government for their tasks. It is given as $I$ = \{Number of task requesters receives fund from the Government\} 
 \begin{equation}
 \label{equ:1}
   I = \begin{cases} 1   ,\quad \text{if $\mathbb{Z}^\ast \geq \frac{3\mathcal{N}}{\lambda}$}  \\
     0, \quad \text{otherwise }
  \end{cases}
  \end{equation}
  From equation \ref{equ:1}, it can be written as 
  \begin{equation*}
      \mathbb{Z}^\ast \geq \frac{3\mathcal{N}}{\lambda}
  \end{equation*}
   \begin{equation*}
      \frac{\mathbb{Z}^\ast}{\big(\frac{3\mathcal{N}}{\lambda}\big)} \geq 1
  \end{equation*}
 \begin{equation}
 \label{equ:2}
      \frac{\mathbb{Z}^\ast}{\big(\frac{3\mathcal{N}}{\lambda}\big)} \geq I
  \end{equation}
  Taking expectation both side of equation \ref{equ:2}, we get
   \begin{equation*}
      E\bigg[\frac{\mathbb{Z}^\ast}{\big(\frac{3\mathcal{N}}{\lambda}\big)} \bigg] \geq E[I]  
  \end{equation*}
  or
    \begin{equation*}
     E[I] \leq  E\bigg[\frac{\mathbb{Z}^\ast}{\big(\frac{3\mathcal{N}}{\lambda}\big)} \bigg]
  \end{equation*}
   \begin{equation*}
   \hspace*{9mm} = \frac{\lambda}{3\mathcal{N}} \cdot E[\mathbb{Z}^\ast]
  \end{equation*}
 \begin{equation*}
E[I] \leq  \frac{\lambda}{3\mathcal{N}} \cdot E[\mathbb{Z}^\ast]
  \end{equation*}
 \begin{equation}
 \label{equ:3}
 Pr\bigg \{\mathbb{Z}^\ast \geq \frac{3\mathcal{N}}{\lambda}\bigg\} \cdot 1 \leq \frac{\lambda}{3\mathcal{N}} \cdot E[\mathbb{Z}^\ast]
  \end{equation}
Substituting the value of $E[\mathbb{Z}^\ast]$ from Lemma \ref{the1} in equation \ref{equ:3} above, we get 
  \begin{equation}
 \label{equ:4}
 Pr\bigg \{\mathbb{Z}^\ast \geq \frac{3\mathcal{N}}{\lambda}\bigg\} \cdot 1\leq \frac{\lambda}{3\mathcal{N}} \cdot \frac{{\mathcal{N}}}{\lambda}
  \end{equation}
  we can rewrite equation \ref{equ:4} as: 
   \begin{equation*}
 Pr\bigg \{\mathbb{Z}^\ast \geq \frac{3\mathcal{N}}{\lambda}\bigg\}\leq \frac{1}{3}
  \end{equation*}
  Hence proved
\end{proof}
\vspace*{2mm}
\begin{lemma}
\label{lemma:ATO}
The probability that the Government's fund is received by at least one of the task requesters is bounded above by $1 - \frac{1}{e^{\mathcal{N}\lceil\ln \mathcal{N}\rceil}}$. 
\end{lemma}
\begin{proof}
In this lemma, the objective is to determine that \emph{what is the probability that at least one task requester among the available ones will receive the Government fund}? In order to complete the proof, the result obtained in Lemma \ref{the1} will be quite handy. Distributing the Government's fund to any $\boldsymbol{r}_i$ task requester is independent of distributing the funds to other task requesters in the set $\boldsymbol{r}$. The probability that $\boldsymbol{r}_i$ task requester has not received any fund from the Government is given as:
\begin{equation*}
 Pr\{Z^\ast < 1\} = \bigg(1-\frac{1}{\lambda}\bigg) \times \bigg(1-\frac{1}{\lambda}\bigg) \times \ldots~\mathcal{N}~times
\end{equation*}
\begin{equation}
\label{equ:lm}
 = \bigg(1-\frac{1}{\lambda}\bigg)^\mathcal{N}
\end{equation}
The inequality $1 + \mathcal{N} \leq e^{\mathcal{N}}$ gives us:
\begin{equation*}
 Pr\{Z^\ast < 1\} \leq e^{-\mathcal{N} \cdot \frac{1}{\lambda}}
\end{equation*}
\begin{equation}
\label{eq:lm1}
 = \frac{1}{e^{\big(\frac{\mathcal{N}}{\lambda}\big)}}
\end{equation}
From equation \ref{eq:lm1}, the probability that  the Government's fund is received by at least one of the task requesters is given as:
\begin{equation}
\label{eq:lm2}
 Pr\{Z^\ast \geq 1\} \leq 1 - \frac{1}{e^{\big(\frac{\mathcal{N}}{\lambda}\big)}}
\end{equation}
Now, considering the value of $\lambda = \lceil \ln \mathcal{N} \rceil$ and substituting it in equation \ref{eq:lm2}, we get
\begin{equation*}
 Pr\{Z^\ast \geq 1\} \leq 1 - \frac{1}{e^{\mathcal{N}\lceil \ln \mathcal{N} \rceil}}
\end{equation*}
Hence proved.
\end{proof}

\begin{corollary}
Let us say the number of available task requester is $\mathcal{N} = 100$, then we get
\begin{equation*}
 Pr\{Z^\ast \geq 1\} \leq 1 - \frac{1}{e^{100 \lceil \ln 100 \rceil}}
\end{equation*}
\begin{equation*}
= 1 - \frac{1}{0.2172}
\end{equation*}
\begin{equation*}
\hspace*{-8.5mm} = 0.783
\end{equation*}
So, it can be seen from above calculation that the probability that  the Government's fund is received by at least one of the task requesters out of 100 task requesters is very high.
\end{corollary}
\vspace*{2mm}
\begin{lemma}
\label{lemma:8}
In expectation the number of tasks placed in any $d_j$ slot is given as $\frac{\sum\limits_{i=1}^{\mathcal{N}} n_i}{|d|}$, where $n_i$ is the number of tasks held by $\boldsymbol{r}_i$ task requester and $|d|$ is the number of slots. More formally, 
\begin{equation*}
E[Y] = \Bigg(\frac{\sum\limits_{i=1}^{\mathcal{N}} n_i}{|d|}\Bigg) 
\end{equation*} 
where, the indicator random variable $Y$ holds the total number of tasks allocated to any slot $d_j$.  
\end{lemma}

\begin{proof}
It can be proved by considering the event of scheduling of tasks to $|d|$ slots as balls and bins problem \cite{Coreman_2009}. When the tasks are scheduled to $|d|$ different slots it can be thought of as tossing a ball (in our case scheduling a task) randomly to a bin (in our case a slot). From the construction of Algorithm \ref{algo:11}, placing  a task to any slot $d_j$ is independent of placing some other tasks to the slots. As the event of placing a task to any slot $d_j \in d$ is independent and any slot out of $|d|$ slots has equal probability to be chosen. So, the probability that the task is placed in $d_j$ slot is $\frac{1}{|d|}$. The probability $\big(1 - \frac{1}{|d|}\big)$ is a task will not be placed in $d_j$ slot. From the definition of random variable, the expected value of indicator random variable is said to be equal to the probability of placing a task $\boldsymbol{t}_k^i$ to any $j^{th}$ slot and is given as:
\begin{equation*}
E[Y_{k}^j] = 1 \cdot \bigg(\frac{1}{|d|}\bigg) + 0 \cdot \bigg(1 - \frac{1}{|d|}\bigg)
\end{equation*} 

\begin{equation*}
 = 1 \cdot \bigg(\frac{1}{|d|}\bigg) 
\end{equation*}

\begin{equation}
\label{label:345}
 = \frac{1}{|d|} 
\end{equation}
$Y_j$ captures the number of tasks out of $n_i$ tasks allocated to any $d_j$ slot. Further considering the random variable $Y_j$, we get
\begin{equation}
\label{label:sd}
Y_j = \sum_{k=1}^{n_i} Y_k^j
\end{equation} 
Take expectation on both side of equation \ref{label:sd}, we get
\begin{equation}
\label{label:sd2}
E[Y_j] = E\bigg[\sum_{k=1}^{n_i} Y_k^j\bigg]
\end{equation} 
Linearity of expectation gives,
\begin{equation}
\label{label:sd3}
E[Y_j] = \sum_{k=1}^{n_i} E[Y_k^j]
\end{equation} 
Substituting the value of $E[Y_k^j]$ from equation \ref{label:345} to equation \ref{label:sd3}, we get
\begin{equation*}
E[Y_j] = \sum_{k=1}^{n_i} \frac{1}{|d|}
\end{equation*} 
\begin{equation}
\label{label:sd31}
 = \frac{n_i}{|d|}
\end{equation}  
In our case, we are interested in the random variable $Y$ and is given as:

\begin{equation}
\label{label:456}
 Y = \sum_{i=1}^{\mathcal{N}} Y_j
\end{equation}
Taking expectation both side of equation \ref{label:456}, we get
\begin{equation}
\label{label:4561}
 E[Y] = E\big[\sum_{i=1}^{\mathcal{N}} Y_j\big]
\end{equation}
By linearity of expectation, we get
\begin{equation}
\label{label:4561}
 E[Y] = \sum_{i=1}^{\mathcal{N}} E[Y_j]
\end{equation}
Substituting the value of $E[Y_j]$ from equation \ref{label:sd31} to equation \ref{label:4561}, we get
\begin{equation*}
 E[Y] = \Bigg(\frac{\sum_{i=1}^{\mathcal{N}} n_i}{|d|}\Bigg)
\end{equation*}
 Hence proved.
\end{proof}
\vspace*{2mm}
\begin{lemma}
\label{lem:2}
 BULINC is truthful.
\end{lemma}
\begin{proof}
To prove that BULINC is truthful, it is sufficient to prove that budget distribution mechanism (Algorithm \ref{algo:10}) and allocation and payment rule (Algorithm \ref{algo:1}) are truthful. It is due to the reason that budget distribution mechanism has private votes and allocation and payment rule has private bid values.\\
\indent In order to prove that budget distribution mechanism is \emph{truthful}, we have to show that by misreporting the votes the city dwellers will not gain. If any city dweller misreport his/her (henceforth her) votes (or preference ordering) then it is obvious that the Government funds will be shifted from the tasks that you want to be funded to the tasks that you do not want to be funded. It means that the tasks preferred by her will get less funds or may not get any funds and that will lead to lower utility value than the case when she is reporting her true votes. So, the best way for the city dwellers to maximize utility in budget distribution mechanism is to report truthfully their votes. Hence, the budget distribution mechanism is \emph{truthful}.\\
\indent In \emph{allocation and payment rule} (Algorithm \ref{algo:1}) let us fix a task executor $\boldsymbol{e}_i$ and slot $d_i$. Let us say the task executor $\boldsymbol{e}_i$ was in the winning set when reporting truthfully. Now, let us say that $\boldsymbol{e}_i$ reports the cost $\hat{c}_i$ such that $\hat{c}_i < c_i$. In this case, the task executor $\boldsymbol{e}_i$ will still win as it will appear early in the sorted ordering of task executors. The utility of $\boldsymbol{e}_i$ will be $\hat{u}_i = p_i - c_i = u_i$. On the other hand, if the task executor $\boldsymbol{e}_i$ reports the cost $\hat{c}_i$ such that $\hat{c}_i > c_i$. In such case, the two situations can happen. One situation could be that due to increase in reported cost the task executors may no longer belong to the wining set and is a loser. In such case, the utility of task executor $\boldsymbol{e}_i$ will be $\hat{u}_i = 0$. Another situation could be that with the increase in cost the task executor $\boldsymbol{e}_i$ can still belong to the winning set then her utility will be $\hat{u}_i = p_i - c_i = u_i$. So,  in both the situation the task executor $\boldsymbol{e}_i$ has not been benefited by misreporting her valuation. Hence Algorithm \ref{algo:1} is truthful.\\
\indent From the above discussion it can be inferred that Algorithm \ref{algo:10} and Algorithm \ref{algo:1} are truthful and so as the BULINC.             
\end{proof}
\begin{lemma}
\label{lem:3}
 BULINC is individually rational.
\end{lemma}

\begin{proof}
To prove that BULINC is individually rational it is to be shown that allocation and payment rule of BULINC is individually rational. Let us fix a task executor $\boldsymbol{e}_i$ and the slot $d_i$. For each of the task executors $\boldsymbol{e}_i$ the check in line 7 of Algorithm \ref{algo:1} guarantees that one of the component of payment rule $i.e.$ $ \frac{\lfloor\mathcal{B}_i/n_i \rfloor}{k}$ is always less than or equal to the cost of the task executor $\boldsymbol{e}_i$. Talking about the other component of the payment rule $i.e.$ $c_{k+1}$ it is always greater than or equal to $c_i$ $\forall i = 1, 2, \ldots, k$. As it is clear that the the reported bid value is always less than or equal to the payment $p_i$ of task executor $\boldsymbol{e}_i$. So, the utility $u_i$ of task executor $\boldsymbol{e}_i$ is always non-negative $i.e.$ $u_i \geq 0$. Hence, BULINC is individually rational.   
\end{proof}

\begin{lemma}
\label{lem:3bf}
BULINC is budget feasible.
\end{lemma}

\begin{proof}
To prove that BULINC is budget feasible it is to be shown that allocation and payment rule of BULINC is budget feasible. Let us fix a task executor $\boldsymbol{e}_i$, the task $\boldsymbol{t}_j^i$ of task requester $\boldsymbol{r}_i$, and the slot $d_i$. From the construction of BULINC it can be seen that the maximum payment that any winning task executor $\boldsymbol{e}_i$ receives is $ \frac{\lfloor\mathcal{B}_i/n_i \rfloor}{k}$, where $k$ is the largest index in the sorted ordering of the task executors that satisfies the stopping criteria $c_k \leq  \frac{\lfloor\mathcal{B}_i/n_i \rfloor}{k}$. For any task $\boldsymbol{t}_j^i$ of task requester $\boldsymbol{r}_i$ in any slot $d_i$, we have
\begin{equation*}
 {p}_j^i =  \displaystyle\sum_{\boldsymbol{e}_i \in \mathcal{W}_j^i} {p}_i \leq \displaystyle\sum_{\boldsymbol{e}_i \in \mathcal{W}_j^i} \frac{\lfloor \mathcal{B}_i/n_i\rfloor}{k} 
\end{equation*}
\begin{equation*}
 = \frac{\lfloor \mathcal{B}_i/n_i\rfloor}{k} \cdot k =  {\lfloor \mathcal{B}_i/n_i\rfloor}
\end{equation*}
So, the total payment made to the task executors as winners for the task $\boldsymbol{t}_j^i$ of the task requester $\boldsymbol{r}_i$ in any $d_i$ slot is at most the budget $i.e.$ $\frac{\lfloor\mathcal{B}_i/n_i \rfloor}{k}$. As it is true for one task of $i^{th}$ task requester, in similar manner it will be true for all the tasks of $i^{th}$ task requester and for all the tasks of other task requesters. Hence, BULINC is budget feasible. 
\end{proof}

\section{Experimental Analysis}
\label{sec:ef}
To support the theoretical analysis carried out in section \ref{sec:atppm} above, in this section the simulations are carried out independently for both the tiers of the proposed framework. First tier \emph{estimates on the number of task requesters receiving the Government fund (GF) from among the available ones} is made. For this purpose, the simulation is carried out that shows the task requesters (TRs) getting the Government fund using BULINC for both, the synthetic data and the real time participatory democracy data (RTPDD) \cite{Som2023}.\\
\indent In the second tier, BULINC is compared with the benchmark mechanisms, namely greedy mechanism (GM) [], Jalaly et. al \cite{Jalaly2017SimpleAE}, and Chen et. al \cite{Chen3801} in terms of \emph{truthfulness}, \emph{budget feasibility}, and \emph{scalability}. The greedy mechanism is written as ‘GM’ in the simulations. For \emph{truthfulness} and \emph{budget feasibility}, BULINC is compared with GM, Jalaly et. al \cite{Jalaly2017SimpleAE}, and Chen et. al \cite{Chen3801} on the metrics \emph{sum of utility of TEs} and \emph{budget utilized} respectively. The metric namely, \emph{sum of utility of TEs} help us to show that in case of BULINC, Jalaly et. al \cite{Jalaly2017SimpleAE}, and  Chen et. al \cite{Chen3801} the participating TEs are not able to gain by misreporting their true bid value. The metric \emph{budget utilized} will help us to show that the payment received by the winning TEs is at most the budget and hence the mechanisms are \emph{budget feasible}. Further, through simulation it is shown that BULINC and the benchmark mechanism are \emph{scalable}. For scalability the running time of each of the mechanisms is determined and plotted. The unit of valuation of the task executors is taken as $\$$. 

\begin{figure*}
        \centering
\begin{subfigure}{0.18\textwidth}
\includegraphics[scale = 0.15]{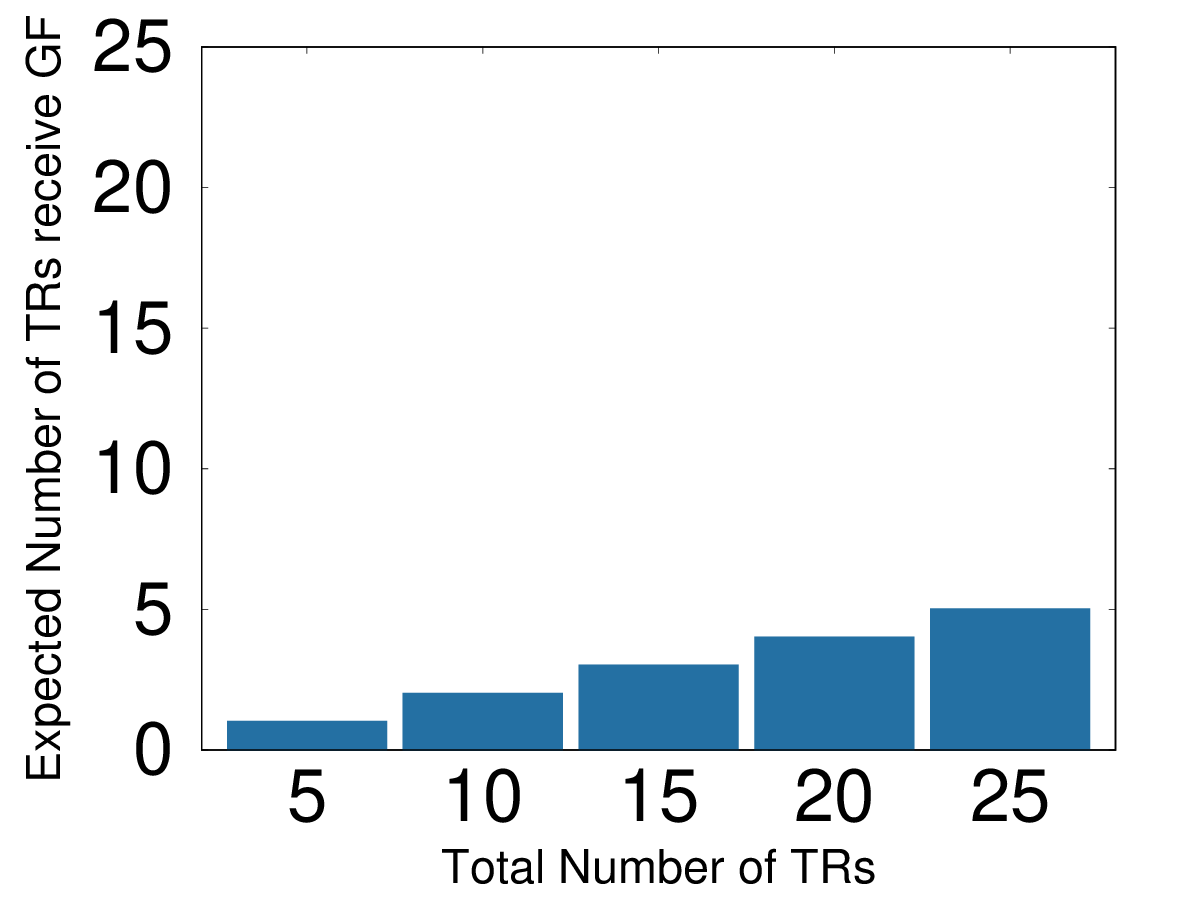}
\subcaption{$Pr\{\mathcal{Z}_i^* = 1\} = 0.20$}
\label{sim:1b1}
\end{subfigure}%
\begin{subfigure}{0.18\textwidth}
\includegraphics[scale = 0.15]{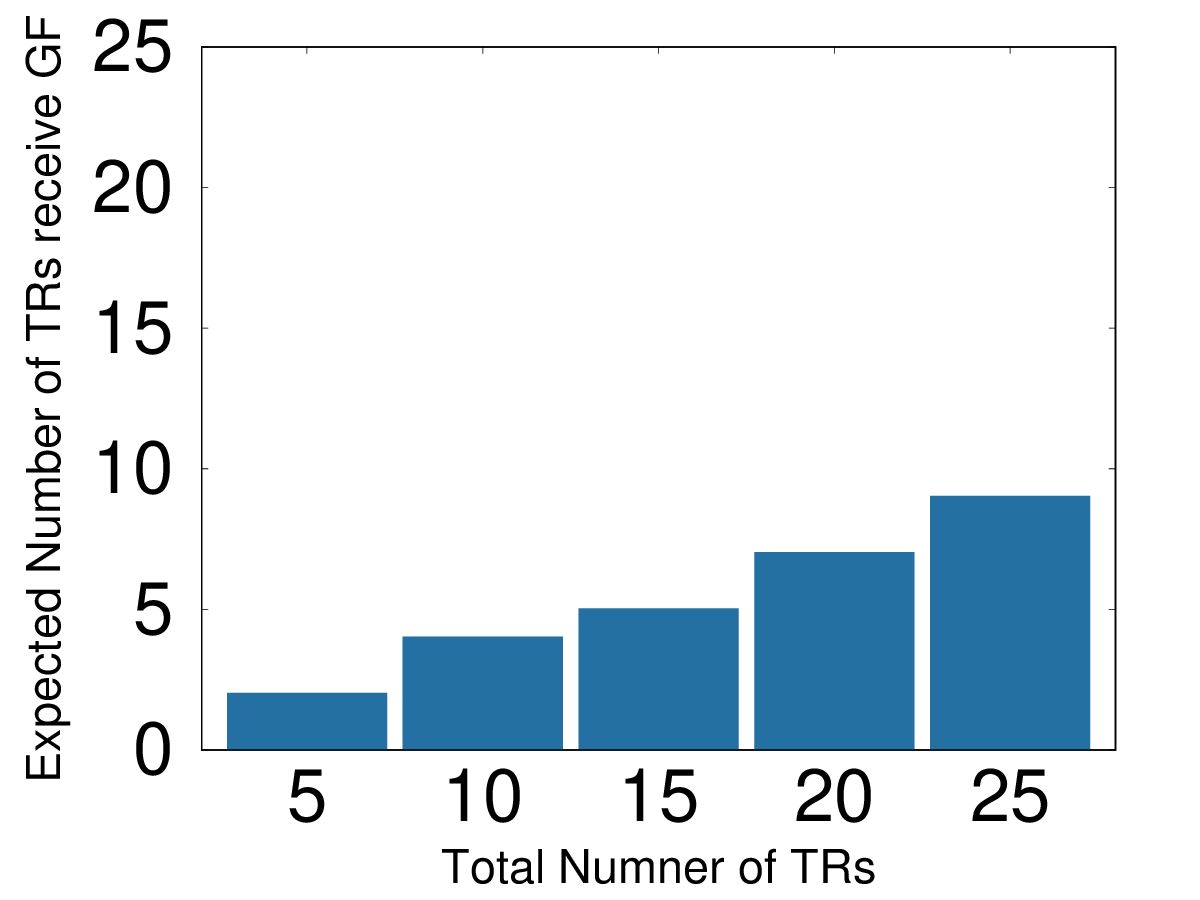}
\subcaption{$Pr\{\mathcal{Z}_i^* = 1\} = 0.33$}
\label{sim:1b2}
\end{subfigure}%
\begin{subfigure}{0.18\textwidth}
\includegraphics[scale = 0.15]{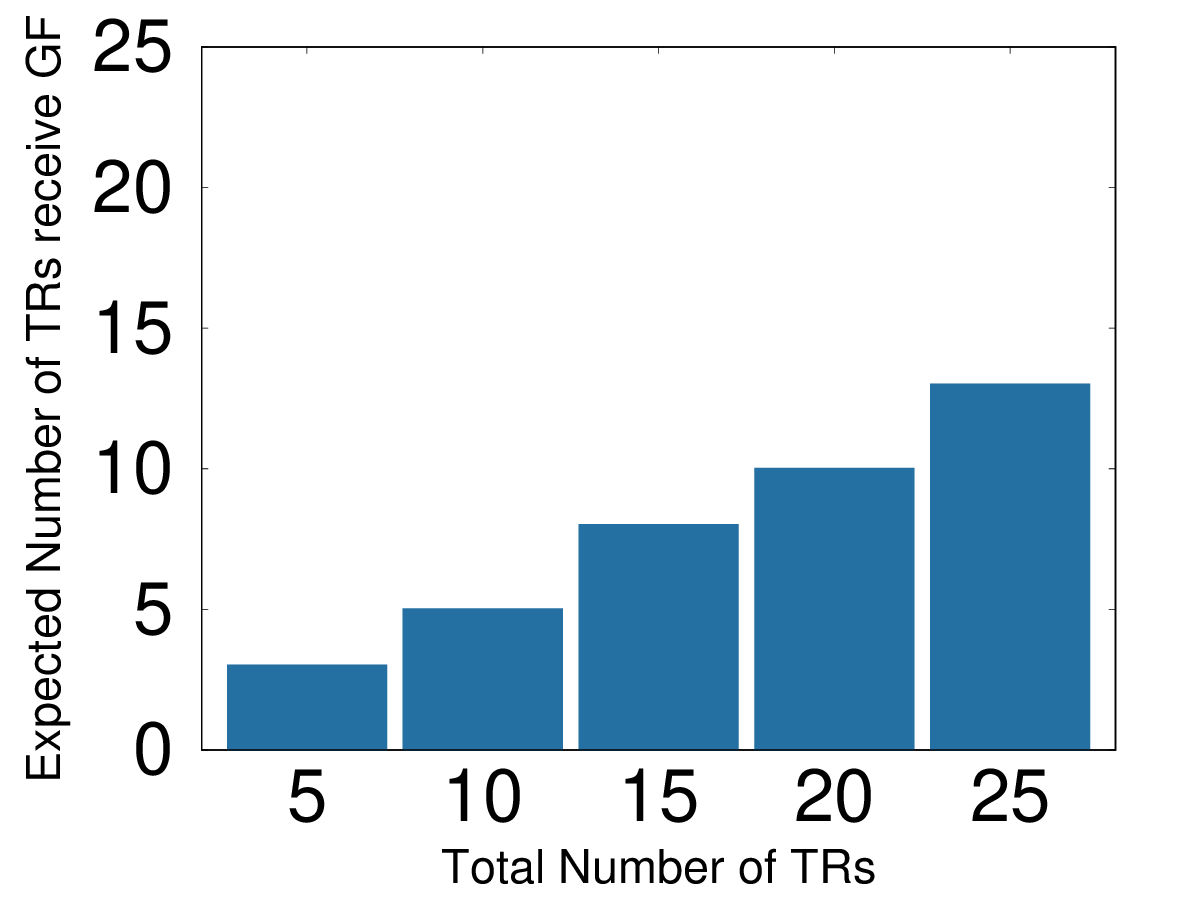}
\subcaption{$Pr\{\mathcal{Z}_i^* = 1\} = 0.50$}
\label{sim:1b3}
\end{subfigure}%
\begin{subfigure}{0.18\textwidth}
\includegraphics[scale = 0.15]{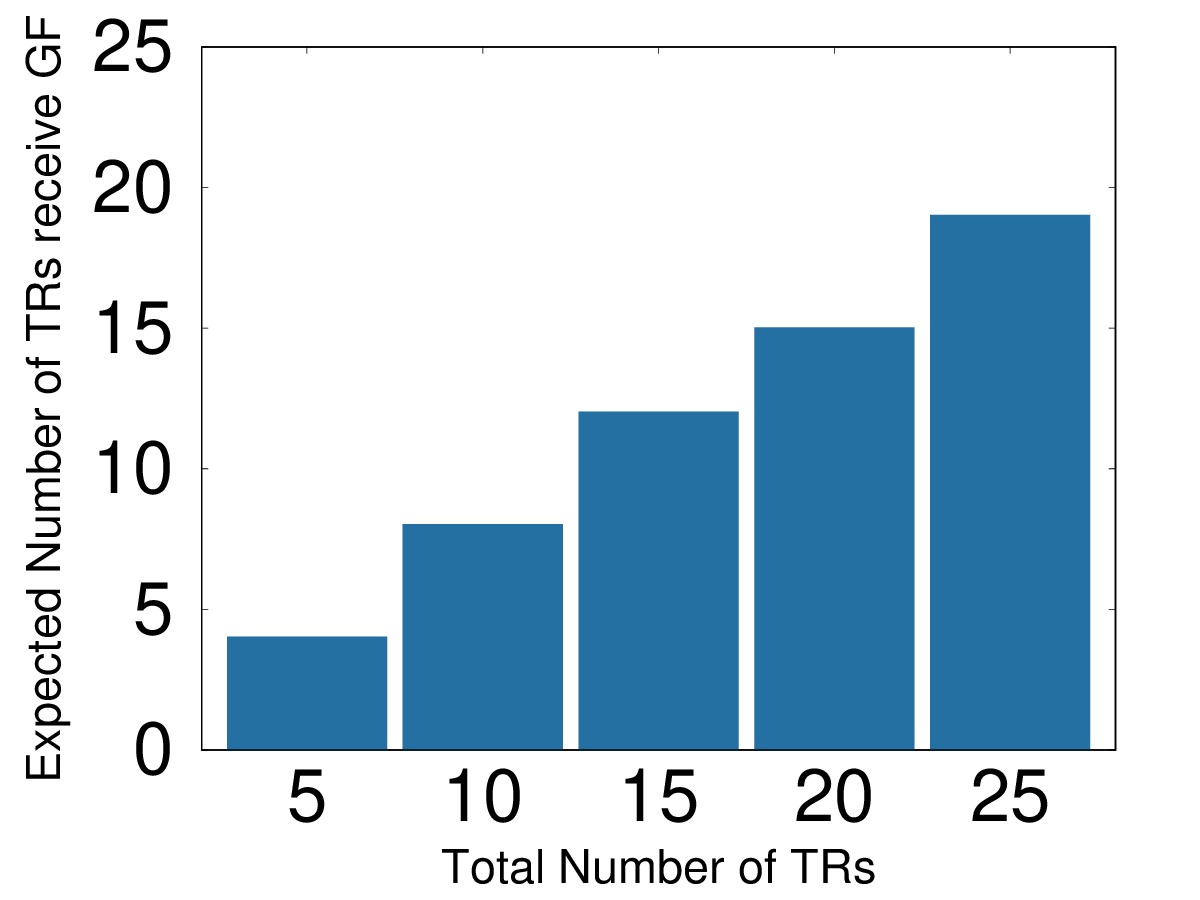}
\subcaption{$Pr\{\mathcal{Z}_i^* = 1\} = 0.75$}
\label{sim:1b4}
\end{subfigure}%
\begin{subfigure}{0.18\textwidth}
\includegraphics[scale = 0.15]{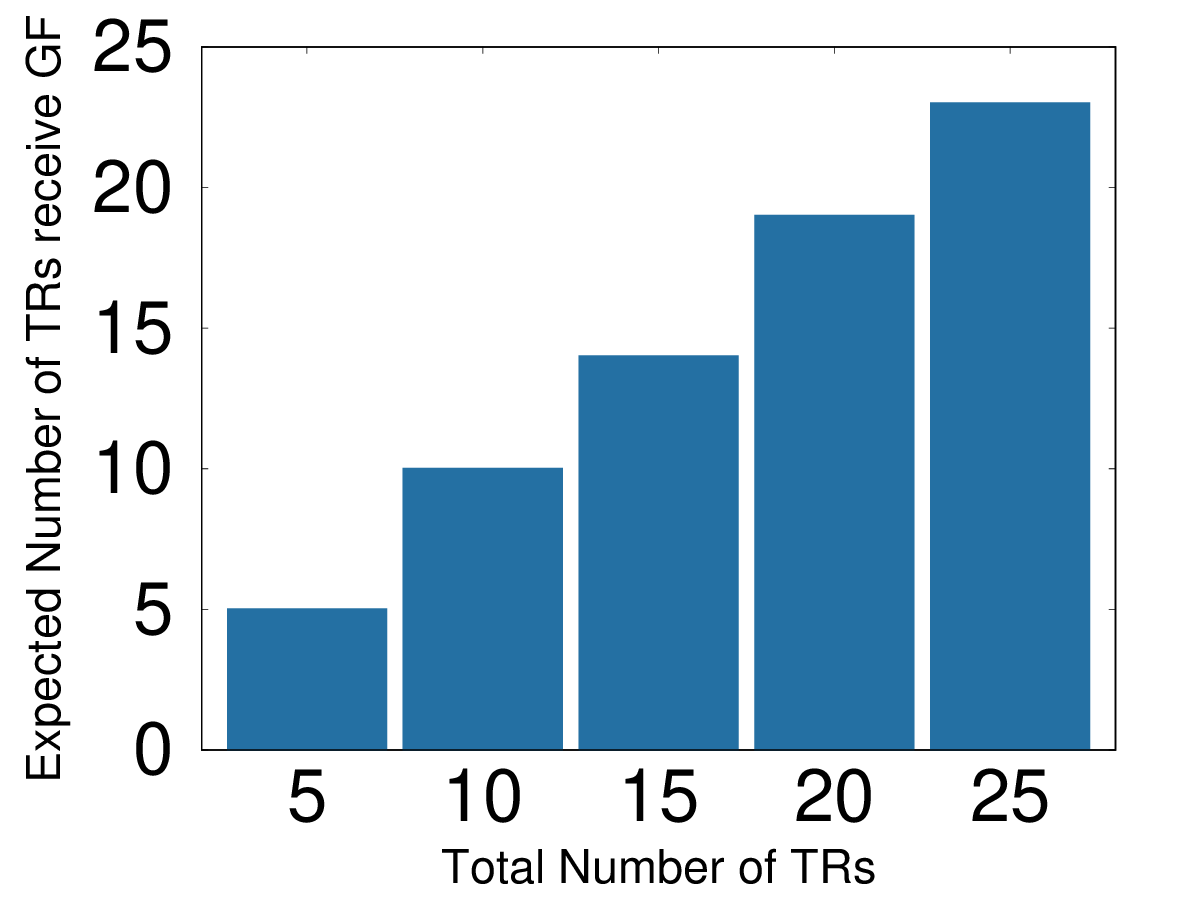}
\subcaption{$Pr\{\mathcal{Z}_i^* = 1\} = 0.92$}
\label{sim:1b5}
\end{subfigure}
        \caption{Comparison of expected number of TRs receive the Government fund based on SPDD}\label{fig:1}
\end{figure*}

 \begin{figure*}
        \centering
\begin{subfigure}{0.18\textwidth}
\includegraphics[scale = 0.15]{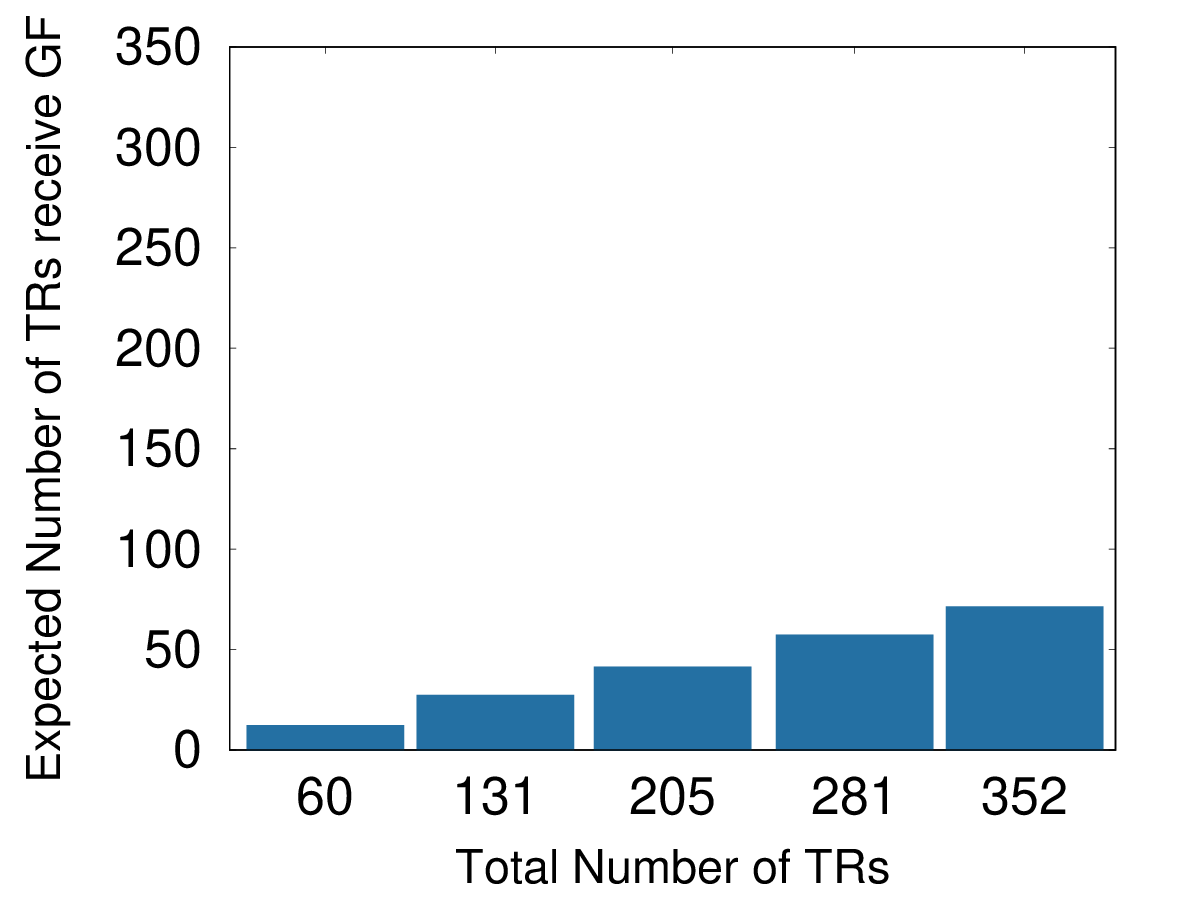}
\subcaption{$Pr\{\mathcal{Z}_i^* = 1\} = 0.20$}
\label{sim:2b1}
\end{subfigure}%
\begin{subfigure}{0.18\textwidth}
\includegraphics[scale = 0.15]{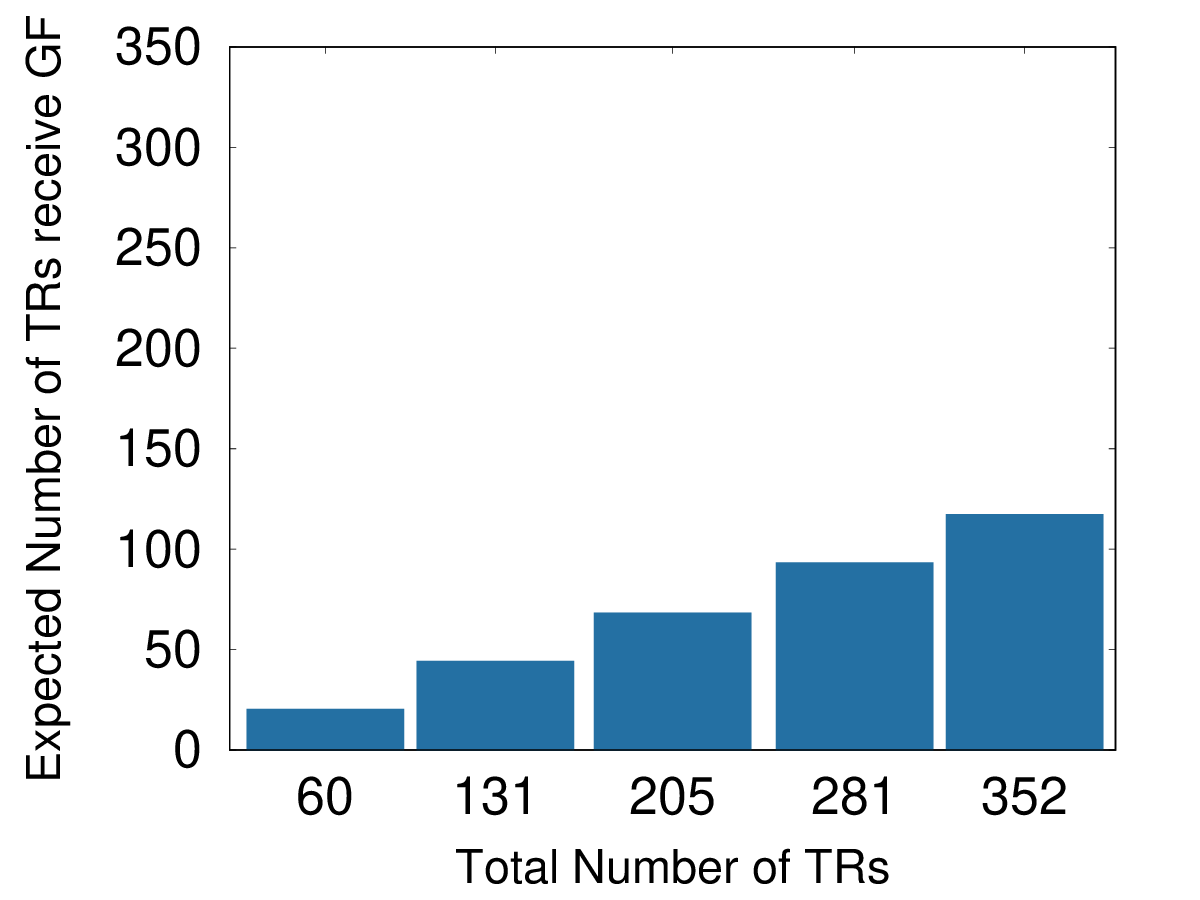}
\subcaption{$Pr\{\mathcal{Z}_i^* = 1\} = 0.33$}
\label{sim:2b2}
\end{subfigure}%
\begin{subfigure}{0.18\textwidth}
\includegraphics[scale = 0.15]{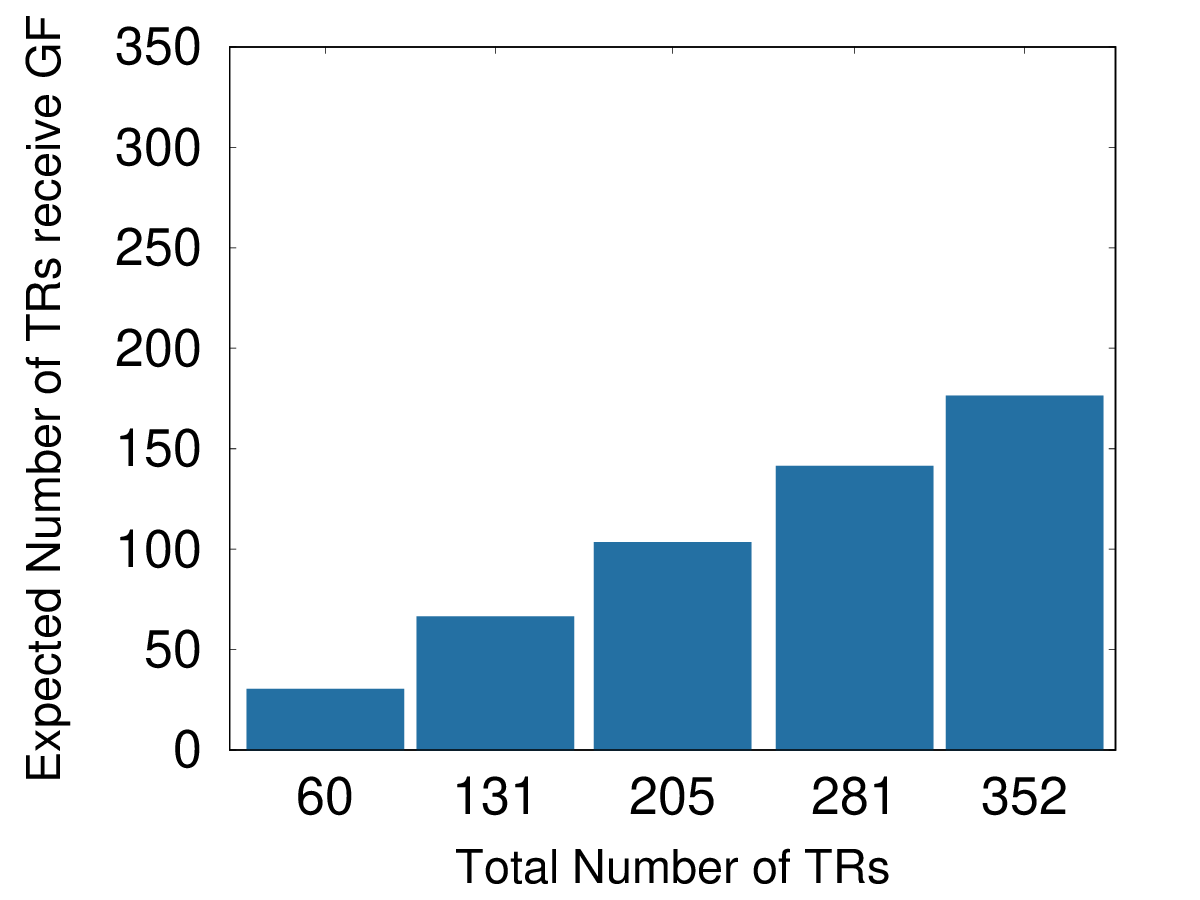}
\subcaption{$Pr\{\mathcal{Z}_i^* = 1\} = 0.50$}
\label{sim:2b3}
\end{subfigure}%
\begin{subfigure}{0.18\textwidth}
\includegraphics[scale = 0.15]{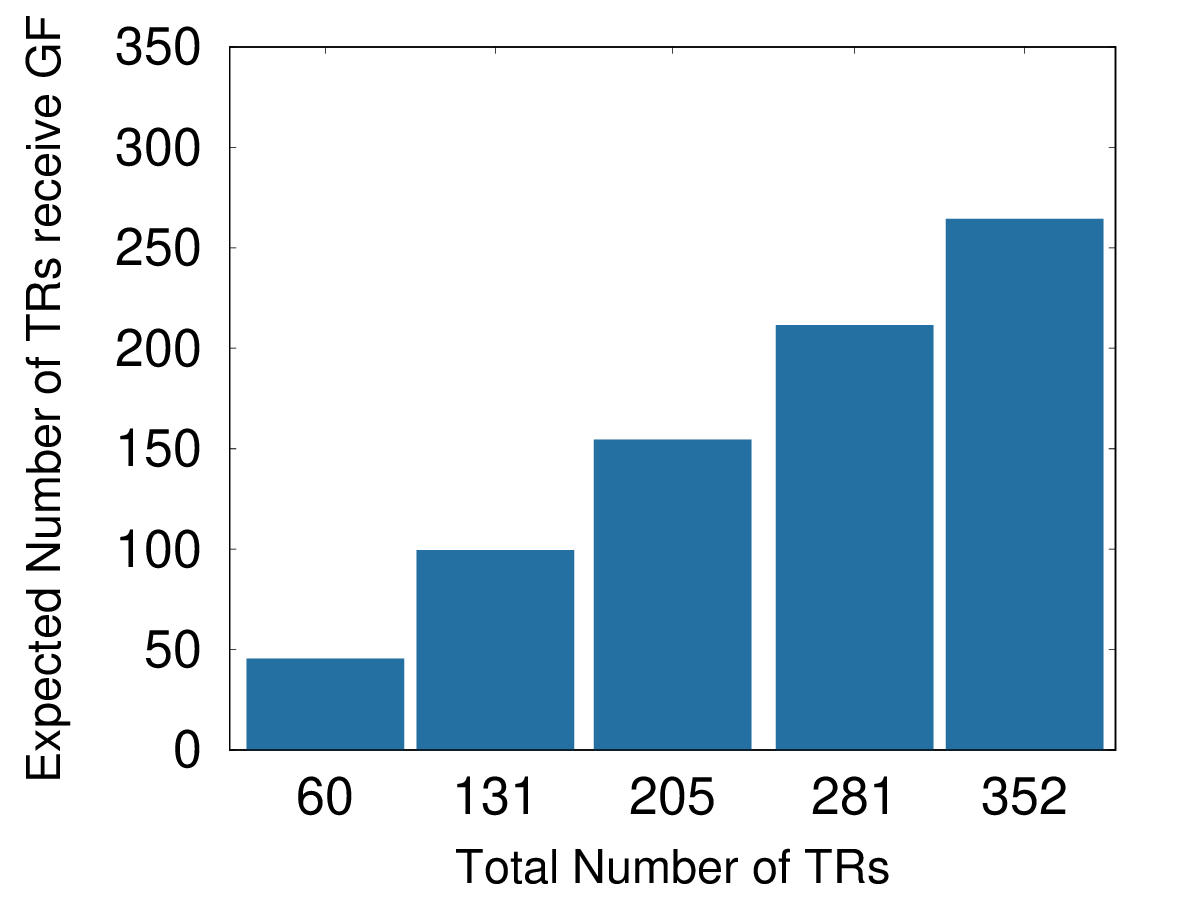}
\subcaption{$Pr\{\mathcal{Z}_i^* = 1\} = 0.75$}
\label{sim:2b4}
\end{subfigure}%
\begin{subfigure}{0.18\textwidth}
\includegraphics[scale = 0.15]{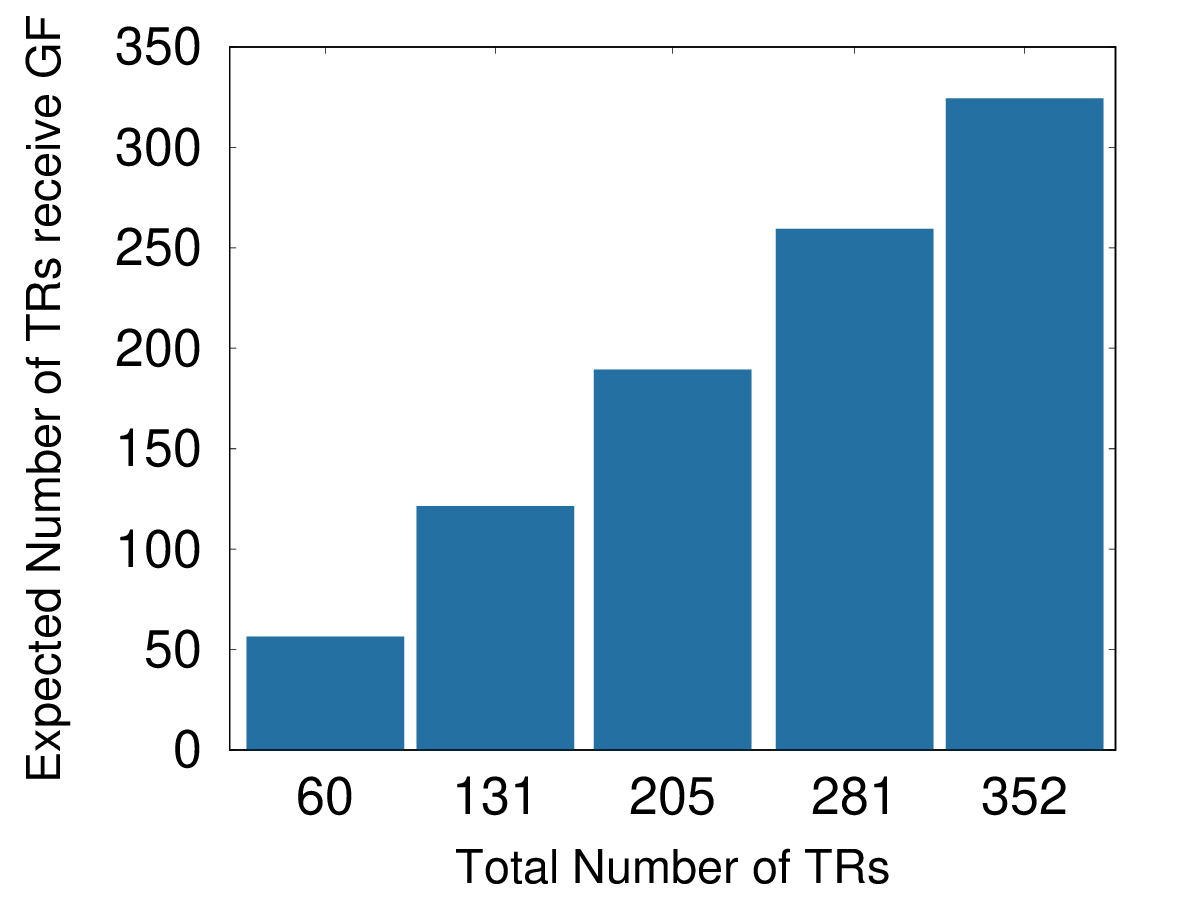}
\subcaption{$Pr\{\mathcal{Z}_i^* = 1\} = 0.92$}
\label{sim:2b5}
\end{subfigure}
        \caption{Comparison of expected number of TRs receive the Government fund based on RTPDD \cite{Som2023}}\label{fig:2}
\end{figure*} 

\subsection{Simulation Set-up}
\label{sec:Sim}
In this section, the simulation set-up utilized in this paper for the two independent tiers are discussed separately. 
\begin{table}[H]
\caption{Synthetic participatory democracy data}
\label{tab: fifth_table}
\centering
\begin{tabular}{c|c|c}
\hline
\textbf{SI. No.} & \textbf{Pr\{$\mathbb{Z}_i^\ast$ = 1\}} & \textbf{No. of TRs}  \\
\hline
 1 & 0.20 & 5, 10, 15, 20, 25 \\
 2 & 0.33 & 5, 10, 15, 20, 25 \\
 3 & 0.50 & 5, 10, 15, 20, 25 \\
 4 & 0.75 & 5, 10, 15, 20, 25 \\
 5 & 0.92 & 5, 10, 15, 20, 25 \\
\hline
\end{tabular}
\end{table}

\begin{itemize}
\item \textbf{Simulation set-up for tier 1 -} For the simulation purpose we have utilized the synthetic participatory democracy data (SPDD) and the real time participatory democracy data (RTPDD) as shown in TABLE \ref{tab: fifth_table}, and TABLES \ref{tab: fourth_table} and \ref{tab: third_table} respectively. As a synthetic data, TABLE \ref{tab: fifth_table} represents two attributes, namely, the probability that any $i^{th}$ task requester receives the Government fund and the number of available task requesters. The simulation is performed for 5 different probability values. For comparison purpose, the available number of task requesters is kept same for 5 different probability values. The real-time participatory data is presented in TABLES \ref{tab: fourth_table} and \ref{tab: third_table}. TABLES \ref{tab: fourth_table} depicts the number of TRs competing in each category (or types of task) for the Government fund. Once the number of competing TRs in each category is determined, then TABLE \ref{tab: third_table} represents the data that we are interested in for the simulation purpose.
\begin{table}[H]
\caption{Number of TRs in each category in real-time participatory democracy data}
\label{tab: fourth_table}
\centering
\begin{tabular}{c|c|c}
\hline
\textbf{SI. No.} & \textbf{Category} & \textbf{No. of TRs}  \\
\hline
 1 & Arts, Culture, and Community & 131 \\
 2 & Education & 352 \\
 3 & Park, Health, and Environment & 281 \\
 4 & Street, Sidewalk, and Safety & 60 \\
 5 & Housing & 205 \\
\hline
\end{tabular}
\end{table}

\begin{table}[H]
\caption{Real-time participatory democracy data}
\label{tab: third_table}
\centering
\begin{tabular}{c|c|c}
\hline
\textbf{SI. No.} & \textbf{Pr\{$\mathbb{Z}_i^\ast$ = 1\}} & \textbf{No. of TRs}  \\
\hline
 1 & 0.20 & 60, 131, 205, 281, 352 \\
 2 & 0.33 & 60, 131, 205, 281, 352 \\
 3 & 0.50 & 60, 131, 205, 281, 352 \\
 4 & 0.75 & 60, 131, 205, 281, 352 \\
 5 & 0.92 & 60, 131, 205, 281, 352 \\
\hline
\end{tabular}
\end{table}

\item \textbf{Simulation set-up for tier 2 -} For the simulation purpose the data sets given in TABLES \ref{fig:tab1} and \ref{fig:tab2} for random distribution (RD) and normal distribution (ND) respectively are utilized. The tables show the \emph{slot number}, \emph{number of TEs}, \emph{bid range}, and \emph{budget}. The budget value mentioned in each row of the tables is the total budget available in respective slots. For RD the bid range is considered in the range 10 to 25. In case of normal distribution the mean ($\mu$) and the standard deviation ($\sigma$) values for the bid range are taken as 17 and 5 and is given in TABLE \ref{fig:tab2}. Each mechanism executes for 10 times on the RD and ND data depicted in TABLES \ref{fig:tab1} and \ref{fig:tab2}. After that an average is taken over 10 rounds and the graphs are plotted. The unit of cost and budget is taken as $\$$. For comparing BULINC with GM, Chen et. al, and Jalaly et. al on the \emph{truthfulness} property it is considered that the subset of TEs are increasing their valuation by $30\%$ of their true value in case of GM and is given as MGM in the simulation figures.     

\begin{table}[H]
\caption{Data set used in tier 2 for RD case}
\label{fig:tab1}
\centering
\begin{tabular}{c|c|c|c}
\hline
\textbf{Slot No.} & \textbf{No. of TEs} & \textbf{Bid range} & \textbf{Budget}  \\
\hline
 1 & 50 & [10, 25] & 134\\
 2 & 45 & [10, 25] & 98\\
 3 & 60 & [10, 25] & 153 \\
 4 & 58 & [10, 25] & 138\\
\hline
\end{tabular}
\end{table}

\begin{table}[H]
\caption{Data set used in tier 2 for ND case}
\label{fig:tab2}
\centering
\begin{tabular}{c|c|c|c}
\hline
\textbf{Slot No.} & \textbf{No. of TEs} & \textbf{Bid range} & \textbf{Budget}  \\
\hline
 1 & 50 & [17, 5] & 120\\
 2 & 45 & [17, 5] & 98\\
 3 & 60 & [17, 5] & 141 \\
 4 & 58 & [17, 5] & 125\\
\hline
\end{tabular}
\end{table}
\end{itemize} 

\subsection{Result Analysis}
\label{sec:RA}
Firstly, the result analysis of tier 1 is carried out, after that the result analysis of tier 2 is done on the basis of parameters mentioned in Section \ref{sec:ef} above. For tier 1, FIGURES \ref{fig:1} and \ref{fig:2} represent the estimate on the number of task requesters getting the Government funds among the available task requesters. The simulation is carried out on both SPDD and RTPDD. In FIGURES \ref{fig:1} and \ref{fig:2}, the $x$-axis of the graphs represent the total number of task requesters and $y$-axis of the graphs represent the expected number of task requesters receiving the Government fund. It can be seen in FIGURE \ref{fig:1} that for a smaller value of Pr$\{Z_i^* =1\}$ (i.e. for  0.20, see FIGURE \ref{sim:1b1}) the number of task requesters receiving the Government fund among the available task requesters is very less. Once the Pr$\{Z_i^* =1\}$ value is increased from 0.20 to 0.33 (see FIGURE \ref{sim:1b2}), the expected number of TRs receiving the GF gets increased. It is due to the reason that the probability that any task requester $\boldsymbol{r}_i$ receives the fund from the Government gets increased. The similar nature of the graphs can be seen for Pr$\{Z_i^* =1\} = 0.50$, Pr$\{Z_i^* =1\} = 0.75$, and Pr$\{Z_i^* =1\} = 0.92$ and the reason is same as above. So, it can be inferred that lower the value of Pr$\{Z_i^* =1\}$ lower will be the number of task requesters receiving the Government fund for their tasks, and higher the value of Pr$\{Z_i^* =1\}$ higher will be the number of task requesters receiving the Government fund for their tasks. Similar nature of the BULINC can be seen in Lemma \ref{the1} and Observation \ref{th:obs} of Section \ref{sec:atppm}. After getting the positive results of BULINC on the synthetic data, we performed the simulation on RTPDD. In case of RTPDD, the number of available TRs is large in number and is shown in TABLE \ref{tab: third_table}. The Pr$\{Z_i^* =1\}$ values are kept same as it was in case of synthetic data. For RTPDD also, with the increase in $Pr\{Z_i = 1\}$ value, the number of TRs receiving the GF gets increased. So, similar nature of BULINC can be seen for the RTPDD and is depicted in FIGURE \ref{fig:2}. The simulation results support the claim made using probabilistic analysis in Section \ref{sec:atppm}.
\begin{figure*}
\begin{subfigure}{0.27\textwidth}
\includegraphics[scale = 0.22]{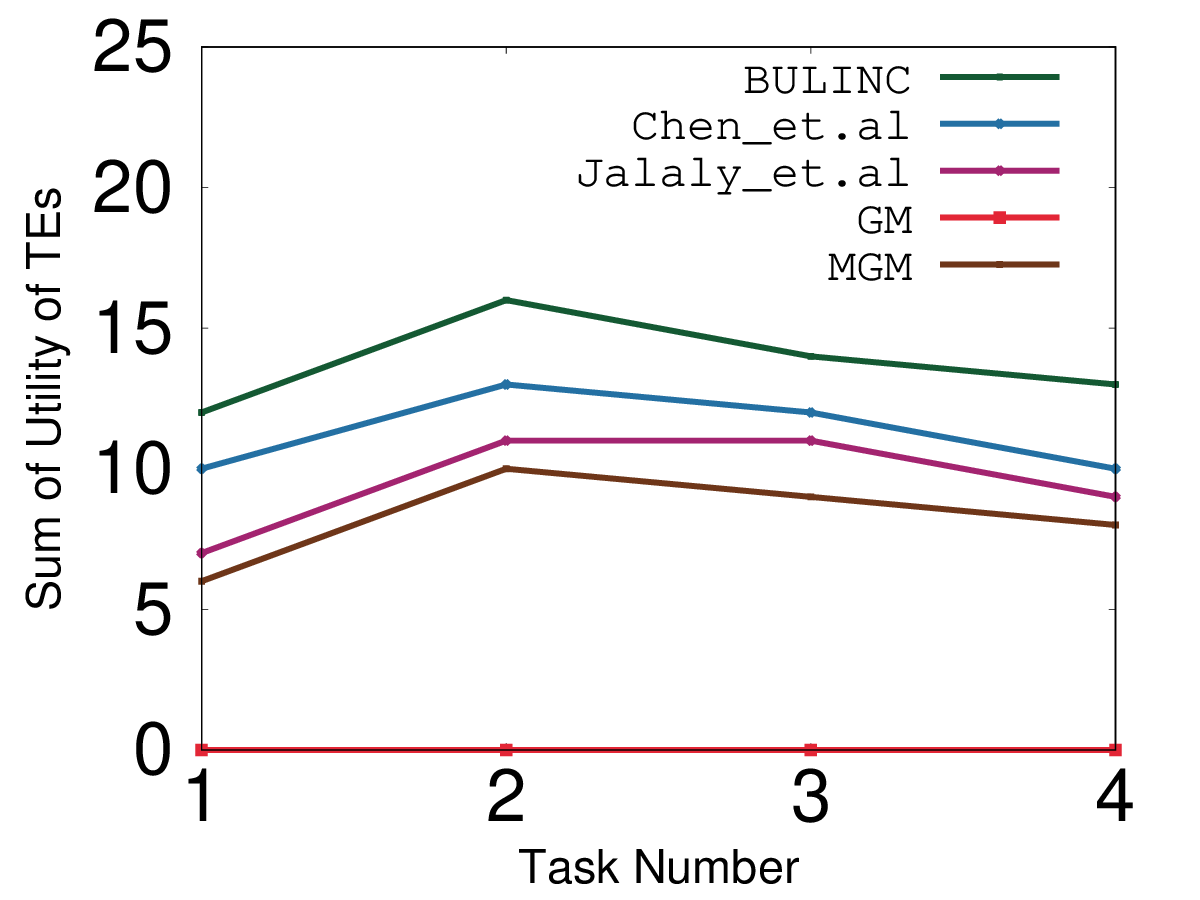}
\subcaption{Utility of TEs in slot 1}
\label{sim:1aslrd}
\end{subfigure}%
\begin{subfigure}{0.27\textwidth}
\includegraphics[scale = 0.22]{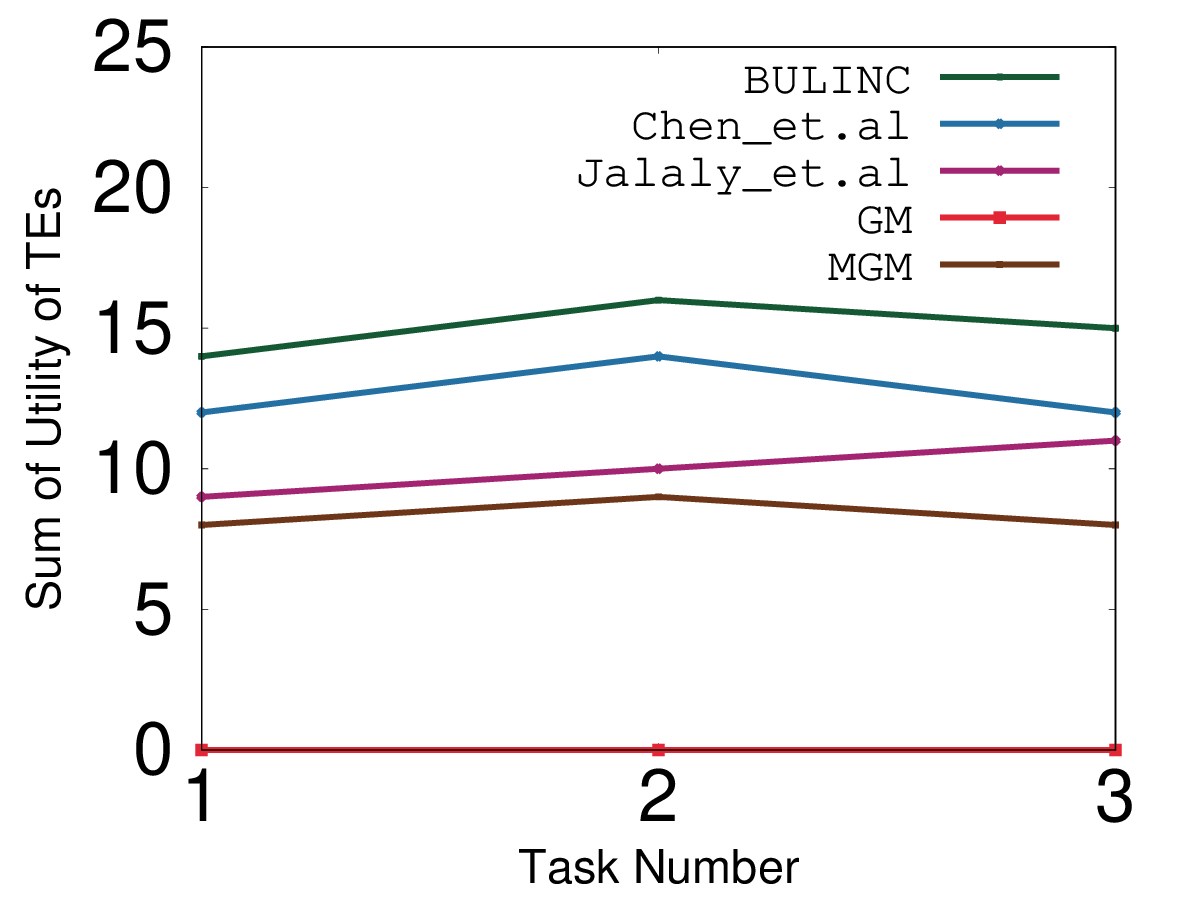}
\subcaption{Utility of TEs in slot 2}
\label{sim:2aslrd}
\end{subfigure}%
\begin{subfigure}{0.27\textwidth}
\includegraphics[scale = 0.22]{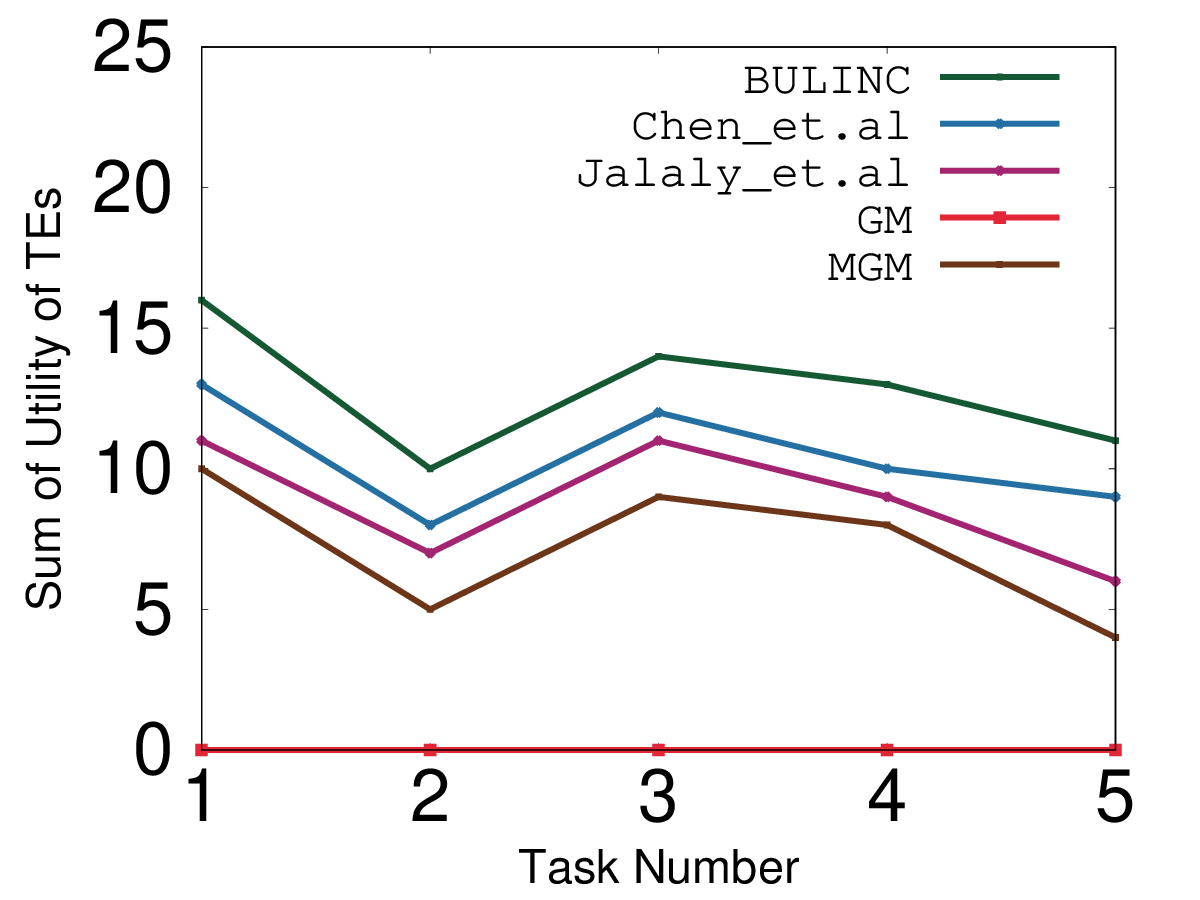}
\subcaption{Utility of TEs in slot 3}
\label{sim:3aslrd}
\end{subfigure}%
\begin{subfigure}{0.27\textwidth}
\includegraphics[scale = 0.22]{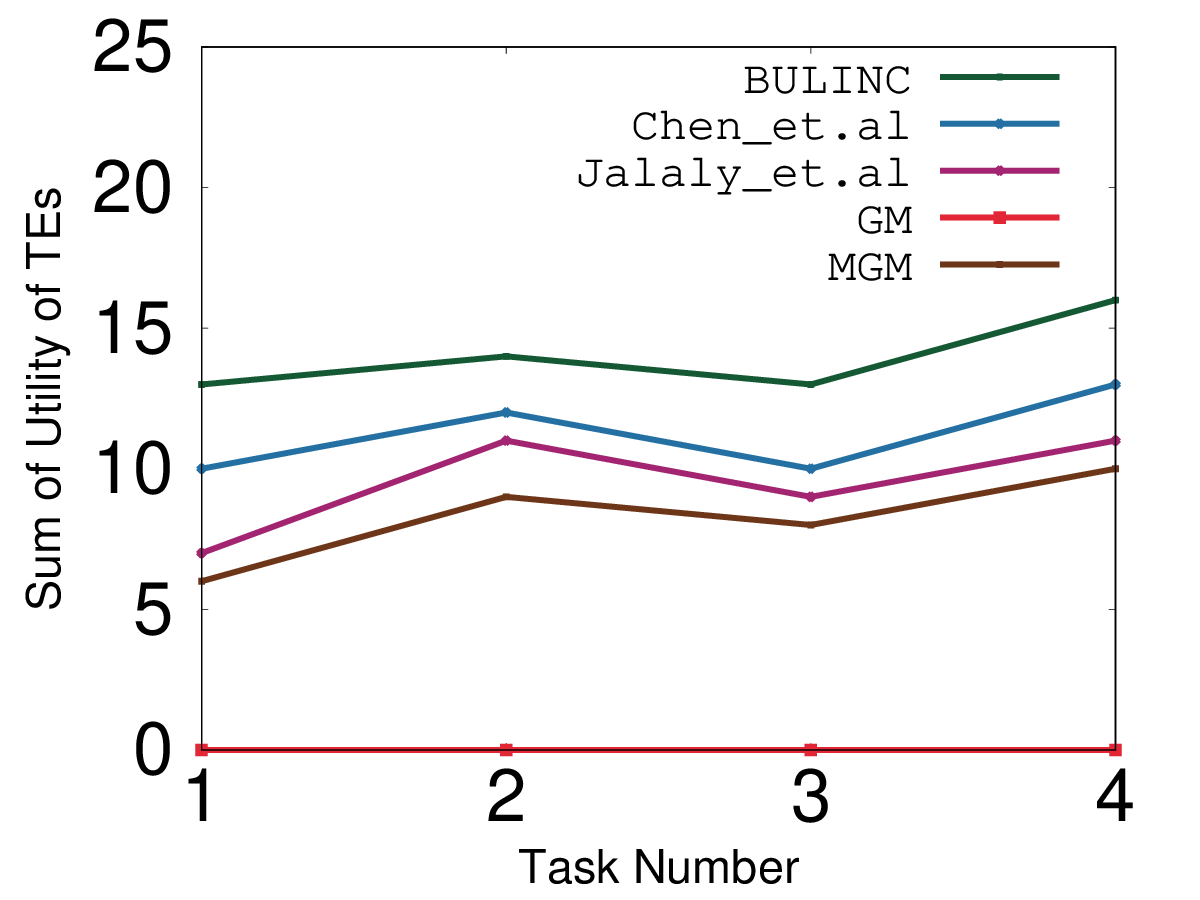}
\subcaption{Utility of TEs in slot 4}
\label{sim:4aslrd}
\end{subfigure}
\caption{Comparison of Utility of TEs for RD case in slot 1 to slot 4 (from left to right)}
\label{Sim:1rd}
\end{figure*}

\begin{figure*}
\begin{subfigure}{0.27\textwidth}
\includegraphics[scale = 0.22]{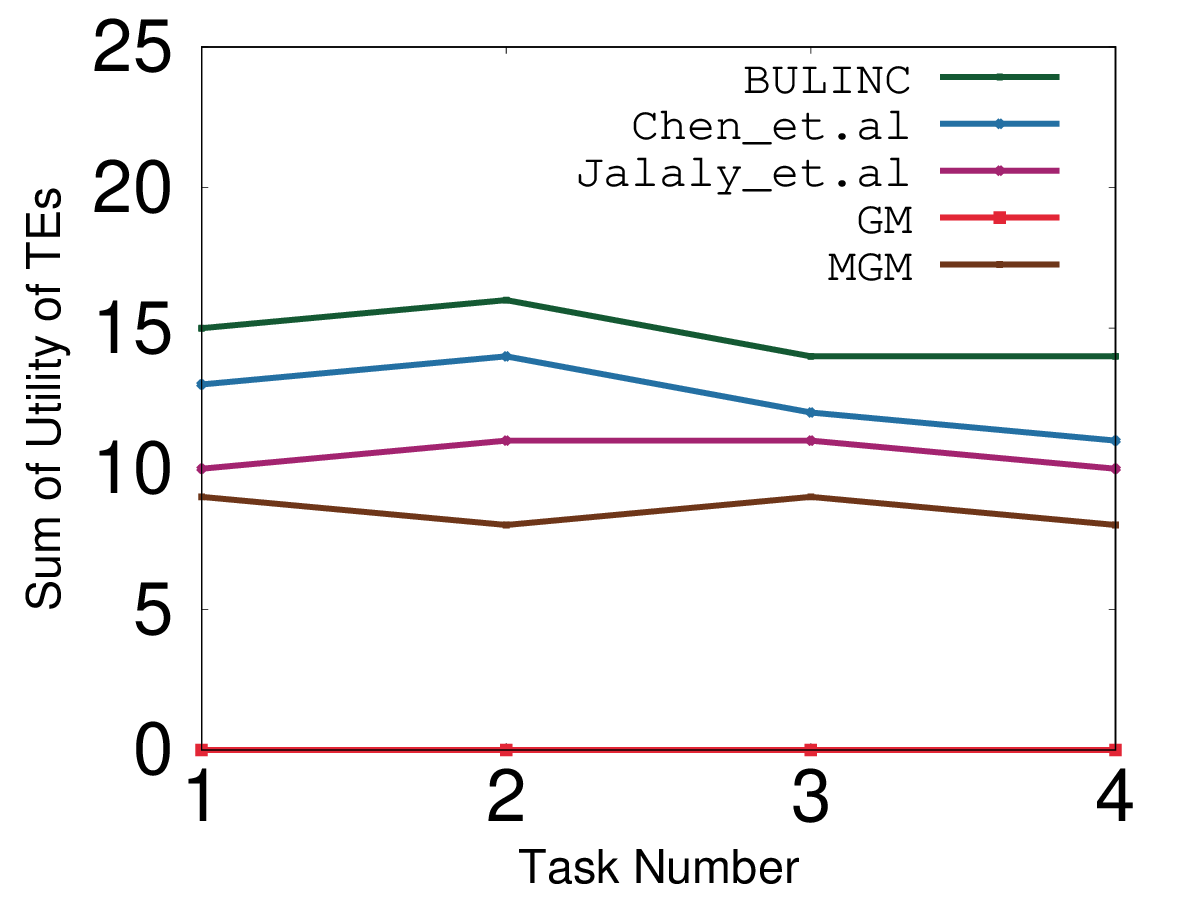}
\subcaption{Utility of TEs in slot 1}
\label{sim:1aslnd}
\end{subfigure}%
\begin{subfigure}{0.27\textwidth}
\includegraphics[scale = 0.22]{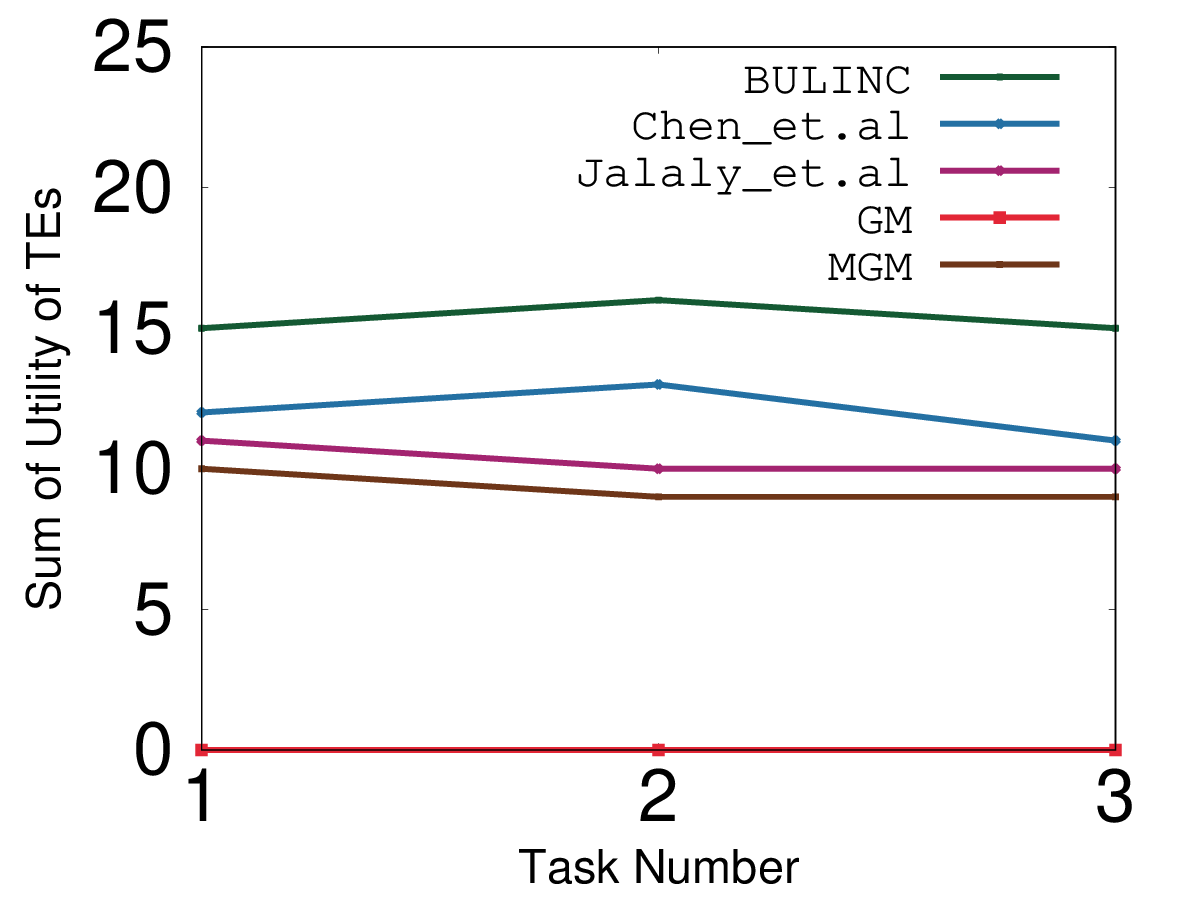}
\subcaption{Utility of TEs in slot 2}
\label{sim:1aslnd}
\end{subfigure}%
\begin{subfigure}{0.27\textwidth}
\includegraphics[scale = 0.22]{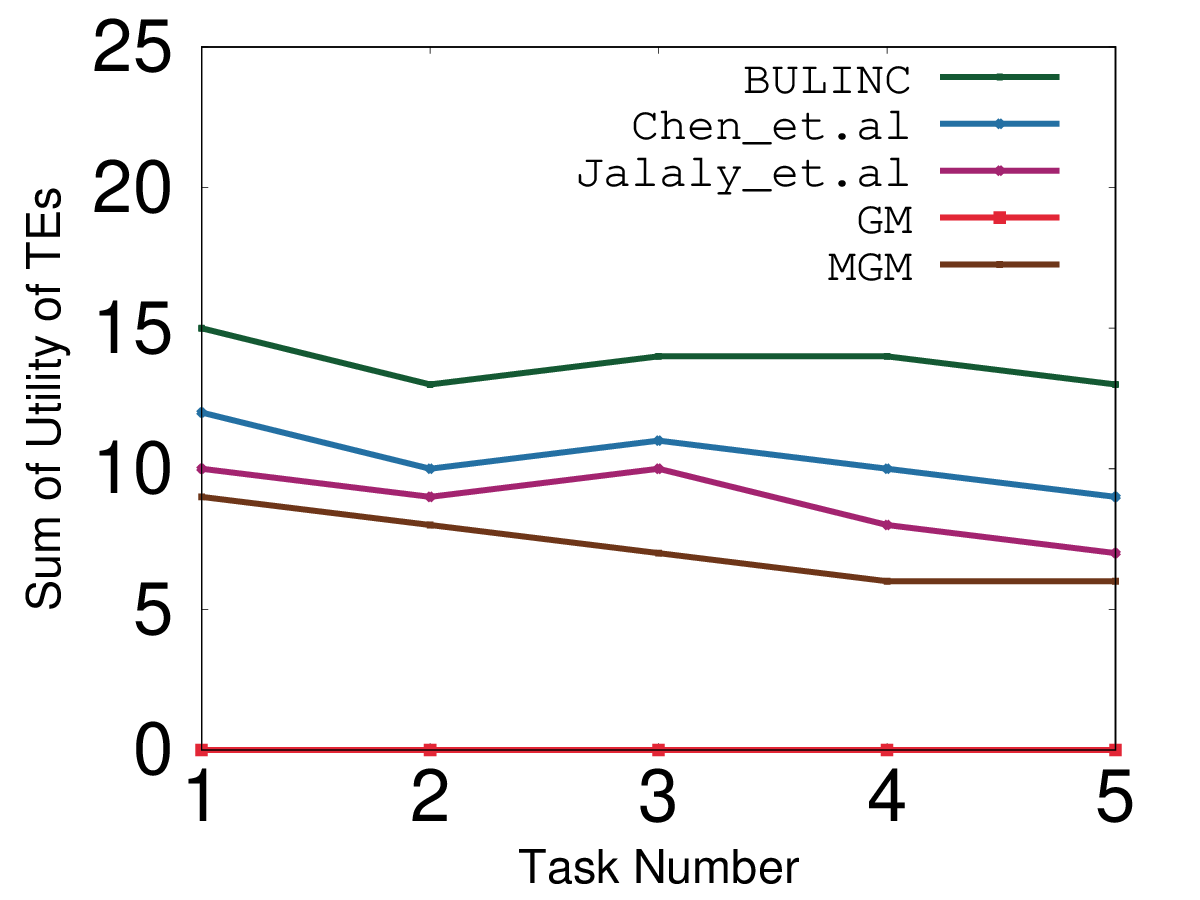}
\subcaption{Utility of TEs in slot 3}
\label{sim:1aslnd}
\end{subfigure}%
\begin{subfigure}{0.27\textwidth}
\includegraphics[scale = 0.22]{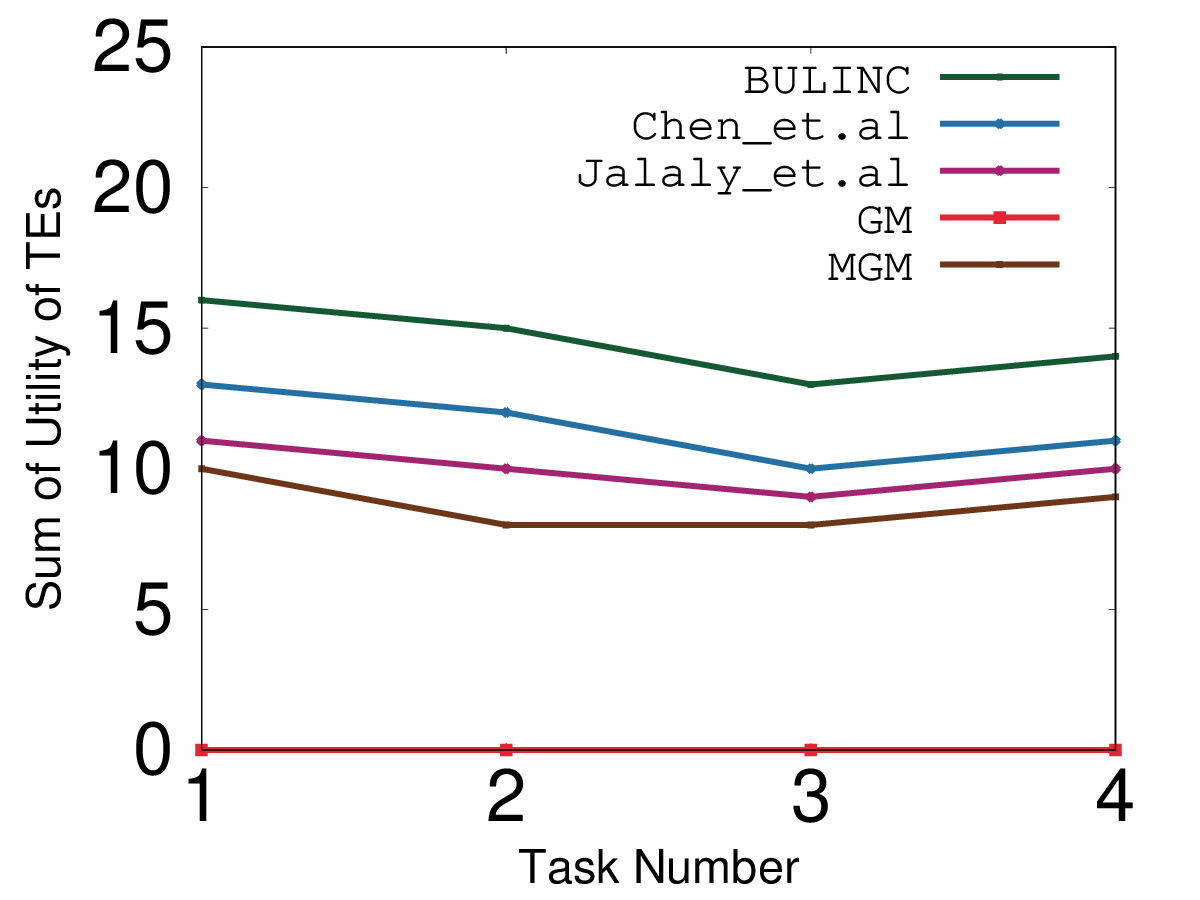}
\subcaption{Utility of TEs in slot 4}
\label{sim:1aslnd}
\end{subfigure}
\caption{Comparison of Utility of TEs for ND case in slot 1 to slot 4 (from left to right)}
\label{Sim:1nd}
\end{figure*}

\noindent In case of tier 2, the simulation is carried out by utilizing the following metrics: (1) \emph{truthfulness}, (2) \emph{budget feasibility}, and (3) \emph{running time} for two different probability distributions. These two probability distributions are RD and ND. The simulation results shown in FIGURES \ref{Sim:1rd} and \ref{Sim:1nd} compares BULINC with the benchmark mechanisms (GM, Chen et. al \cite{Chen3801}, and Jalaly et. al \cite{Jalaly2017SimpleAE}) on the basis of \emph{utility of the task executors}. For the simulation purpose, 4 different slots are considered. It can be seen from FIGURES \ref{Sim:1rd} and \ref{Sim:1nd} that the overall profit of task executors in case of BULINC is large than in case of Chen et. al and is larger than in case of Jalaly et. al and is much higher than in case of GM (in case of GM the utility of all the winning task executors is 0) for both RD and ND cases. The reason for the above mentioned nature of BULINC is that whatever incentive received by the TEs in BULINC that is larger than the incentive received by TEs in Chen et. al and is larger than the incentive received by the TEs in case of Jalaly et. al and is larger than the incentive received by the TEs in case of GM. Further, in case of GM, the utility of TEs is 0 because the incentive received by the task executors is equal as their bid value. However, GM can be manipulated by strategic TEs. It means that the task executors will make some profit if they misreport their privately held true bid value. In order to show the manipulative behavior of GM, in our simulation, for the subset of TEs the bid value is increased by $30\%$ from their true value. It can be seen in FIGURES \ref{Sim:1rd} and \ref{Sim:1nd} that, in case of GM, the utility of tasks executors are higher when they are not reporting in a truthful manner (it is depicted as MGM in the graphs) than in case when they are reporting truthfully. From simulation results, it can be inferred that GM can be manipulated by the strategic TEs in the crowdsensing system.

\begin{figure*}
\begin{subfigure}{0.52\textwidth}
\includegraphics[scale = 0.40]{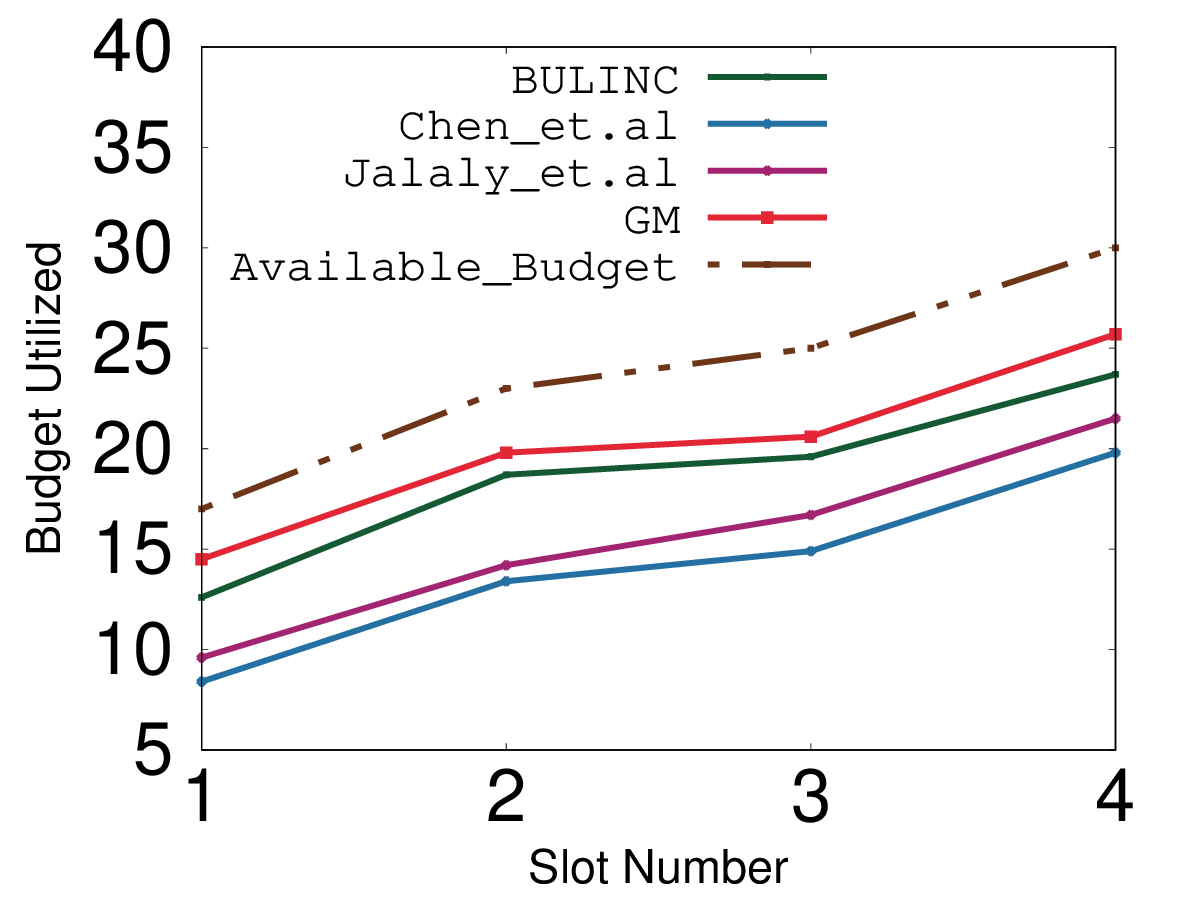}
\subcaption{Budget utilized by BULINC and benchmark mechanisms in RD}
\label{sim:1rda}
\end{subfigure}%
\begin{subfigure}{0.52\textwidth}
\includegraphics[scale = 0.40]{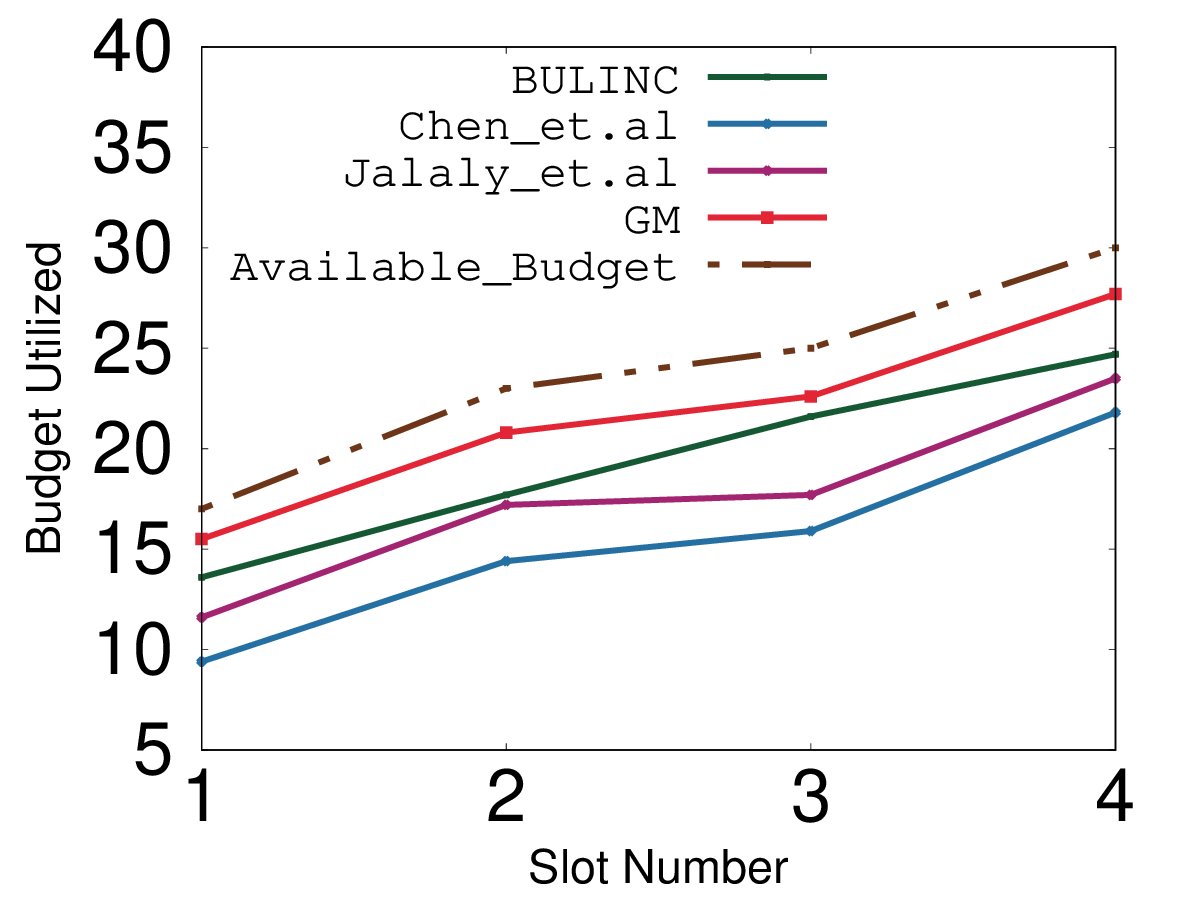}
\subcaption{Budget utilized by BULINC and benchmark mechanisms in ND}
\label{sim:1nda}
\end{subfigure}
\caption{Comparison of budget utilized by BULINC and benchmark mechanisms.}
\label{sim:3}
\end{figure*}
\begin{figure*}
\begin{subfigure}{0.52\textwidth}
\includegraphics[scale = 0.40]{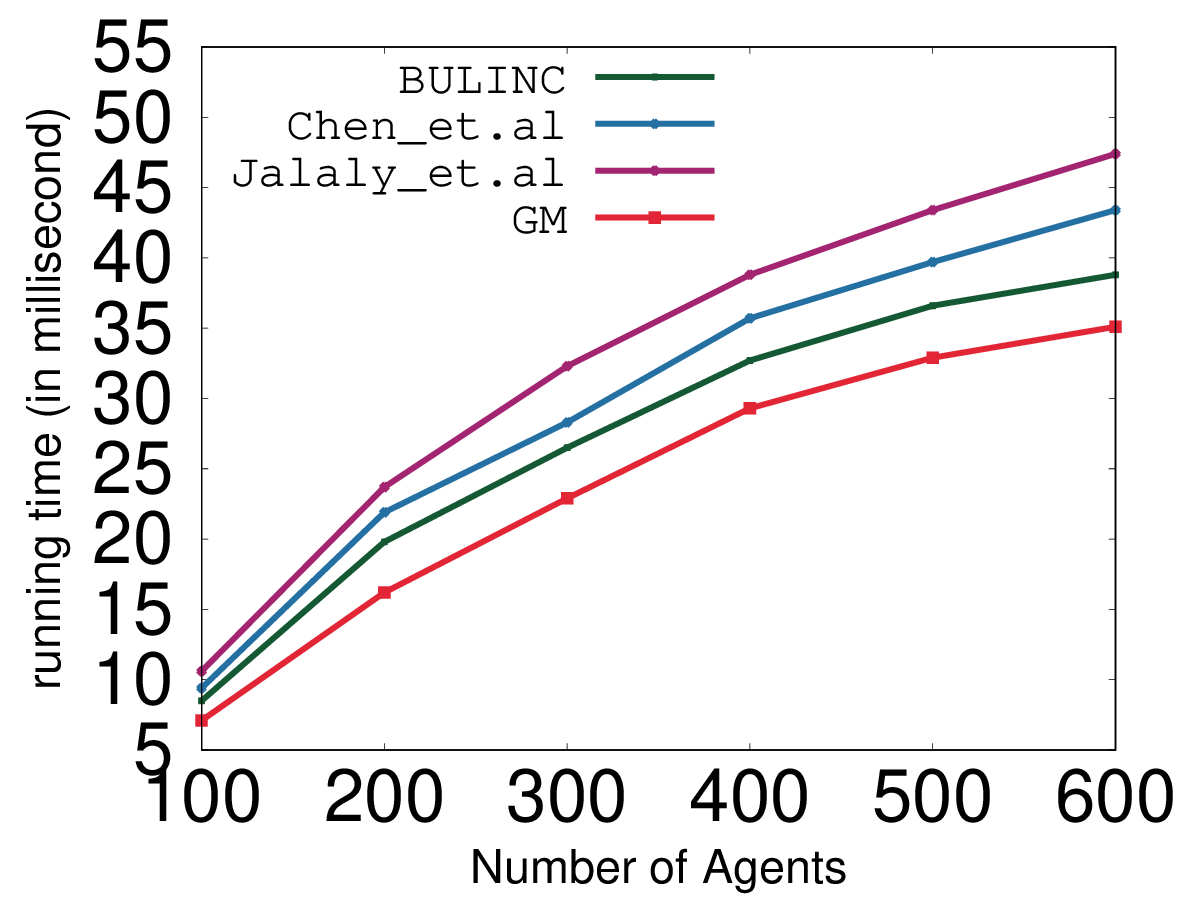}
\subcaption{Running time of BULINC and benchmark mechanisms in RD}
\label{sim:3rda}
\end{subfigure}%
\begin{subfigure}{0.52\textwidth}
\includegraphics[scale = 0.40]{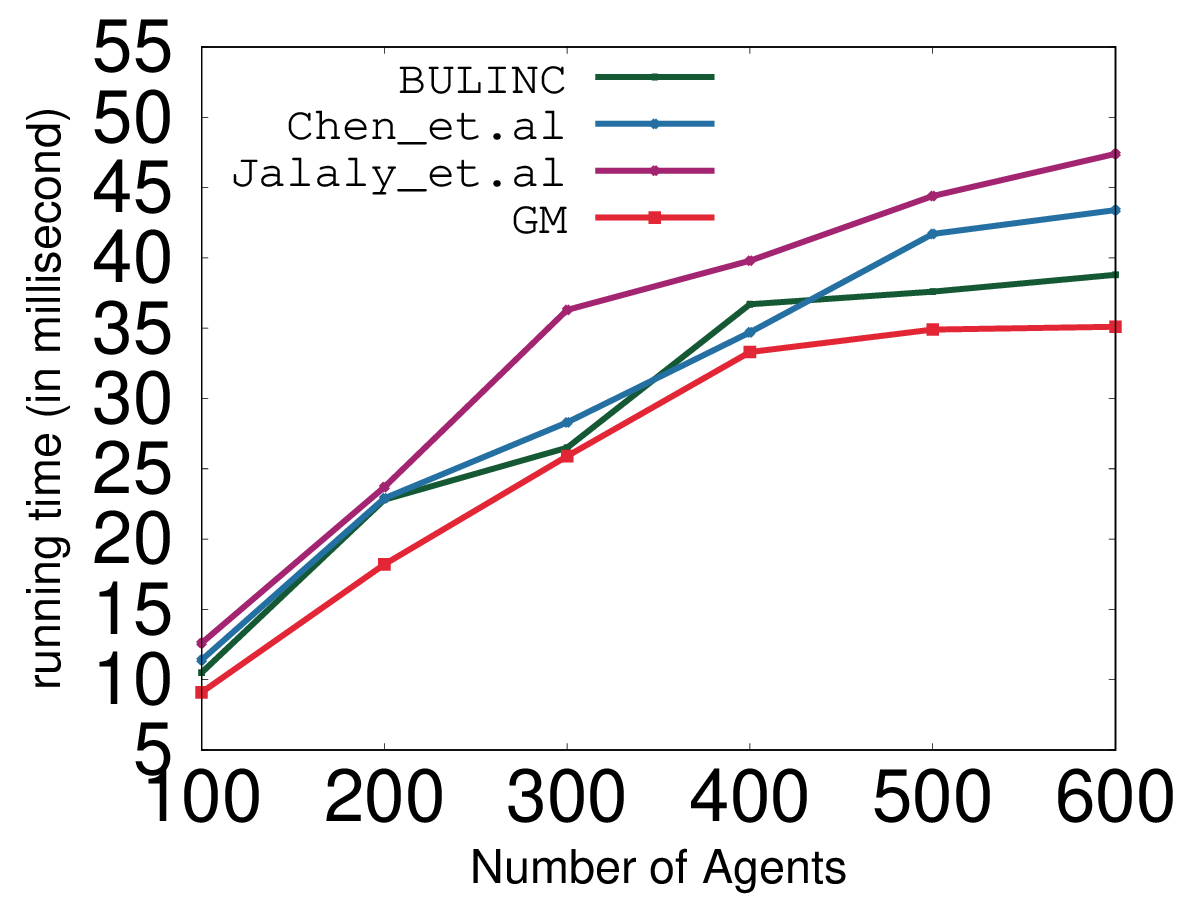}
\subcaption{Running time of BULINC and benchmark mechanisms in ND}
\label{sim:3nda}
\end{subfigure}
\caption{Comparison of running time of BULINC and benchmark mechanisms}
\label{sim:3b}
\end{figure*}
\indent FIGURE \ref{sim:3} tells us about the amount of budget utilized in each of the slots. It is evident from FIGURE \ref{sim:3} that the overall budget utilized as an incentive of TEs in case of BULINC is higher than the overall budget utilized as an incentive of the TEs in case of Jalaly et. al and is higher than the overall budget utilized as an incentive of the TEs in case of Chen et. al. However, the overall budget utilized as the payment of TEs in case of BULINC is lower than the overall budget utilized as an incentive of the TEs in case of GM. The reason behind this is that the incentive given to the TEs in case of BULINC is higher than the incentives provided to the TEs in case of Jalaly et. al and Chen et. al. Also, the overall incentives given to the TEs in case of BULINC and the benchmark mechanisms is less than or equal to the total allotted quota of budget in each slot. So, the BULINC and the benchmark mechanisms are \emph{budget feasible}.\\
\indent Talking in terms of our third metric $i.e.$ \emph{running time} that is depicted in FIGURE \ref{sim:3b} for both RD and ND cases. It is evident from the graphs shown in FIGURE \ref{sim:3b} that the $x$-axis of the graph represents the \emph{number of agents} and $y$-axis of the graph represents the running time in millisecond. Here, the number of agents means the sum of the number of task executors and the number of task requesters available in the crowdsensing market. From simulation results, it is evident that the running time of BULINC is lower than \emph{Chen et. al} and \emph{Jalaly et. al} but is higher than the GM. The reason behind GM having lowest running time is that the \emph{allocation rule} and {payment rule} of GM is taking less time for execution. Further, the reason behind higher running time of \emph{Chen et. al} and \emph{Jalaly et. al} is that the \emph{allocation rule} and {payment rule} of these mechanism takes more time for execution than BULINC and GM. So, BULINC is scalable.

\section{Conclusion and Future Works}
\label{sec13}
In this article, one crowdsensing scenario is investigated as a two-tiered framework in strategic setting. As the task requesters are not having the budget required for completing the tasks and is to be supplied by the Government, for that reason the scenario is investigated as a two tiered process. In the tier 1, it is decided by the city dwellers (or local people) that which of the available task requesters will be receiving the budget from the Government? Once such task requesters are filtered out, in the second tier, firstly the tasks are distributed to multiple slots in order to ensure that the tasks with different time frame is lying in the same slot. After that for each slot the task executors are hired such that the total incentive given to the TEs is at most the allotted quota of budget for the respective tasks.\\
\indent Through analysis it is shown that BULINC is \emph{truthful}, \emph{budget feasible}, and is \emph{individually rational}. Further, through probabilistic analysis it is estimated that the number the number of task requesters receiving the Government fund is $\frac{\mathcal{N}}{\lambda}$. The results obtained from the simulation of the first tier shows an estimate on the number of task requesters that receives the GF from among the available once for both SPDD and RTPDD. The results show that if $Pr\{\mathcal{Z}_i=1\}$ value is high then, the number of task requesters receiving GF will be high but it that is low then, the number of task requesters receiving GF will be low. The simulation results of the second tier shows that BULINC is \emph{truthful}, \emph{budget feasible}, and \emph{computationally efficient}. Further results shows that BULINC and the benchmark mechanisms are budget feasible and scalable.\\
\indent  As a future works, the following possible directions could be explored:
\begin{enumerate}
\item Design an incentive compatible mechanism by extending the set-up discussed in this article with the constraint that the participating executors follow general valuation (not specifically only submodular valuation).
\item Further in the proposed framework, the quality of the task executors are not taken into consideration and can be considered in our future work. So, designing a truthful budget feasible mechanism will be challenging that will also take care of the quality of the task executors.        
\end{enumerate}
\indent To reproduce the results obtained in this paper, the code, the synthetic data, and the real time participatory democracy data will be made available in the following public repository  

\section*{Acknowledgment}
The authors would like to thank the faculties and PhD students of School of Computer Science and Engineering, VIT-AP University, Amaravati, India for providing valuable suggestions and support.

\bibliographystyle{IEEEtran}
\bibliography{bare_jrnl_new_sample4}

\begin{thebibliography}{10}
\providecommand{\url}[1]{#1}
\csname url@samestyle\endcsname
\providecommand{\newblock}{\relax}
\providecommand{\bibinfo}[2]{#2}
\providecommand{\BIBentrySTDinterwordspacing}{\spaceskip=0pt\relax}
\providecommand{\BIBentryALTinterwordstretchfactor}{4}
\providecommand{\BIBentryALTinterwordspacing}{\spaceskip=\fontdimen2\font plus
\BIBentryALTinterwordstretchfactor\fontdimen3\font minus
  \fontdimen4\font\relax}
\providecommand{\BIBforeignlanguage}[2]{{%
\expandafter\ifx\csname l@#1\endcsname\relax
\typeout{** WARNING: IEEEtran.bst: No hyphenation pattern has been}%
\typeout{** loaded for the language `#1'. Using the pattern for}%
\typeout{** the default language instead.}%
\else
\language=\csname l@#1\endcsname
\fi
#2}}
\providecommand{\BIBdecl}{\relax}
\BIBdecl

\bibitem{Howe2006}
\BIBentryALTinterwordspacing
J.~Howe, ``The rise of crowdsourcing,'' \emph{Wired Magazine}, vol.~14, no.~6,
  06 2006. [Online]. Available:
  \url{http://www.wired.com/wired/archive/14.06/crowds.html}
\BIBentrySTDinterwordspacing

\bibitem{Wiki_crowdsourcing}
Wikipedia, ``https://en.wikipedia.org/wiki/crowdsourcing,'' February 2024.

\bibitem{NEVO2020101593}
\BIBentryALTinterwordspacing
D.~Nevo and J.~Kotlarsky, ``Crowdsourcing as a strategic is sourcing
  phenomenon: Critical review and insights for future research,'' \emph{The
  Journal of Strategic Information Systems}, vol.~29, no.~4, p. 101593, 2020,
  2020 Review Issue. [Online]. Available:
  \url{https://www.sciencedirect.com/science/article/pii/S0963868720300019}
\BIBentrySTDinterwordspacing

\bibitem{Singh_2020}
\BIBentryALTinterwordspacing
V.~K. Singh, S.~Mukhopadhyay, F.~Xhafa, and P.~Krause, ``A quality-assuring,
  combinatorial auction based mechanism for {I}o{T}-based crowdsourcing,'' in
  \emph{Advances in Edge Computing: Massive Parallel Processing and
  Applications}.\hskip 1em plus 0.5em minus 0.4em\relax IOS Press, 2020,
  vol.~35, pp. 148--177. [Online]. Available: \url{10.3233/APC200006}
\BIBentrySTDinterwordspacing

\bibitem{JIN2020103351}
\BIBentryALTinterwordspacing
Y.~Jin, M.~Carman, Y.~Zhu, and Y.~Xiang, ``A technical survey on statistical
  modelling and design methods for crowdsourcing quality control,''
  \emph{Artificial Intelligence}, vol. 287, p. 103351, 2020. [Online].
  Available:
  \url{https://www.sciencedirect.com/science/article/pii/S0004370218307343}
\BIBentrySTDinterwordspacing

\bibitem{fi14020049}
\BIBentryALTinterwordspacing
K.~L.~M. Ang, J.~K.~P. Seng, and E.~Ngharamike, ``Towards crowdsourcing
  internet of things (crowd-iot): Architectures, security and applications,''
  \emph{Future Internet, \emph{MDPI}}, vol.~14, no.~2, pp. 1--50, 2022.
  [Online]. Available: \url{https://www.mdpi.com/1999-5903/14/2/49}
\BIBentrySTDinterwordspacing

\bibitem{10443612}
\BIBentryALTinterwordspacing
X.~Liu, H.~Chen, Y.~Liu, W.~Wei, H.~Xue, and F.~Xia, ``Multi-task data
  collection with limited budget in edge-assisted mobile crowdsensing,''
  \emph{IEEE Internet of Things Journal}, pp. 1--1, 2024. [Online]. Available:
  \url{10.1109/JIOT.2024.3364239}
\BIBentrySTDinterwordspacing

\bibitem{8733838}
\BIBentryALTinterwordspacing
X.~Kong, X.~Liu, B.~Jedari, M.~Li, L.~Wan, and F.~Xia, ``Mobile crowdsourcing
  in smart cities: Technologies, applications, and future challenges,''
  \emph{IEEE Internet of Things Journal}, vol.~6, no.~5, pp. 8095--8113, 2019.
  [Online]. Available: \url{10.1109/JIOT.2019.2921879}
\BIBentrySTDinterwordspacing

\bibitem{Phuttharak2019ARO}
\BIBentryALTinterwordspacing
J.~Phuttharak and S.~W. Loke, ``A review of mobile crowdsourcing architectures
  and challenges: Toward crowd-empowered internet-of-things,'' \emph{IEEE
  Access}, vol.~7, pp. 304--324, 2019. [Online]. Available:
  \url{10.1109/ACCESS.2018.2885353}
\BIBentrySTDinterwordspacing

\bibitem{s20072055}
\BIBentryALTinterwordspacing
V.~S. Dasari, B.~Kantarci, M.~Pouryazdan, L.~Foschini, and M.~Girolami, ``Game
  theory in mobile crowdsensing: A comprehensive survey,'' \emph{Sensors},
  vol.~20, no.~7, p. 2055, 2020. [Online]. Available:
  \url{https://www.mdpi.com/1424-8220/20/7/2055}
\BIBentrySTDinterwordspacing

\bibitem{https://doi.org/10.48550/arxiv.2203.06647}
\BIBentryALTinterwordspacing
V.~K. Singh, A.~S. Jasti, S.~K. Singh, S.~Mishra, and A.~Alkhayyat, ``A quality
  aware multiunit double auction framework for iot-based mobile crowdsensing in
  strategic setting,'' \emph{IEEE Access}, vol.~10, pp. 67\,976--67\,999, 2022.
  [Online]. Available: \url{10.1109/ACCESS.2022.3186095}
\BIBentrySTDinterwordspacing

\bibitem{doi:10.1177/2399808320987567}
\BIBentryALTinterwordspacing
A.~Pődör and S.~Szabó, ``Geo-tagged environmental noise measurement with
  smartphones: Accuracy and perspectives of crowdsourced mapping,''
  \emph{Environment and Planning B: Urban Analytics and City Science}, vol.~48,
  no.~9, pp. 2710--2725, 2021. [Online]. Available:
  \url{https://doi.org/10.1177/2399808320987567}
\BIBentrySTDinterwordspacing

\bibitem{STANIEK2021554}
\BIBentryALTinterwordspacing
M.~Staniek, ``Road pavement condition diagnostics using smartphone-based data
  crowdsourcing in smart cities,'' \emph{Journal of Traffic and Transportation
  Engineering (English Edition)}, vol.~8, no.~4, pp. 554--567, 2021. [Online].
  Available:
  \url{https://www.sciencedirect.com/science/article/pii/S2095756421000192}
\BIBentrySTDinterwordspacing

\bibitem{s20195564}
\BIBentryALTinterwordspacing
C.~Wu, Z.~Wang, S.~Hu, J.~Lepine, X.~Na, D.~Ainalis, and M.~Stettler, ``An
  automated machine-learning approach for road pothole detection using
  smartphone sensor data,'' \emph{Sensors}, vol.~20, no.~19, pp. 1--23, 2020.
  [Online]. Available: \url{https://www.mdpi.com/1424-8220/20/19/5564}
\BIBentrySTDinterwordspacing

\bibitem{10356055}
\BIBentryALTinterwordspacing
E.~Zhang, R.~Trujillo, J.~M. Templeton, and C.~Poellabauer, ``A study on mobile
  crowd sensing systems for healthcare scenarios,'' \emph{IEEE Access},
  vol.~11, pp. 140\,325--140\,347, 2023. [Online]. Available:
  \url{10.1109/ACCESS.2023.3342158}
\BIBentrySTDinterwordspacing

\bibitem{Covid-191}
\BIBentryALTinterwordspacing
K.~B. Ramadi and F.~Nguyen, ``Rapid crowdsourced innovation for covid-19
  response and economic growth,'' \emph{NPJ Digital Medicine}, vol.~4, no.~1,
  pp. 1--5, 2021. [Online]. Available: \url{10.1038/s41746-021-00397-5}
\BIBentrySTDinterwordspacing

\bibitem{8432319}
\BIBentryALTinterwordspacing
V.~K. Singh, S.~Mukhopadhyay, F.~Xhafa, and A.~Sharma, ``A budget feasible
  mechanism for hiring doctors in e-healthcare,'' in \emph{2018 IEEE 32nd
  International Conference on Advanced Information Networking and Applications
  (AINA)}, Krakow, Poland, May 2018, pp. 785--792. [Online]. Available:
  \url{10.1109/AINA.2018.00117}
\BIBentrySTDinterwordspacing

\bibitem{Giannetsos20111295}
\BIBentryALTinterwordspacing
T.~Giannetsos, T.~Dimitriou, and N.~R. Prasad, ``People-centric sensing in
  assistive healthcare: Privacy challenges and directions,'' \emph{Sec. and
  Commun. Netw.}, vol.~4, no.~11, pp. 1295--1307, Nov. 2011. [Online].
  Available: \url{DOI: 10.1002/sec}
\BIBentrySTDinterwordspacing

\bibitem{DBLP:journals/corr/SinghM16}
\BIBentryALTinterwordspacing
V.~K. Singh, S.~Mukhopadhyay, and R.~Das, ``Hiring doctors in e-healthcare
  with zero budget,'' in \emph{Advances on P2P, Parallel, Grid, Cloud and
  Internet Computing}, F.~Xhafa, S.~Caball{\'e}, and L.~Barolli, Eds.\hskip 1em
  plus 0.5em minus 0.4em\relax Cham: Springer International Publishing, 2018,
  pp. 379--390. [Online]. Available:
  \url{https://doi.org/10.1007/978-3-319-69835-9_36'}
\BIBentrySTDinterwordspacing

\bibitem{Nagatani:2013:ERN:2421033.2421037}
\BIBentryALTinterwordspacing
K.~Nagatani, S.~Kiribayashi, Y.~Okada, K.~Otake, K.~Yoshida, S.~Tadokoro,
  T.~Nishimura, T.~Yoshida, E.~Koyanagi, M.~Fukushima, and S.~Kawatsuma,
  ``Emergency response to the nuclear accident at the fukushima daiichi nuclear
  power plants using mobile rescue robots,'' \emph{Journal of Field Robotics},
  vol.~30, no.~1, pp. 44--63, 2013. [Online]. Available:
  \url{https://doi.org/10.1002/rob.21439}
\BIBentrySTDinterwordspacing

\bibitem{Poblet2014}
\BIBentryALTinterwordspacing
M.~Poblet, E.~Garc{\'i}a-Cuesta, and P.~Casanovas, ``Crowdsourcing tools for
  disaster management: A review of platforms and methods,'' in \emph{AI
  Approaches to the Complexity of Legal Systems}, P.~Casanovas, U.~Pagallo,
  M.~Palmirani, and G.~Sartor, Eds.\hskip 1em plus 0.5em minus 0.4em\relax
  Berlin, Heidelberg: Springer Berlin Heidelberg, 2014, pp. 261--274. [Online].
  Available: \url{https://doi.org/10.1007/978-3-662-45960-7_19}
\BIBentrySTDinterwordspacing

\bibitem{Covid-19}
\BIBentryALTinterwordspacing
S.~Vermicelli, L.~Cricelli, and M.~Grimaldi, ``How can crowdsourcing help
  tackle the covid‐19 pandemic? an explorative overview of innovative
  collaborative practices.'' \emph{R\&D Management}, vol.~51, no.~2, pp.
  183--194, November 2020. [Online]. Available:
  \url{https://doi.org/10.1111/radm.12443}
\BIBentrySTDinterwordspacing

\bibitem{Mukhopadhyay2021}
\BIBentryALTinterwordspacing
J.~Mukhopadhyay, V.~K. Singh, A.~Pal, and A.~Kumar, ``A truthful budget
  feasible mechanism for iot-based participatory sensing with incremental
  arrival of budget,'' \emph{Journal of Ambient Intelligence and Humanized
  Computing}, vol.~13, pp. 1107--1124, Feb 2022. [Online]. Available:
  \url{https://doi.org/10.1007/s12652-020-02844-9}
\BIBentrySTDinterwordspacing

\bibitem{Singh2019}
\BIBentryALTinterwordspacing
V.~K. Singh, S.~Mukhopadhyay, F.~Xhafa, and A.~Sharma, ``A budget feasible peer
  graded mechanism for iot-based crowdsourcing,'' \emph{Journal of Ambient
  Intelligence and Humanized Computing}, vol.~11, no.~4, pp. 1531--1551, Jan
  2020. [Online]. Available: \url{https://doi.org/10.1007/s12652-019-01219-z}
\BIBentrySTDinterwordspacing

\bibitem{DBLP:conf/hcomp/GoelNS14}
\BIBentryALTinterwordspacing
G.~Goel, A.~Nikzad, and A.~Singla, ``Mechanism design for crowdsourcing markets
  with heterogeneous tasks,'' in \emph{Proceedings of the Second {AAAI}
  Conference on Human Computation and Crowdsourcing, {HCOMP} 2014, November
  2-4, 2014, Pittsburgh, Pennsylvania, {USA}}, vol.~2, 2014, pp. 77--86.
  [Online]. Available: \url{https://doi.org/10.1609/hcomp.v2i1.13158}
\BIBentrySTDinterwordspacing

\bibitem{10.1177/0956247817746279}
\BIBentryALTinterwordspacing
Y.~Cabannes and B.~Lipietz, ``Revisiting the democratic promise of
  participatory budgeting in light of competing political, good governance and
  technocratic logics,'' \emph{Environment and Urbanization}, vol.~30, no.~1,
  pp. 67--84, 2018. [Online]. Available:
  \url{https://doi.org/10.1177/0956247817746279}
\BIBentrySTDinterwordspacing

\bibitem{10.1177/0020852320943668}
\BIBentryALTinterwordspacing
W.~No and L.~Hsueh, ``How a participatory process with inclusive structural
  design allocates resources toward poor neighborhoods: the case of
  participatory budgeting in seoul, south korea,'' \emph{International Review
  of Administrative Sciences}, vol.~88, no.~3, pp. 663--681, 2020. [Online].
  Available: \url{https://doi.org/10.1177/0020852320943668}
\BIBentrySTDinterwordspacing

\bibitem{10.1177/0020852321991208}
\BIBentryALTinterwordspacing
S.-M. Jung, ``Participatory budgeting and government efficiency: evidence from
  municipal governments in south korea,'' \emph{International Review of
  Administrative Sciences}, vol.~88, no.~4, pp. 1105--1123, 2021. [Online].
  Available: \url{https://doi.org/10.1177/0020852321991208}
\BIBentrySTDinterwordspacing

\bibitem{pateman_2012}
\BIBentryALTinterwordspacing
C.~Pateman, ``Participatory democracy revisited,'' \emph{Perspectives on
  Politics}, vol.~10, no.~1, p. 7–19, 2012. [Online]. Available:
  \url{10.1017/S1537592711004877}
\BIBentrySTDinterwordspacing

\bibitem{T.roughgarden_20164}
\BIBentryALTinterwordspacing
T.~Roughgarden, ``C{S}269{I}: Incentives in computer science ({S}tanford
  {U}niversity {L}ecture {N}otes),'' 2016, lecture 4: Voting, Machine Learning,
  and Participatory Democracy. [Online]. Available:
  \url{https://timroughgarden.org/f16/l/l4.pdf}
\BIBentrySTDinterwordspacing

\bibitem{A.Goel}
\BIBentryALTinterwordspacing
A.~Goel, A.~K. Krishnaswamy, S.~Sakshuwong, and T.~Aitamurto, ``Knapsack voting
  for participatory budgeting,'' \emph{ACM Trans. Econ. Comput.}, vol.~7,
  no.~2, jul 2019. [Online]. Available: \url{https://doi.org/10.1145/3340230}
\BIBentrySTDinterwordspacing

\bibitem{Smith:2009:DID:1795662}
\BIBentryALTinterwordspacing
G.~Smith and G.~Heidegger, ``Book reviews:,'' \emph{Theoria}, vol.~58, no. 126,
  pp. 109 -- 115, 2011. [Online]. Available:
  \url{https://www.berghahnjournals.com/view/journals/theoria/58/126/th5812606.xml}
\BIBentrySTDinterwordspacing

\bibitem{10.1177/00208523221078938}
\BIBentryALTinterwordspacing
L.~Bartocci, G.~Grossi, S.~G. Mauro, and C.~Ebdon, ``The journey of
  participatory budgeting: a systematic literature review and future research
  directions,'' \emph{International Review of Administrative Sciences},
  vol.~89, no.~3, pp. 757--774, 2023. [Online]. Available:
  \url{https://doi.org/10.1177/00208523221078938}
\BIBentrySTDinterwordspacing

\bibitem{pbp}
PBP, ``where has it worked-the participatory budgeting project,'' 2016,
  https://www.participatorybudgeting.org/pb-map/.

\bibitem{Aaron}
A.~Schneider and B.~Goldfrank, \emph{Budgets and ballots in Brazil:
  participatory budgeting from the city to the state}.\hskip 1em plus 0.5em
  minus 0.4em\relax IDS, 2002, no. 149.

\bibitem{8031314}
\BIBentryALTinterwordspacing
Z.~Duan, L.~Tian, M.~Yan, Z.~Cai, Q.~Han, and G.~Yin, ``Practical incentive
  mechanisms for iot-based mobile crowdsensing systems,'' \emph{IEEE Access},
  vol.~5, pp. 20\,383--20\,392, 2017. [Online]. Available:
  \url{10.1109/ACCESS.2017.2751304}
\BIBentrySTDinterwordspacing

\bibitem{6848055}
\BIBentryALTinterwordspacing
Z.~Feng, Y.~Zhu, Q.~Zhang, L.~M. Ni, and A.~V. Vasilakos, ``{TRAC}: Truthful
  auction for location-aware collaborative sensing in mobile crowdsourcing,''
  in \emph{IEEE INFOCOM 2014 - IEEE Conference on Computer Communications},
  Toronto, ON, Canada, April 2014, pp. 1231--1239. [Online]. Available:
  \url{10.1109/INFOCOM.2014.6848055}
\BIBentrySTDinterwordspacing

\bibitem{DBLP:journals/tpds/WangGCG18}
\BIBentryALTinterwordspacing
H.~Wang, S.~Guo, J.~Cao, and M.~Guo, ``Melody: {A} long-term dynamic
  quality-aware incentive mechanism for crowdsourcing,'' \emph{{IEEE} Trans.
  Parallel Distrib. Syst.}, vol.~29, no.~4, pp. 901--914, 2018. [Online].
  Available: \url{https://doi.org/10.1109/TPDS.2017.2775232}
\BIBentrySTDinterwordspacing

\bibitem{8667429}
\BIBentryALTinterwordspacing
D.~Yu, Z.~Zhou, and Y.~Wang, ``Crowdsourcing software task assignment method
  for collaborative development,'' \emph{IEEE Access}, vol.~7, pp.
  35\,743--35\,754, 2019. [Online]. Available:
  \url{10.1109/ACCESS.2019.2905054}
\BIBentrySTDinterwordspacing

\bibitem{9877885}
\BIBentryALTinterwordspacing
H.~Gao, J.~An, C.~Zhou, and L.~Li, ``Quality-aware incentive mechanism for
  social mobile crowd sensing,'' \emph{IEEE Communications Letters}, vol.~27,
  no.~1, pp. 263--267, 2023. [Online]. Available:
  \url{10.1109/LCOMM.2022.3204348}
\BIBentrySTDinterwordspacing

\bibitem{9877939}
\BIBentryALTinterwordspacing
X.~Liu, C.~Fu, W.~Wu, M.~Li, W.~Wang, V.~Chau, and J.~Luo, ``Budget-feasible
  mechanisms in two-sided crowdsensing markets: Truthfulness, fairness, and
  efficiency,'' \emph{IEEE Transactions on Mobile Computing}, vol.~22, no.~12,
  pp. 6938--6955, 2023. [Online]. Available: \url{10.1109/TMC.2022.3201260}
\BIBentrySTDinterwordspacing

\bibitem{9992184}
\BIBentryALTinterwordspacing
Y.~Zhou, F.~Tong, and S.~He, ``Bi-objective incentive mechanism for mobile
  crowdsensing with budget/cost constraint,'' \emph{IEEE Transactions on Mobile
  Computing}, vol.~23, no.~1, pp. 223--237, 2024. [Online]. Available:
  \url{10.1109/TMC.2022.3229470}
\BIBentrySTDinterwordspacing

\bibitem{10.1504/ijwgs.2021.116536}
\BIBentryALTinterwordspacing
J.~Mukhopadhyay, V.~K. Singh, S.~Mukhopadhyay, M.~M. Dhananjaya, and A.~Pal,
  ``An egalitarian approach of scheduling time restricted tasks in mobile
  crowdsourcing for double auction environment,'' \emph{Int. J. Web Grid
  Serv.}, vol.~17, no.~3, p. 221–245, jan 2021. [Online]. Available:
  \url{https://doi.org/10.1504/ijwgs.2021.116536}
\BIBentrySTDinterwordspacing

\bibitem{SUHAG2023102952}
\BIBentryALTinterwordspacing
D.~Suhag and V.~Jha, ``A comprehensive survey on mobile crowdsensing systems,''
  \emph{Journal of Systems Architecture}, vol. 142, p. 102952, 2023. [Online].
  Available:
  \url{https://www.sciencedirect.com/science/article/pii/S1383762123001315}
\BIBentrySTDinterwordspacing

\bibitem{e1}
\BIBentryALTinterwordspacing
H.~Gao and C.~H. Liu, ``A survey of incentive mechanisms for participatory
  sensing,'' \emph{IEEE Communications Surveys \& Tutorials}, vol.~17, no.~2,
  pp. 918--943, 2015. [Online]. Available: \url{DOI:
  10.1109/COMST.2014.2387836}
\BIBentrySTDinterwordspacing

\bibitem{9701340}
\BIBentryALTinterwordspacing
M.~Zulfiqar, M.~N. Malik, and H.~H. Khan, ``Microtasking activities in
  crowdsourced software development: A systematic literature review,''
  \emph{IEEE Access}, vol.~10, pp. 24\,721--24\,737, 2022. [Online]. Available:
  \url{10.1109/ACCESS.2022.3148400}
\BIBentrySTDinterwordspacing

\bibitem{10.1145/3185504}
\BIBentryALTinterwordspacing
J.~Liu, H.~Shen, H.~S. Narman, W.~Chung, and Z.~Lin, ``A survey of mobile
  crowdsensing techniques: A critical component for the internet of things,''
  \emph{ACM Trans. Cyber-Phys. Syst.}, vol.~2, no.~3, pp. 1--26, Jun. 2018.
  [Online]. Available: \url{https://doi.org/10.1145/3185504}
\BIBentrySTDinterwordspacing

\bibitem{Daniel:2018:QCC:3177787.3148148}
\BIBentryALTinterwordspacing
F.~Daniel, P.~Kucherbaev, C.~Cappiello, B.~Benatallah, and M.~Allahbakhsh,
  ``Quality control in crowdsourcing: A survey of quality attributes,
  assessment techniques, and assurance actions,'' \emph{ACM Computing Surveys},
  vol.~51, no.~1, pp. 7:1--7:40, Jan. 2018. [Online]. Available:
  \url{http://doi.acm.org/10.1145/3148148}
\BIBentrySTDinterwordspacing

\bibitem{10256160}
\BIBentryALTinterwordspacing
J.~Huo, L.~Wang, X.~Wen, D.~Gesbert, and Z.~Lu, ``Cost-efficient vehicular
  crowdsensing based on implicit relation aware graph attention networks,''
  \emph{IEEE Transactions on Industrial Informatics}, vol.~20, no.~3, pp.
  3715--3725, 2024. [Online]. Available: \url{10.1109/TII.2023.3313649}
\BIBentrySTDinterwordspacing

\bibitem{8570744}
\BIBentryALTinterwordspacing
K.~Abualsaud, T.~M. Elfouly, T.~Khattab, E.~Yaacoub, L.~S. Ismail, M.~H. Ahmed,
  and M.~Guizani, ``A survey on mobile crowd-sensing and its applications in
  the {I}o{T} era,'' \emph{IEEE Access}, vol.~7, pp. 3855--3881, 2019.
  [Online]. Available: \url{10.1109/ACCESS.2018.2885918}
\BIBentrySTDinterwordspacing

\bibitem{BHATTI2020110611}
\BIBentryALTinterwordspacing
S.~S. Bhatti, X.~Gao, and G.~Chen, ``General framework, opportunities and
  challenges for crowdsourcing techniques: A comprehensive survey,''
  \emph{Journal of Systems and Software}, vol. 167, p. 110611, 2020. [Online].
  Available:
  \url{https://www.sciencedirect.com/science/article/pii/S0164121220300893}
\BIBentrySTDinterwordspacing

\bibitem{10.1145/3371425.3371459}
\BIBentryALTinterwordspacing
W.~Tan and Z.~Jiang, ``A novel experience-based incentive mechanism for mobile
  crowdsensing system,'' in \emph{Proceedings of the International Conference
  on Artificial Intelligence, Information Processing and Cloud Computing}, ser.
  AIIPCC '19.\hskip 1em plus 0.5em minus 0.4em\relax New York, NY, USA:
  Association for Computing Machinery, 2019. [Online]. Available:
  \url{https://doi.org/10.1145/3371425.3371459}
\BIBentrySTDinterwordspacing

\bibitem{NNis_Pre_2007}
N.~Nisan, T.~Roughgarden, E.~Tardos, and V.~V. Vazirani, \emph{Algorithmic Game
  Theory}.\hskip 1em plus 0.5em minus 0.4em\relax New York, NY, USA: Cambridge
  University Press, 2007.

\bibitem{Kleinberg:2005:AD:1051910}
J.~Kleinberg and E.~Tardos, \emph{Algorithm Design}.\hskip 1em plus 0.5em minus
  0.4em\relax Boston, MA, USA: Addison-Wesley Longman Publishing Co., Inc.,
  2005.

\bibitem{Singer:2014:BFM:2692359.2692366}
\BIBentryALTinterwordspacing
Y.~Singer, ``Budget feasible mechanism design,'' \emph{SIGecom Exch.}, vol.~12,
  no.~2, pp. 24--31, Nov. 2014. [Online]. Available:
  \url{https://doi.org/10.1145/2692359.2692366}
\BIBentrySTDinterwordspacing

\bibitem{Coreman_2009}
T.~H. Cormen, C.~E. Leiserson, R.~L. Rivest, and C.~Stein, \emph{Introduction
  to algorithms}.\hskip 1em plus 0.5em minus 0.4em\relax MIT press, 2009.

\bibitem{Som2023}
\BIBentryALTinterwordspacing
SomerStat,
  ``https://catalog.data.gov/dataset/participatory-budgeting-submissions,''
  2023. [Online]. Available: \url{data.somervillema.gov}
\BIBentrySTDinterwordspacing

\bibitem{Jalaly2017SimpleAE}
\BIBentryALTinterwordspacing
P.~Jalaly and {\'E}.~Tardos, ``Simple and efficient budget feasible mechanisms
  for monotone submodular valuations,'' \emph{ACM Transactions on Economics and
  Computation (TEAC)}, vol.~9, no.~1, pp. 1 -- 20, 2017. [Online]. Available:
  \url{https://api.semanticscholar.org/CorpusID:15747265}
\BIBentrySTDinterwordspacing

\bibitem{Chen3801}
\BIBentryALTinterwordspacing
N.~Chen, N.~Gravin, and P.~Lu, ``On the approximability of budget feasible
  mechanisms,'' in \emph{Proceedings of the Twenty-Second Annual ACM-SIAM
  Symposium on Discrete Algorithms}, ser. SODA '11.\hskip 1em plus 0.5em minus
  0.4em\relax USA: Society for Industrial and Applied Mathematics, 2011, p.
  685–699. [Online]. Available:
  \url{https://dl.acm.org/doi/10.5555/2133036.2133090}
\BIBentrySTDinterwordspacing

\end{thebibliography}

\vspace{11pt}

\begin{IEEEbiographynophoto}{Chattu Bhargavi}
 is a Ph.D. student in the School of Computer Science and Engineering at VIT-AP University, Amaravati, Andhra Pradesh, India. Her research interests include Algorithmic Game Theory and its Applications, Crowdsourcing, and Mobile crowsourcing. She can be reached at bhargavi.21phd7122@vitap.ac.in.
\end{IEEEbiographynophoto}

\vspace{11pt}

\begin{IEEEbiographynophoto}{Vikash Kumar Singh} is an Assistant Professor in the School of Computer Science and Engineering at VIT-AP University, Amaravati, Andhra Pradesh, India. He received his Ph.D. in Engineering and M.Tech in Computer Science and Engineering in 2019 and 2014, respectively from the Department of Computer Science and Engineering of National Institute of Technology (NIT) Durgapur, West Bengal, India. He received his B.Tech in Information Technology from Gautam Buddh Technical University, Lucknow, Uttar Pradesh in 2011. He has published papers in reputed conferences and journals such as IEEE HealthCom, IEEE Access, IEEE AINA, LNCS Transactions on Computational Collective Intelligence (Springer), etc. He is a reviewer for the journals like Internet of Things (Elsevier), Concurrency and Computation: Practice and Experience (Wiley), and International Journal for Grid and Utility Computing (Inderscience). His research interests include game theory, mechanism design, crowdsourcing, and the matching market.

\end{IEEEbiographynophoto}

\vfill

\end{document}